%% file: main.tex
\documentclass[11pt]{article}

\usepackage{hyperref}
\hypersetup{
    unicode=false,          % non-Latin characters in Acrobat’s bookmarks
    colorlinks=true,        % false: boxed links; true: colored links
    linkcolor=blue,          % color of internal links (change box color with linkbordercolor)
    citecolor=darkgreen,        % color of links to bibliography
    filecolor=magenta,      % color of file links
    urlcolor=cyan           % color of external links
}

\usepackage{xspace}
\usepackage{graphicx}
\usepackage{waingarten}
\usepackage[linesnumbered,ruled]{algorithm2e}
\usepackage{thmtools}
\usepackage{thm-restate}
\usepackage{color}
\usepackage{setspace}
\usepackage{aligned-overset}
\usepackage{tikz}
\usepackage{xparse}% So that we can have two optional parameters

\NewDocumentCommand\DownArrow{O{2.0ex} O{black}}{%
   \mathrel{\tikz[baseline] \draw [<-, line width=0.5pt, #2] (0,0) -- ++(0,#1);}
}
\NewDocumentCommand\UpArrow{O{2.0ex} O{black}}{%
   \mathrel{\tikz[baseline] \draw [->, line width=0.5pt, #2] (0,0) -- ++(0,#1);}
}
\def\FullBox{\hbox{\vrule width 8pt height 8pt depth 0pt}}
\newcommand{\QED}{\;\;\;\FullBox}
\renewenvironment{proof}{\noindent{{\textbf{Proof:}~}}} {\hfill\QED}
\providecommand{\email}[1]{\href{mailto:#1}{\nolinkurl{#1}\xspace}}

\newcommand{\eqdef}{\stackrel{\rm def}{=}}

\def\FullBox{\hbox{\vrule width 8pt height 8pt depth 0pt}}

\newcommand{\Elms}{\textsc{Elms}}
\newcommand{\bElms}{\textbf{\Elms}}
\newcommand{\Rep}{\textsc{Rep}}
\newcommand{\bRep}{\textbf{\Rep}}

\newcommand{\SampleTree}{\textsc{SampleTree}}
\newcommand{\quadtree}{\textsc{QuadTree}}

\newcommand{\Split}{\textsc{Split}}

\newcommand{\DataInd}{\textsc{Data-Ind}}

\newcommand{\Lone}{\textsc{Lone}}
\newcommand{\nbr}{\textsc{Nbr}}

\makeatletter
\newcommand{\Raj}[1]{{\color{mygreen}[Raj: #1]}}

\definecolor{goblue}{RGB}{0,128,255}

\def\be{\boldsymbol{e}}

\def\E{\mathbf{E}}

\def\T{\mathbf{T}}

\def\sfD{\mathsf{D}}
\def\sfh{\mathsf{h}}

\definecolor{b2}{RGB}{51,153,255}
\definecolor{mygreen}{RGB}{80,180,0}

\definecolor{darkgreen}{rgb}{0,0.5,0}

\newcommand{\CorePreprocess}{\textsc{Core-Preprocess}}
\newcommand{\CoreQuery}{\textsc{Core-Query}}

\newcommand{\EMD}{\mathsf{EMD}}
\newcommand{\ignore}[1]{}
\newcommand{\bhomg}{\boldsymbol{\hat{\Omega}}}
\newcommand{\Qind}{Q^{\textsc{Ind}}}
\newcommand{\Qdep}{Q^{\textsc{Dep}}}

\newcommand{\Dep}{\textsc{Dep}}
\title{Data-Dependent LSH for the Earth Mover's Distance}
\author {
  Rajesh Jayaram\thanks{Google Research.}
  \and
  Erik Waingarten\thanks{University of Pennsylvania.}
  \and
  Tian Zhang\thanks{University of Pennsylvania.}
}

\begin{document}
\maketitle

\begin{abstract}
We give new data-dependent locality sensitive hashing schemes (LSH) for the Earth Mover's Distance ($\mathsf{EMD}$), and as a result, improve the best approximation for nearest neighbor search under $\mathsf{EMD}$ by a quadratic factor. Here, the metric $\mathsf{EMD}_s(\mathbb{R}^d,\ell_p)$ consists of sets of $s$ vectors in $\R^d$, and for any two sets $x,y$ of $s$ vectors the distance $\mathsf{EMD}(x,y)$ is the minimum cost of a perfect matching between $x,y$, where the cost of matching two vectors is their $\ell_p$ distance. Previously, Andoni, Indyk, and Krauthgamer gave a (data-independent) locality-sensitive hashing scheme for $\mathsf{EMD}_s(\mathbb{R}^d,\ell_p)$ when $p \in [1,2]$ with approximation $O(\log^2 s)$. By being data-dependent, we improve the approximation to $\tilde{O}(\log s)$. 

Our main technical contribution is to show that for any distribution $\mu$ supported on the metric $\mathsf{EMD}_s(\mathbb{R}^d, \ell_p)$, there exists a data-dependent LSH for dense regions of $\mu$ which achieves approximation $\tilde{O}(\log s)$, and that the data-independent LSH actually achieves a $\tilde{O}(\log s)$-approximation outside of those dense regions. Finally, we show how to ``glue'' together these two hashing schemes without any additional loss in the approximation. 

Beyond nearest neighbor search, our data-dependent LSH also gives optimal (distributional) sketches for the Earth Mover's Distance. By known sketching lower bounds, this implies that our LSH is optimal (up to $\mathrm{poly}(\log \log s)$ factors) among those that collide close points with constant probability.
\end{abstract}
\thispagestyle{empty}
\newpage
\begin{spacing}{0.75}
\setcounter{tocdepth}{2}
\tableofcontents
\end{spacing}
\thispagestyle{empty}

\newpage

\pagenumbering{arabic}
\setcounter{page}{1}

\input{intro.tex}

\input{contributions}

\input{related-work}

\input{Prelims}

\input{data-dep-hashing.tex}

\input{DynamicEmbedding}

\input{towards-hash-fam.tex}
\input{data-ind-hash-and-SampleTree/st_main}

\input{data-dep-lb}

%\section*{Acknowledgments }

%\appendix

\input{data-dep-hash-to-nns-proof}
\input{extension-hypercube}

\bibliographystyle{alpha}
\bibliography{main}

\end{document}

%% file: intro.tex
%!TEX root = main.tex

\section{Introduction}
\label{sec:intro}

In the approximate nearest neighbor problem (ANN), we are given a set $P$ of $n$ points in a metric space $(X,d_X)$, and the goal is to build a data structure that, upon receiving a query point $q \in X$, can quickly return a point $p \in P$ such that $d(p,q) \leq c \cdot \min_{x \in X} d(q,x)$, for some approximation factor $c \geq 1$. The goal is to minimize $c$ while answering queries as fast as possible---ideally, significantly faster than a linear scan. Nearest neighbor search is a fundamental problem in computer science, with applications in areas such as machine learning, data mining, information retrieval, computer vision, and many others. In this paper, we study approximate nearest neighbor search for the Earth Mover's Distance ($\EMD$), also known as the Optimal Transport or Wasserstein-1 metric.% based on new data-dependent locality-sensitive hash functions for EMD. 

Let $(X, d_{X})$ be a ``ground metric'' (which, for us, will be $\R^d$ with the $\ell_p$-norm for $p \in [1, 2]$). Given two collections of $s$ elements from the ground metric, i.e., two multi-sets $x = \{ x_1,\dots, x_s \}, y = \{ y_1,\dots, y_s\} \subset X$ of size $s$, the Earth Mover's distance ($\EMD$) between $x$ and $y$ is
\[ \EMD(x, y) = \min_{\substack{\pi \colon [s]\to[s] \\ \text{bijection}}} \hspace{0.2cm} \sum_{i=1}^n d_X(x_i, y_{\pi(i)}). \]
We will write $\EMD_s(X, d_X)$ to denote the metric space of size-$s$ subsets of the $(X,d_X)$ under the Earth Mover's distance. 
Computational aspects of $\EMD$ have long been studied within the theoretical computer science literature~\cite{C02, IT03, I04b, AIK08, AIK09, ABIW09, HIM12,  SA12, mcgregor2013sketching, AS14, BI14, ANOY14,YO14, AKR15, S17, AFPVX17, KNP19, BDIRW20, CJLW22, ACRX22, CCRW23, FL23}. It is a central problem in algorithms, since it is a geometric version of bipartite matching. In addition, the Earth Mover's distance, and in particular nearest neighbor search under $\EMD$, has gained immense popularity in natural language processing and machine learning~\cite{KSKW15,arjovsky2017wasserstein, PC19,backurs2020scalable}, where it is a popular measure of distance between sets of embeddings (such as Word2Vec or GloVe~\cite{PSM14}). 

The canonical approach for approximate nearest neighbor search is to employ \textit{locality sensitive hashing} (LSH). These are randomized hash functions which partition the underlying metric space into hash buckets such that closer points are more likely to collide. An ANN data structure can then restrict its search to the hash buckets which the query maps to. 
By now, the theory of LSH for basic metrics like $\ell_1/\ell_2$ is well understood; % and algorithms exist that achieve quite good approximations: 
 the best $c$-approximations have query time $n^{1/c}$ for $\ell_1$, and query time $n^{1/(2c^2-1)+o(1)}$ for $\ell_2$~\cite{IM98, AI06,AINR14, AR15, ALRW17}, leading to highly sublinear $n^{\eps}$-time algorithms which achieve constant-factor (i.e., $1/\sqrt{\eps}$ or $1/\eps$) approximations.

Despite its popularity in theory and practice, LSH functions for $\EMD$ are not nearly as accurate as for $\ell_p$ spaces. This is because computing $\EMD$, unlike $\ell_p$, is significantly more computationally complex (for example, it does not decompose into a sum across coordinates).
Computing $\EMD$ exactly requires solving a min-cost bipartite matching problem, achieved classically by the Hungarian algorithm (in $O(s^3)$ time), and only recently in $O(s^{2+o(1)})$ time~\cite{CKLPPS22}. In addition, a simple heuristic like greedily generating a matching achieves a poor $\Omega(n^{0.58\dots})$ approximation~\cite{RT81}. 
This makes $\EMD$ difficult to reason about, and computations involving $\EMD$ especially challenging for sublinear algorithms which are limited in their computational abilities. The typical approach in sublinear algorithms is to \emph{embed} $\EMD$ into a ``simpler'' metric (usually $\ell_1$) and use LSH in the simpler metric.
Indyk~\cite{I04b} gave such an embedding of $\EMD$ into $\ell_1$ with distortion $O(d \log \Phi)$ (where $\Phi$ is the aspect ratio and should be read as $\poly(s)$, as there is a simple reduction to this case), leading to a $O(d \log s)$-approximation. This was later improved by~\cite{AIK08}, who gave a (randomized) embedding resulting in a LSH with approximation $O(\log s \log(d \Phi))$ (i.e., $O(\log^2 s)$ as it will also suffice to consider $d = \poly(s)$). However, despite significant and recent focus from the sublinear algorithms community \cite{backurs2020scalable,CJLW22,andoni2023sub,charikar2023fast,bakshi2023near,beretta2023approximate}, %\enote{removed citations to~\cite{BI14,doi:10.1137/15M1017958} since I would say these are significantly older and do not address the high-dimensional regime we do, but the other papers do.} 
to date no further improvements to the $O(\log^2 s)$-approximation of~\cite{AIK08} have been made. Our main result is a  nearly quadratic improvement in this approximation with the same runtime.

%Talk about \cite{backurs2020scalable}
%andoni2009efficient

\begin{theorem}[Main Result---Informal version of Theorem \ref{thm:ann-main}]\label{thm:main-intro-informal}
For any constant $\eps > 0$ and $p \in [1,2]$, there is a data structure for nearest neighbor search in $\EMD_s(\R^d,\ell_p)$, with approximation $\tilde{O}(\log s)$,  pre-processing time $n^{1+\eps} \cdot \poly(sd)$,  and query time $n^\eps \cdot \poly(sd)$.
\end{theorem}

%Theorem~\ref{thm:main-intro-informal} is %a quadratic improvement over the known approximation for nearest neighbor search over $\EMD$, 
%the first improvement since the embedding approach of~\cite{AIK08}. 
With regards to the runtime, note that in nearest neighbor search the primary goal is to have query time that significantly sublinear in $n$, which is the number of data points. In the context of EMD, the parameter $s$ (along with $d$) is the description size of a single point in the metric space; in fact, it takes $O(sd)$ time to simply read a query.  Thus, polynomial query time dependencies on $s,d$ are generally acceptable, however exponential dependency on $sd$ would be undesirable (see Remark \ref{remark:lowd}). 

The key component of Theorem~\ref{thm:main-intro-informal} is a new \textit{data-dependent} locality-sensitive hash family for $\EMD$, which, as we expand on next, is a relatively new algorithmic primitive for sublinear algorithms in geometric spaces~\cite{AINR14, AR15,ALRW17,ANNRW18, ANNRW18b}. We believe these data-dependent hash families are of independent interest, as they give rise to new and space optimal sketches for $\EMD$ in a distributional setting (see Section~\ref{sec:dd-lb}). Specifically, our LSH scheme gives a $\tilde{O}(\log s)$ approximation for this problem, nearly matching a $\Omega(\log s)$ lower bound 
of \cite{AIK08}. In particular, this implies a $\Omega(\log s)$-approximation lower bound for \emph{any} LSH family where close points collide with constant probability (Theorem~\ref{thm:dd-lb}), which is a property our LSH family satisfies.

%This demonstrates that $\Omega(\log s)$ is a natural barrier for LSH-based approaches to nearest neighbor search (Theorem~\ref{thm:dd-lb}).

\textbf{Data-Dependent (Locality-Sensitive) Hashing for $\EMD$.} As we further expand on in Section~\ref{sec:contributions}, the traditional guarantees of LSH are ``data-independent,'' or ``data-oblivious.'' In particular, LSH guarantees that, for \emph{any} pair of points $x,y$ from the metric, $x$ and $y$ tend to collide if they are close, and separate if they are far. One could imagine---and first successfully implemented in~\cite{AINR14}---that the hash function be specifically tailored to the dataset $P$, and that doing so would improve the approximation. In data-dependent LSH, the dataset is still arbitrary and worst-case; yet, by exploiting properties of an arbitrary dataset, one may improve on the best approximations. Put succinctly, we show that every dataset of $\EMD_s(\R^d, \ell_p)$ has special structure to exploit algorithmically which we cannot capture with known (data-independent) LSH.

%In this work, we show that this  phenomenon also arises: even though the best data-independent LSH for $\EMD$ gives approximation $O(\log^2 s)$ (from embedding to $\ell_1$~\cite{AIK08}), for every collection of size-$s$ subsets of $\R^d$, one may design a data-dependent LSH for that dataset with $\EMD$ obtaining approximation $\tilde{O}(\log s)$.

\textbf{Sketching for Sets of Vectors.} By now, there are various techniques for dealing with computationally
``simple'' objectives of high-dimensional vectors in sublinear regimes. For example, for $\ell_p$-norms we now have an essentially complete understanding of sketching (i.e., communication complexity), locality-sensitive hashing, and metric embeddings~\cite{KNW10, BJKS04, DIIM04, OWZ14, AKR15, AR15}. This work, as well as recent developments in geometric streaming~\cite{CJLW22,CJKVY22,CW22,CCJLW23,CJK23} and parallel algorithms~\cite{czumaj2023fully,jayaram2023massively}, aims to develop sketching techniques (which were initially designed for a single high-dimensional vector) to support objectives over entire collections of high-dimensional vectors. In particular, an important technical contribution of this paper is to generalize the probabilistic tree embeddings of~\cite{CJLW22} (which were designed for streaming algorithms) to obtain an improved data-dependent LSH family for nearest neighbor search. 
We believe that the LSH families developed in this paper are an important step towards closing the gap in our understanding between sketching for individual vectors and sketching for sets of vectors.

%% file: contributions.tex
\subsection{Overview of Contributions and Techniques}\label{sec:contributions}
We now overview the techniques involved in proving Theorem~\ref{thm:main-intro-informal}, and additionally state our formal results for data-dependent LSH (Theorem \ref{thm:LSH-main}) and nearest neighbor search (Theorem \ref{thm:ann-main}).  At a high level, this work can be seen within a progression of works, starting with~\cite{AIK08} and continuing with~\cite{backurs2020scalable,CJLW22}, on sketching for $\EMD$ via probabilistic tree embeddings. We aim to explain this progression, as it will highlight our main ideas (and the limitations of prior work).

\textbf{(Data-Independent) LSH for $\EMD$.}  An LSH for a metric space $(X,d_X)$ is a hash family $\calH$ which is so-called $(r, cr, p_1,p_2)$-sensitive. For a threshold $r \geq 0$, an approximation $c \geq 1$,  and $0 < p_2 < p_1 < 1$, the guarantees are:
\begin{enumerate}
    \item\label{en:close} \textbf{Close Points Collide}: $\prb{\bh \sim \calH}{\bh(x) = \bh(y)} \geq p_1$  for every $x, y \in X$ with $d_X(x,y) \leq r$. 
    \item\label{en:far} \textbf{Far Points Separate}: 
    $\prb{\bh \sim \calH}{\bh(x) = \bh(y)} \leq p_2$  for every $x, y \in X$ with $d_X(x,y) \geq cr$. 
    %For any two points $x, y \in X$  with distance at least $cr$, the probability over $\bh \sim \calH$ that $\bh(x) = \bh(y)$ is at most $p_2$.
\end{enumerate}
The seminal work of~\cite{IM98, HIM12} designed such LSH families for several metric spaces (like $(\R^d,\ell_p)$ for $p \in [1,2]$) and showed how to use them for $c$-approximate nearest neighbor with query time and space complexity governed by the gap between $p_1$ and $p_2$ (see Theorem~\ref{thm:hashing-to-nn}). Using~\cite{AIK08}, one may construct an LSH for $\EMD$ with an arbitrary threshold $r$, approximation $c = O(\log^2 s)$, and constant $0 < p_2 < p_1 < 1$ (resulting in a theorem like Theorem~\ref{thm:main-intro-informal}, although with approximation $O(\log^2 s)$).

\textbf{Probabilistic Tree Embeddings of~\cite{AIK08}.} The (data-independent) LSH for $\EMD$ crucially relies on an embedding from $\EMD_s(\R^d, \ell_1)$ into a randomized tree metric---this is known as a probablistic tree embedding.\footnote{\label{foot:embeddings} For applications in sublinear algorithms such as ours, it is important that the embeddings themselves can be efficiently stored and efficiently evaluated. Thus, the classical works on probabilistic tree embeddings~\cite{B98, FRT04} are not applicable. See Remark~\ref{rem:prob-trees}.}  
%\cite{AIK08} showed how to construct a (data-independent) probabilistic tree embedding for arbitrary subsets of $m$ vectors in $(\R^d, \ell_p)$ with expected distortion $O(\log m \log(d\Phi))$. 
Specifically, \cite{AIK08} define a distribution supported over (weighted) trees $\bT$, as well as a mapping $\psi\colon\R^d \to \bT$ to leaves of the tree $\bT$, such that for any subset $\Omega \subset \R^d$ of at most $m$ vectors, \textbf{(i)} the embedding $\psi$ is non-contracting on $\Omega$ with high probability, i.e., $d_{\bT}(\psi(a), \psi(b)) \geq \|a - b\|_p$ for every $a, b \in \Omega$,\footnote{Note that $d_{\bT}(\cdot,\cdot)$ is the length of the path in $\bT$.} and \textbf{(ii)} the expectation of $d_{\bT}(\psi(a), \psi(b))$ is at most $O(\log m \log(d \Phi)) \cdot \| a - b\|_p$. As mentioned, there is a simple reduction to always consider the aspect ratio $\Phi$ and dimensionality $d$ to be $\poly(s)$ (see Lemma~\ref{lem:reduction-to-hypercube}), so this becomes a $O(\log m\log s)$ expected distortion. 

By applying $\psi$ to each vector in a set $x \in \EMD_s(\R^d,\ell_p)$, the mapping $\psi$ naturally induces a metric embedding of $\EMD_s(\R^d, \ell_p)$ into $\EMD_s(\bT, d_{\bT})$.  Applying the guarantees \textbf{(i)} and \textbf{(ii)} above to the set of vectors $\Omega = x \cup y$, where $x,y \in \EMD_s(\R^d,\ell_p)$, one can show that the embedding $\psi$ is non-contracting with high probability and  satisfies that for any $x,y \in \EMD_s(\R^d,\ell_p)$: 
\[\Ex_{\bT}\left[\EMD_{\bT}(\psi(x),\psi(y)) \right] \leq O(\log^2 s) \cdot \EMD(x,y).\] 
The reason for embedding $\EMD_s(\R^d,\ell_p)$ into $ \EMD_s(\bT,d_{\bT})$ is that $\EMD$ over tree-metrics is a much simpler metric. In particular, the greedy algorithm is optimal for $\EMD$ over trees, and as a consequence there is a folklore \emph{isometric} embedding of $\EMD_s(\bT,d_{\bT})$ into $\ell_1$~\cite{C02, I04b} (see Fact~\ref{fact:ell-1-embed}), thereby embedding a \textit{set} of vectors in a tree into a single vector in $\ell_1$. Finally, after applying this embedding into $\ell_1$, one can apply the classic LSH functions for $\ell_1$~\cite{IM98} (denoted as $\bphi$ below) to obtain a LSH function for $\EMD_s(\R^d,\ell_p)$. This process is shown in the diagram below, where the names of the embeddings are shown on top of the arrows, and the distortion of those embeddings is shown below:
%Given the mapping from $\EMD_s(\R^d,\ell_p)$ into tree metric, the LSH for $\EMD$ then proceeds by the following composition shown below:
%\Raj{Add the mapping $\R^d \mapsto \bT$ to figure below}%\mathop{\longrightarrow}^{\tiny \text{(A)}}
   \begin{align}
\R^d \quad  \quad & \xrightarrow[]{a  \mapsto \psi(a)}  \quad \quad  \bT \nonumber \\ 
 & \quad \; \; \;  \DownArrow[0.85 cm][>=stealth,black, thick, dashed] \nonumber \\ 
\EMD(\R^d, \ell_p) \quad & \xrightarrow[O(\log^2 s)]{\; \; x  \mapsto \psi(x) \; \; }  \quad \EMD(\bT, d_{\bT}) \quad \xrightarrow[\text{isometric}]{\quad  \text{folklore} \quad}  \quad \ell_1 \quad \xrightarrow[]{\quad\bphi\quad}  \quad \{\text{hash buckets}\} \label{eq:hash-transform-intro-1}
\end{align}
Since the second mapping is isometric, the distortion of the entire embedding into $\ell_1$ is $O(\log^2 s)$, thus the resulting LSH for $\EMD_s(\R^d,\ell_p)$ is a $O(\log^2 s)$ factor larger than the distortion incurred by the LSH $\bphi$ for $\ell_1$.

%Here, we first sample $\bT$ from the probabilistic tree construction, as well as an classic LSH functions $\varphi$ for $\ell_1$~\cite{IM98}. The analysis then proceeds by considering an arbitrary $x,y \in \EMD_s(\R^d,\ell_p)$ and checking both (\ref{en:close}) and (\ref{en:far}). Arrow (A) is the main component, where we consider the subset $S = x \cup y$, which contains at most $2s$ vectors. Since $\bT$ is data-independent (and oblivious to $S$), we obtain properties \textbf{(i)} and \textbf{(ii)} on the embedding of $S$ to $\bT$. The high probability guarantee of \textbf{(i)} implies that no distances in $\bT$ contract, so that $\EMD_{\bT}(\psi(x), \psi(y)) \geq \EMD(x, y)$; furthermore, we may apply Markov's inequality to \textbf{(ii)} and upper bound $\EMD_{\bT}(\psi(x), \psi(y)) \leq O(\log^2 s) \cdot \EMD(x,y)$---here, $\psi(x)$ (and $\psi(y)$) denote the subset obtained by applying $\psi$ to each vector in $x$ (and $y$). Arrow (B) is a (folklore) isometric embedding from $\EMD(\bT, d_{\bT})$ to $\ell_1$~\cite{C02, I04b} (see Fact~\ref{fact:ell-1-embed}), and arrow (C) is the LSH family for $\ell_1$~\cite{IM98}.

\paragraph{Data-Dependent Probabilistic Tree Embeddings~\cite{CJLW22}.}
Recently,~\cite{CJLW22} improved the probabilistic tree embedding of~\cite{AIK08} by being data-dependent. They show that, for an arbitrary subset $\Omega$ of $m$ vectors in $(\R^d, \ell_p)$, there exists a probabilistic tree embedding $\psi_{\Omega}: (\Omega,\ell_p) \to (\bT_{\Omega}, d_{\bT_{\Omega}})$ which depends on $\Omega$, and that embeds $\Omega$ obtaining guarantees \textbf{(i)} and \textbf{(ii)} above as achieved by~\cite{AIK08}, except with an expected distortion of $\tilde{O}(\log(ms))$, where $m=|\Omega|$ (see Lemma~\ref{lem:cjlw} and Appendix~\ref{sec:dynamic-l1}).\footnote{Similarly to Footnote~\ref{foot:embeddings}, it is especially important that the embeddings be efficiently stored and evaluated. See Remark~\ref{rem:prob-trees}.} In order to compute $\EMD(x, y)$ given any $x, y \in \EMD_s(\R^d,\ell_p)$, the analogous diagram to above first considers the subset of vectors $\Omega = x \cup y$, generates $(\bT_{\Omega}, d_{\bT_{\Omega}})$, and proceeds by
\begin{align} 
\EMD(\Omega, \ell_p) \qquad\xrightarrow[\tilde{O}(\log (ms))]{x \mapsto \psi_{\Omega}(x) } \qquad \EMD(\bT_{\Omega}, d_{\bT_{\Omega}}) \qquad \xrightarrow[\text{isometric}]{\quad \text{folklore} \quad} \qquad \ell_1. \label{eq:transform-intro-2} 
\end{align}
%Arrow (A) now applies the (improved) data-dependent probabilistic tree embedding $(\bT_S,d_{\bT_S})$, and the second arrow is the (same folklore) isometric embedding into $\ell_1$ from above.
An important point here is that the tree embedding into $(\bT_{\Omega}, d_{\bT_{\Omega}})$ depends on the set of vectors in $\Omega$. This means that, if we wanted to use the above embedding for nearest neighbor search, then even if we used an LSH for $\ell_1$ (e.g. the mapping  $\bphi$ above) to map to hash buckets, the resulting hash family would be for points in $\EMD(\Omega, \ell_p)$, and it is not at all clear what the set $\Omega$ should be. In fact, there are two immediate challenges here:
\begin{itemize}
    \item \textbf{Challenge 1}: In nearest neighbor search, the input is an arbitrary dataset $x_1,\dots, x_n \in \EMD_s(\R^d, \ell_p)$, where each $x_i$ is a subset of $(\R^d, \ell_p)$ of $s$ vectors. The natural choice is $\Omega = \bigcup_{i=1}^n x_i$. The resulting (data-dependent) probabilistic tree $(\bT_{\Omega}, d_{\bT_{\Omega}})$, and composition of the maps (with an LSH for $\ell_1$), would give an LSH family for $\EMD(\Omega, \ell_p)$. By construction, each $x_1,\dots, x_n$ is inside $\EMD(\Omega, \ell_p)$, so dataset vectors can be hashed. However, the approximation increases to $\tilde{O}(\log (ns))$, which is far from the claimed $\tilde{O}(\log s)$-bound, and may be strictly worse than the $O(\log^2 s)$ approximation of~\cite{AIK08}.
    \item \textbf{Challenge 2}: Even if we set $\Omega$ to all vectors used by $x_1,\dots, x_n$, a crucial component of LSH involves applying the hash functions to the (unknown) query point. In particular, the data structure will hash the dataset during preprocessing, and in the future, a query comes (which was unknown during preprocessing) and needs to be hashed as well. 
\end{itemize}

\textbf{Warm-Up: Overcoming Challenge 2.} We first show, as a warm-up and independent contribution, that the second challenge can be overcome by making~\cite{CJLW22} dynamic (Theorem~\ref{thm:dynamic-main} below, there is a reduction to $d, \Phi$ being $\poly(s)$). The data structure sets $\Omega = \bigcup_{i=1}^n x_i$, generates a tree embedding $(\bT_{\Omega}, d_{\bT_{\Omega}})$, and constructs a hash function to the dataset $x_1,\dots, x_n$. Then, whenever a query point $y \in \EMD(\R^d,\ell_p)$ comes, we first \emph{update} the tree to $(\bT_{\Omega \cup y}, d_{\bT_{\Omega \cup y}})$ (and corresponding hash functions) and identify the (few) dataset points $x_i$ whose hash value changes. This allows the algorithm to maintain a view consistent with having preprocessed the dataset with the tree $(\bT_{\Omega \cup y}, d_{\bT_{\Omega\cup y}})$.

%Specifically, we prove the following:
\begin{theorem}[Dynamic and Data-Dependent Probabilistic Tree Embedding] \label{thm:warm-up}
    For a fixed $d \in \N$ and $p \in [1,2]$, there is a data structure that maintains maintains a set $\Omega \subset [\Delta]^d$ of $m$ vectors and an non-contracting embedding $\varphi: (\Omega,\ell_p) \to \bT$, with expected distortion $\tilde{O}(\log (m d \Delta))$ for any pair $x,y \in \Omega$.  Moreover, it supports the following operations in expected time $O(d \log (d \Delta))$
    \begin{itemize}
        \item \emph{\textbf{Query}}: Given a vector $x \in \Omega$, return the weighted path from the root of $\bT$ to $\varphi(x)$
        \item \emph{\textbf{Insertions/Deletions}}: Add or remove vectors from the set $\Omega$, and also return the updated weighted paths of every vector $v \in \Omega$ whose path weights changed from the insertion/deletion.
        \end{itemize}
\end{theorem}

\paragraph{Tree Construction and Proof of Theorem~\ref{thm:warm-up}.}\label{sec:tree-construction} Given~\cite{CJLW22}, the proof of Theorem~\ref{thm:warm-up} is very intuitive.
We first consider the case of embedding $\EMD$ over the hypercube $(\{0,1\}^d, \ell_1)$ (in Section~\ref{sec:hamming-cube-dynamic} and then extend to $(\R^d,\ell_p)$ in Section~\ref{sec:dynamic-l1}). The construction, in Figure~\ref{fig:DDquadtree-prelims}, builds the probabilistic tree $\bT$ of depth $O(\log(d))$ where each level $\ell$ samples $2^{\ell}$ random coordinates; each node $v$ at depth $\ell$ has $2^{2^{\ell}}$ child nodes $v_u$, one for each possible setting $u \in \{0,1\}^{2^{\ell}}$ of the $2^{\ell}$ sampled coordinates, and, this defines a natural mapping of $\{0,1\}^d$ to leaves to $\bT$~(Definition~\ref{def:quadtree-map}).\footnote{The above is a hypercube version of the ``randomly shifted grid,'' called ``quadtree'' in~\cite{backurs2020scalable}.} Moreover, for any vertex $v \in \bT$, we can define the set  $\Elms(v,\Omega) \subset \Omega$ to be the set of vectors $a \in \Omega$ whose root-to-leaf path (after the mapping $\psi_{\Omega}$) goes through $v$ (in Figure~\ref{fig:DDquadtree-prelims}, $\Elms(v,\Omega)$ corresponds to $\Elms(v) \cap \Omega$).

The data-dependent part of \cite{CJLW22} is how the edge weights are set. Specifically, the data-independent embedding of~\cite{AIK08} sets the weights at depth $\ell$ of $\bT$ to be proportional to $d / 2^{\ell}$, since vectors $a, b \in \{0,1\}^d$ at distance $d/2^{\ell}$ are first separated at depth $\ell$ with constant probability. In contrast, for a subset $S$, in~\cite{CJLW22}, an edge $(v, v_u)$ is defined by the distance from a random sampled vector $\bc \sim \Elms(v,\Omega)$ to a randomly sampled vector $\bc' \sim \Elms(v_u,\Omega)$ (Lemma~\ref{lem:cjlw-mod}). This data-dependent setting of the weights improves the expected distortion to $\tilde{O}(\log(md))$. Theorem \ref{thm:warm-up}  shows that this embedding can maintained dynamically. For example, whenever there is an insertion of $a \in \{0,1\}^d$ to $S$, we can find the root-to-leaf path of $a$ in $\bT$, and for each vertex $v$ on the path, we must update the draw $\bc$ for $v$ such that it remains uniform (now over $\Elms(v,\Omega) \cup \{a \}$). We do this by setting $\bc$ to $a$ with probability is $1/|\Elms(v,\Omega) \cup \{a\}|$, and leave otherwise (Claim~\ref{cl:remain-uni}). If we do change the sample, then we must update the weight of each edge incident to $v$, and therefore must update the embeddings of every $b \in \Elms(v)$. Thus, the expected number of embeddings that must be updated is constant, allowing for small expected update time (see Section~\ref{en:ds-dynamic}). 

\begin{remark}[Using Classical Probabilistic Tree Embeddings]\label{rem:prob-trees}
In sublinear algorithms, an embedding $f$ mapping $(X, d_X)$ to $(Y, d_Y)$ must have a succinct description and admit efficient evaluations of $f(x)$ (where time should be polynomial, or near-linear, in the description of $x$). General theorems for probabilistic tree embeddings, like~\cite{B98,FRT04}, obtain expected distortion $O(\log m)$ for any size-$m$ metric, but do not have efficient evaluations so cannot be used in sublinear settings. In this work, all embeddings can be evaluated $f(x)$ in time polynomial in the description of $x$.
\end{remark}

\paragraph{An Improved Data-Dependent LSH for $\EMD$.} The above dynamic embedding still suffers a $\tilde{O}(\log(ns))$ distortion, and does not address \textbf{Challenge 1}. We now proceed with the main technical component, of designing a data-dependent LSH for $\EMD_s(\R^d,\ell_p)$. It turns out that for nearest neighbor search, it suffices to tailor (and relax) the second condition of LSH to an arbitrary fixed distribution (see Definition~\ref{def:data-dep} and Theorem~\ref{thm:hashing-to-nn} for how data-dependent hashing implies nearest neighbor search). Specifically, a hash family $\calH$ is $(r, cr,p_1,p_2)$-sensitive for a distribution $\mu$ supported on a metric $(X, d_X)$ whenever:
\begin{enumerate}
    \item\label{en:close-dd} \textbf{Close Points Collide}: $\Prx_{\bh \sim \calH}[\bh(x) = \bh(y)] \geq p_1$ for every $x, y \in X$ with $d_X(x,y) \leq r$.
    \item\label{en:far-dd} \textbf{Far Points Separate on Average}: For any $x \in X$, the probability over $\bh \sim \calH$ and $\by \sim \mu$ that $\bh(x) = \bh(\by)$ and $d_X(x, \by) \geq cr$ is at most $p_2$.
\end{enumerate}
The only difference is the second condition (\ref{en:far-dd}) above, where one considers any $x \in X$ and ensures that a sampled point $\by \sim \mu$ far from $x$ collides with probability at most $p_2$ (see Section~\ref{def:ann-from-dd}, for comparison with~\cite{AR15}). Roughly speaking, even though $\mu$ is arbitrary, $\calH$ ``knows'' $\mu$, and can cater to particular properties of $\mu$. Our main technical result is designing a data-dependent LSH for $\EMD_s(\R^d, \ell_p)$ which satisfies the conditions above for approximation $\tilde{O}(\log s)$ with $0 < p_2 < p_1 < 1$. In particular, we prove the following theorem, which by a reduction from approximate near neighbors to data-dependent LSH (Theorem~\ref{thm:hashing-to-nn}) implies Theorem~\ref{thm:main-intro-informal} by setting $p_2$ to $1/10$ and $p_1 = 1-\eps$.

\begin{theorem}[Data-Dependent Hashing for $\EMD$ (Theorem \ref{thm:data-dep-hashing} + Lemma \ref{lem:reduction-to-hypercube})]\label{thm:LSH-main}
For any $s, d \in \N$, $p \in [1,2]$, a threshold $r > 0$, and any $0 < p_2 < p_1 < 1$, there exists a data structure with the following guarantees:
\begin{itemize}
    \item \emph{\textbf{Preprocessing}}: The data structure receives sample access to a distribution $\mu$ supported on $\EMD_s(\R^d,\ell_p)$, and in time $\poly(sd/((1-p_1)p_2))$, initializes a draw $\bh$ from a hash family $\calD$ (which depends on $\mu$) and is $(r, cr, p_1, p_2)$-sensitive for $\mu$ (see Definition~\ref{def:data-dep}), with
    \[ c = \tilde{O}\left(\log s \cdot \dfrac{\log^2(1/p_2)}{1-p_1} \right).\]
    \item \emph{\textbf{Query}}: Given any $q \in \EMD_s(\R^d, \ell_p)$, the data structure computes $\bh(q)$ in time $\poly(sd)$.
    %\item \emph{\textbf{Insertions/Deletions}}: When the distribution $\mu$ is a uniform distribution over a dataset $P$ (which occurs in nearest neighbor search), in expected time $\tilde{O}(sd)$, we may add or remove points from $P$, and output the updated values $\bh(p)$ for every $p \in P$ which changes.
\end{itemize}
\end{theorem}

The above is our main technical theorem, and most of the work is devoted to that proof. Similarly to before, it will suffice via a simple reduction, to consider $\EMD$ over the hypercube $\{0,1\}^d$ with $\ell_1$ distance, where $d \leq \poly(s)$ and the threshold $r =\omega(s)$ (see Lemma~\ref{lem:reduction-to-hypercube}). Then, the construction of the data-dependent hashing scheme from Theorem~\ref{thm:LSH-main} can be split into three parts, which we now describe. 
 
\textbf{Step 1: The $\SampleTree$ Embedding.} %The key barrier in improving the $\tilde{O}(\log (ns))$ approximation from Theorem \ref{thm:warm-up} is that the distortion of the data-dependent probabilistic tree embeddings $(\bT_{\Omega}, d_{\bT_{\Omega}})$ and $\psi_{\Omega}$ of~\cite{CJLW22} degrade as the number of points $|\Omega|$ it depends on grows. 
Since we aim for a $\tilde{O}(\log s)$-approximation, we will use the data-dependent probabilistic trees of~\cite{CJLW22} on the union $\bOmega$ of a small number of $m = \poly(s)$ samples $\by_1,\dots, \by_m \sim \mu$ (i.e., $\bOmega = \bigcup_{i=1}^m \by_i$, which is a subset of $s\cdot m$ vectors in $\{0,1\}^d$, in boldface $\bOmega$ since it is random). Composing the data-dependent probabilistic tree $\bT_{\bOmega}$ (which we will refer to as $\bT$) with the isometric embedding defines an embedding of $\EMD_s(\bOmega)$ into $\ell_1$, we aim to \emph{extend} the embedding to the entire space $\EMD_s(\{0,1\}^d)$:\footnote{In both cases, $\EMD_s(\bOmega)$ and $\EMD_{s}(\{0,1\}^d)$ refers to $\EMD_s(\bOmega, \ell_1)$ and $\EMD_s(\{0,1\}^d,\ell_1)$, respectively. Furthermore, the map $\psi \colon \R^d \to \bT$ is implicit in the notation, so we write $d_{\bT}(a,b)$ for $d_{\bT}(\psi(a), \psi(b))$ and $\EMD_{\bT}(x,y)$ for $\EMD_{\bT}(\psi(x), \psi(y))$.}
\begin{align}
&\EMD_s(\bOmega) \quad  \xrightarrow[\tilde{O}(\log s)]{\tiny \cite{CJLW22}}  \quad \EMD_s(\bT, d_{\bT}) \quad \xrightarrow[\text{isometric}]{\quad  \text{folklore} \quad} & \hspace{-2cm}\ell_1 \label{eq:transform-before}\\ 
&\qquad  \UpArrow[0.5 cm][>=stealth,black, thick] & \hspace{-2cm}\UpArrow[0.5 cm][>=stealth,black, thick]\hspace{.2cm}  \nonumber  \\ 
&\EMD_s(\{0,1\}^d) \quad \xrightarrow[\qquad\qquad\qquad\qquad\qquad\qquad\qquad\qquad]{\text{desired new map in Section~\ref{sec:sample-tree-def}}} &  \hspace{-2cm} \ell_1 \label{eq:extension-intro}
\end{align}
In the above diagram, (\ref{eq:transform-before}) has expected $\tilde{O}(\log s)$-distortion from $\EMD_s(\bOmega)$ to $\ell_1$ from~\cite{CJLW22} on the samples $\by_1,\dots ,\by_m \sim \mu$. We then define the extension (\ref{eq:extension-intro}) of (\ref{eq:transform-before}), which is a natural ``hybrid'' of~\cite{CJLW22} and~\cite{AIK08}, that we call $\SampleTree(\mu,m)$ in Section~\ref{sec:sample-tree-def}. In particular, $\SampleTree(\mu, m)$ is defined similarly to the tree construction in Theorem~\ref{thm:warm-up} but with the following combination of edge weights:
\begin{itemize}
    \item \textbf{Data-Dependent Weights}: We let $(\bT, d_{\bT})$ be the data-dependent probabilistic tree embedding of~\cite{CJLW22} on $\bOmega$ which defines the edge weights $(v, v_u)$ for a node $v$ at level $\ell$ in the ``data-dependent'' fashion when $v_u$ contains vectors from $\bOmega$ (recall, the average distance of vectors sampled from $\Elms(\cdot)$ in $v$ and $v_u$). Note that we will modify the set $\bOmega$ very slightly later on (see definition of $\hat{\bOmega}$ in the subsequent discussions). 
    \item\textbf{Data-Independent Weights}: Suppose, on the other hand, that $(v, v_u)$ is an edge with $v$ at depth $\ell$, such that $v_u$ does not contain any vectors from $\bOmega$, we set the weight of $(v, v_u)$ according to~\cite{AIK08}, to $\xi \cdot d / 2^{\ell}$ (for a parameter $\xi = \tilde{O}(\log s)$). 
\end{itemize}
%From that, we consider the natural extension (\ref{eq:extension-intro}) obtained by a hybrid of~\cite{CJLW22} and~\cite{AIK08}. The extension inherits the expected $\tilde{O}(\log s)$-distortion on $\EMD_s(\bOmega, \ell_p)$, and our task will be to upper bound the degradation in the distortion as we move further from $\EMD_s(\Omega,\ell_p)$.
%Note that the resulting map defines embedding $\psi_S : S \to \bT$ is only defined on points $x \in S$---to extend $\psi_S$ to a mapping on all of $(\R^d,\ell_p)$ we must define the weights of edges $(v,v_u)$ in the tree $\bT$ for which $\Elms(v) = \emptyset$. For such edges, our approach is simple: just give $(v,v_u)$ the weight it would have had in the data-independent tree embedding of \cite{AIK09}; namely, if $v$ is at depth $\ell$, we give it a fixed weight proportional to $d/2^\ell$. The resulting tree embedding extends $\psi_S$ from all of $\R^d$ to $\bT$, and we refer to it as $\SampleTree$, with the corresponding embedding denoted $\overline{\psi}_S: (\R^d,\ell_p) \to \bT$, since we will ultimately apply it to a set $S \subset \EMD_s(\R^d,\ell_p)$ which is uniformly sampled from the dataset $P$. Note that the edge weights of $\SampleTree$ are a hybridization of data-independent and dependent weight --- only the edges which a point in $S$ passes through are given a data-dependent weight. 
%\textbf{Step 1b: The Guarantee of $\SampleTree$ on Locally-Dense Points.} 
With both data-dependent and data-independent weights, we obtain a tree metric $(\bT, d_{\bT})$, and an embedding of the entire space $\EMD_s(\{0,1\}^d)$ to $\EMD_s(\bT)$. It is not too difficult to show that the extension (\ref{eq:extension-intro}) given by $\SampleTree(\mu,m)$ is non-contracting with high-probability (Lemma~\ref{lem:sample-tree-contr}), roughly speaking, because both the data-independent and data-dependent probabilistic trees are non-contracting with high probability. 

The more subtle argument, however, is upper bounding the expansion. On the one hand, suppose $x, y \in \EMD_s(\{0,1\}^d)$ are two arbitrary points, and all vectors in $x \cup y$ happened to be in $\bOmega$, then (\ref{eq:extension-intro}) inherits the $\tilde{O}(\log s)$ expected distortion from $(\bT, d_{\bT})$. On the other hand, if all vectors of $x \cup y$ are very far from $\bOmega$, then the root-to-leaf paths of vectors in $x$ and $y$ in $\bT$ are mostly disjoint from those of $\bOmega$. This means $\EMD_{\bT}(x, y)$ is effectively always using the data-independent weights, and similarly to the analysis of~\cite{AIK08}, incurs distortion $O(\log^2 s)$. In summary, the distortion of $\SampleTree(\mu,m)$ on a pair of points $x, y \in \EMD_s(\{0,1\}^d)$ depends on how well the sample $\bOmega$ ``represents'' the two points $x,y$ geometrically. One needs a suitable notion of how the set $\bOmega$ ``represents'' a region of $\EMD_s(\{0,1\}^d)$; then, we can partition $\mu$ into the region represented by $\bOmega$ that which the data-dependent probabilistic tree obtains approximation $\tilde{O}(\log s)$. %and a quantitative bound on how a set $\bOmega$ represents a pair of points $x, y \in \EMD_s(\R^d,\ell_p)$. 

\newcommand{\Chamfer}{\textsf{Chamfer}}
\paragraph{Step 2(a): Extensions on Chamfer Neighborhoods.} Our notion of representation in $\EMD_s(\{0,1\}^d)$ will consider the \textit{Chamfer Distance}, which is an (assymmetric) measure capturing dissimilarity of subsets in $\R^d$. Formally, given two subsets of vectors $x, z$ in $\{0,1\}^d$, we use the Chamfer distance from $x$ to $z$ in $\{0,1\}^d$ with $\ell_1$ distance,
\[ \Chamfer(x,z) = \sum_{a \in x} \min_{b \in z}\|a - b\|_1. \]
Chamfer lower bounds $\EMD(\cdot,\cdot)$, since it relaxes the bijection condition $\pi\colon x \to z$, and is much simpler to reason about. In the context of the extension (\ref{eq:extension-intro}), it captures, for any point $x \in \EMD_s(\{0,1\}^d)$, how far $x$ is from $\bOmega$ (and from the data-dependent edge weights in $\SampleTree(\mu,m)$). A naive argument proceeds as follows: consider $x, y \in \EMD_s(\{0,1\}^d)$, let $\sigma \colon x\cup y \to \bOmega$ be the nearest-neighbor map realizing $\Chamfer(x\cup y, \bOmega)$, and let $\sigma(x), \sigma(y)$ be the subsets of $\bOmega$ obtained by applying $\sigma$ to each vector in $x,y$. First, \textbf{(i)} the expected $\EMD_{\bT}(\sigma(x), \sigma(y))$ is at most $\tilde{O}(\log s) \cdot \EMD(\sigma(x), \sigma(y))$ by the data-dependent edge weight analysis; second, \textbf{(ii)} $(\bT, d_{\bT})$ achieves $O(\log^2 s)$ expected distortion on $x \cup y \cup \sigma(x) \cup \sigma(y)$ by the data-independent edge weight analysis. Thus, for a fixed sample $\bOmega$, the triangle inequality would result in the upper bound:
\begin{align}
    \Ex_{\bT}\left[ \EMD_{\bT}(x, y)\right] %&\leq \Ex_{\bT_{\bOmega}}\left[ \EMD_{\bT_{\bOmega}}(x, \pi(x)) + \EMD_{\bT_{\bOmega}}(\pi(x), \pi(y)) + \EMD_{\bT_{\bOmega}}(\pi(y), y)\right]\\
        %&\leq O(\log^2 s) \cdot \Chamfer(x \cup y, \bOmega) + \tilde{O}(\log s) \cdot \EMD(\pi(x), \pi(y)) \\
        &\leq O(\log^2 s) \cdot \Chamfer(x \cup y, \bOmega) + \tilde{O}(\log s) \cdot \EMD(x, y).\label{eq:naive-chamfer}
\end{align}
By (\ref{eq:naive-chamfer}), for any pair of points $x, y \in \EMD_s(\{0,1\}^d)$ in a $\Theta(\EMD(x,y)/\log s)$ Chamfer neighborhood of $\bOmega$, the $\SampleTree$ embedding will give a $\tilde{O}(\log s)$ distortion to their distance. In other words, given the threshold $r$ (from the definition of data-dependent hashing), we can consider a Chamfer neighborhood of size $\Theta(r/\log s)$ around $\bOmega$.
Unfortunately, this Chamfer neighborhood will not be sufficiently large, as it is easy to consider natural datasets where all pairs $\by,\by' \sim \mu$ have Chamfer distance $r$ from each other, in which case this neighborhood would be empty, and (\ref{eq:naive-chamfer}) would only give a $O(\log^2 s)$ approximation. Thus, we will need to give a significantly improved bound 
than (\ref{eq:naive-chamfer}), to obtain a $\tilde{O}(\log s)$ approximation.
%in order to achieve a improved $\tilde{O}(\log s)$ approximation. % on pairs $(x,y)$ in a larger $\polylog(s) \cdot \EMD(x,y)$ Chamfer neighborhood of the sample $\bOmega$. Specifically:

%

\begin{quote}
\textbf{Key Idea 1}: We demonstrate that, with a last modification to $\SampleTree(\mu,m)$, all points $x,y$ in a Chamfer neighborhoods of radius $\EMD(x,y) \cdot \poly(\log s)$ (for arbitrary constant power) around $\bOmega$ still maintain a $\tilde{O}(\log s)$ expected distortion, and this will suffice for the remainder of the argument. 
\end{quote}
Specifically, in Lemma~\ref{lem:sample-tree-expan} (using Lemma~\ref{lem:dense-conse}), we argue that in $\SampleTree(\mu, m)$, if in addition to taking $m$ samples $\by_1,\dots, \by_m \sim \mu$ and letting $\bOmega = \bigcup_{i=1}^m \by_i$, we let 
\[ \hat{\bOmega} = \nbr(\bOmega) = \left\{ b' \in \{0,1\}^d : \exists b \in \bOmega, \|b - b'\|_1 \leq 1 \right\}, \]
where $|\hat{\bOmega}| \leq \poly(s)$ (recall $m$ is $\poly(\log s)$ and $d$ is $\poly(s)$), and define data-dependent weights with respect to $\hat{\bOmega}$, then we have the improved version of (\ref{eq:naive-chamfer}): 
\begin{align} 
\Ex_{\bT}\left[ \EMD_{\bT}(x,y)\right] \leq \tilde{O}(\log s) \cdot \EMD(x, y) \left( 1 + \log\left(\frac{\Chamfer(x, \bOmega)}{\EMD(x,y)} + 1 \right) \right). \label{eq:better-chamfer}
\end{align}
Before overviewing the proof of (\ref{eq:better-chamfer}), we note how it leads to the extension we desire:
\begin{itemize}
    \item We call a point $x \in \EMD_s(\{0,1\}^d)$ ``locally-dense'' with respect to $\mu$ if on a random sample $\by_1,\dots, \by_m \sim \mu$, letting $\bOmega = \bigcup_{i=1}^m \by_i$ satisfies $\Chamfer(x, \bOmega) \leq r \cdot \poly(\log(s))$ in expectation. (The above is the important consequence of the locally-dense in Definition~\ref{def:ingred-2}, see Lemma~\ref{lem:dense-conse}).
    \item Then, if $x, y \in \EMD_s(\{0,1\}^d)$ is an arbitrary pair with $\EMD(x,y) \leq r$, and $x$ is locally-dense with respect to $\mu$, then when we sample $\bOmega$, we obtain the expected bound on $\Chamfer(x,\bOmega)$, and then (\ref{eq:better-chamfer}) implies the expected $\EMD_{\bT}(x,y)$ is at most $\tilde{O}(\log s) \cdot r$.
\end{itemize}

\paragraph{Step 2(b): Proof of Equation (\ref{eq:better-chamfer}) (in Section~\ref{sec:proofofexpand}).}

Consider a pair of points $x, y \in \EMD_s(\{0,1\}^d)$ where $x$ is locally-dense for $\mu$, and a sample $\bOmega$. Our goal is now to upper bound the expectation of $\EMD_{\bT}(x, y)$. %and following the discussion of probabilistic trees and $\EMD$, 
We proceed by bounding the expected distortion of the probabilistic tree $\bT$ as an embedding of $\{0,1\}^d$, where $\bT$ is generated from $\SampleTree(\mu, m)$, using the data-dependent edge weights on $\hat{\bOmega}$. Consider any $a \in x$ and let $b \in y$ be the vector assigned to $a$ in an optimal matching which realizes $\EMD(x,y)$, and furthermore, let $c \in \bOmega$ be the closest vector to $a$. 

The distance $d_{\bT}(a, b)$ is given by the sum of edge weights along the path in $\bT$ from the leaf containing $a$ to the leaf containing $b$. We break up the path into four segments, which naturally divides into two segments (one for each $a$ or $b$) that meet at the lowest common ancestor (LCA) of $a,b$:
\begin{itemize}
    \item The first segment comes up from the leaf containing $a$, and proceeds up via edges $(v, v_u)$ satisfying \textbf{(i)} the vector $b$ is not in $v_u$'s subtree, i.e., $a$ and $b$ have been ``split'' above node $v_u$, and \textbf{(ii)} the weight on $(v, v_u)$ is data-independent, so $\Elms(v) \cap \hat{\bOmega}$ is empty, and thus $a$ has also been split from every vector $c'$ with $\|c - c'\|_1 \leq 1$ (from $\hat{\bOmega})$ above node $v$.
    \item The second segment continues up after the first segment, on edges $(v, v_u)$ satisfying \textbf{(iii)} the vector $b$ is not in $v_u$'s subtree, so $a$ and $b$ remain ``split'' before $v_u$; however, \textbf{(iv)} the weight on $(v, v_u)$ is data-dependent, so $\Elms(v) \cap \hat{\bOmega}$ is non-empty.
\end{itemize}
The third and fourth segment proceed up from $b$, and are defined analogously. Note that, the second and fourth segment meet at the LCA $v$ of $a$ and $b$, and that any of the four segments may be empty. We overview the expected contribution of the first and second segments (and the third and fourth follow analogously).\footnote{Even though $a$ and $b$ are non-symmetric (as $x$ is the locally-dense point), the triangle inequality implies $\Chamfer(y, \bOmega) \leq \Chamfer(x, \bOmega) + \EMD(x, y)$, which will lead to only a constant factor loss in the symmetric argument.} The second segment (that which contains only data-dependent edge weights) is easiest to upper bound (Lemma~\ref{lem:sample-tree-expan-dep} in Section~\ref{sec:proof-sample-tree-exp-dep}). Roughly speaking, imagine a data-dependent probabilistic tree $\tilde{\bT}$ on $\hat{\bOmega} \cup \{ a, b \}$, which by~\cite{CJLW22} satisfies $\Ex_{\tilde{\bT}}[d_{\tilde{\bT}}(a, b)] \leq \tilde{O}(\log s) \cdot \|a-b\|_1$. The edge weights on the second segment from $\hat{\bOmega}$ differ from those of $\hat{\bOmega} \cup \{ a, b \}$ only in that the vector $a$ contributes to the average distance from $(v, v_u)$ in $\hat{\bOmega} \cup \{ a, b\}$ but may not in $\hat{\bOmega}$; since $\Elms(v_u) \cap \hat{\bOmega}$ is always non-empty in the second segment, one can account for this by losing a constant factor. 

We turn to the first segment, where we incorporate $\|a - c\|_1$, which later contributes to $\Chamfer(x, \bOmega)$ in (\ref{eq:better-chamfer}) (Lemma~\ref{lem:sample-tree-expan-ind}). Here, the important point is that an edge at depth $\ell$ contributes to the first segment whenever, among the sampled coordinates up to depth $\ell$, there is a coordinate $\bi_j \sim [d]$ where $a_{\bi_j} \neq b_{\bi_j}$ (so $a, b$ split), and in addition, there is another coordinate $\bi_{j'} \sim [d]$ where $a_{\bi_{j'}} \neq c_{\bi_{j'}}$. This is because at least two coordinate samples $\bi_{j'},\bi_{j''} \sim [d]$ must disagree on settings of $a$ and $c$ (otherwise, some $c'$ and $a$ are not split), but one of them may be $\bi_j$. This crucial observation will imply our desired bound. Consider the two levels $\ell_b$ and $\ell_c$ of the tree where:
\[ \frac{d}{2^{\ell_{b}}} \leq \|a - b\|_1 \leq 2 \cdot \frac{d}{2^{\ell_{b}}} \qquad \text{and} \qquad  \frac{d}{2^{\ell_{c}}} \leq \|a - c\|_1 \leq 2 \cdot \frac{d}{2^{\ell_c}}. \]
We defined $\ell_b, \ell_c$ such that $\bT$ splits $a,b$ before depth $\ell_b$ with constant probability, and $a,c$ before depth $\ell_c$ with constant probability. One can show for $z \in \{ b, c\}$ and $k > 0$, the probability that $\bT$ splits $a, z$ before depth $\ell_z - k$ is $\Theta(2^{-k})$. Recall that the data-independent edge weights at level $\ell$ are $\xi \cdot d / 2^{\ell}$, for $\xi = \tilde{O}(\log s)$. Thus, we upper bound the expected contribution of the first segment by considering levels which are below $\ell_b$, potentially between $\ell_b$ and $\ell_c$, and above both $\ell_c$ and $\ell_b$:
\begin{itemize}
    \item \textbf{Levels Below $\ell_b$}: Here, weights on levels $\ell_b + k$ contribute edge weight $\xi \cdot d / 2^{\ell_b + k}$, which is equal to $O(\xi) \cdot \|a-b\|_1 \cdot 2^{-k}$, and summing over $k > 0$ gives a geometric sum $O(\xi) \cdot \|a-b\|_1$.
    \item \textbf{Levels Between $\ell_b$ and $\ell_c$}: There are potentially $(\ell_b - \ell_c)^+$ levels between $\ell_b$ and $\ell_c$, and each level $\ell = \ell_b - k$ contributes edge weight $O(\xi) \cdot \|a-b\|_1 \cdot 2^k$, but since $a,b$ must split before $\ell_b - k$, the edge appears with probability $\Theta(2^{-k})$. This gives a total contribution of $O(\xi) \cdot \|a-b\|_1 \cdot (\ell_b - \ell_c)^+$.
    \item \textbf{Levels Above $\ell_c$ and $\ell_b$}: The weight of a level $\ell = \ell_b - (\ell_b - \ell_c)^+ - k$ is $O(\xi) \cdot \|a-b\|_1 \cdot 2^{(\ell_b - \ell_c)^+ + k}$. Here, such edges appear with probability $\Theta(2^{-((\ell_b-\ell_c)^+ + k)})$ since $a,b$ must split, \emph{times} $\Theta(2^{-k})$ since $a,c$ must split (the events are negatively correlated since sampled coordinates must be distinct). For each $k > 0$, this gives a contribution of $O(\xi) \cdot \|a-b\|_1 \cdot 2^{-k}$, which forms a geometric sum. 
\end{itemize}
%In summary, we obtain
%\begin{align}
%    &O(\xi) \cdot\|a-b\|_1 \left(\underbrace{\sum_{k < 0} 2^{k}}_{\text{below $\ell_b$}} + \underbrace{\sum_{k \in \{0,\dots, \ell_c-\ell_b\}} 2^k \cdot \Theta(2^{-k})}_{\text{between $\ell_b, \ell_c$}} + \underbrace{\sum_{k>0} 2^{(\ell_c - \ell_b)^+ + k} \cdot \Theta(2^{-((\ell_c - \ell_b)^+ + k)}) \cdot  \Theta(2^{-k}) }_{\text{above $\ell_c$ and $\ell_b$}} \right) \nonumber \\
%        &\qquad =O(\xi) \cdot \|a - b\|_1 \cdot \left( 1 + (\ell_c - \ell_b)^{+} \right) = \tilde{O}(\log s) \cdot \|a-b\|_1 \cdot \left(1 + \log\left(\frac{\|a - c\|_1}{\|a-b\|_1} + 1\right) \right). \label{eq:segments}
%\end{align}
%Levels below $\ell_b$ contribute a total of $O(\xi) \|a-b\|_1$, since we obtain a geometric sum over $k < 0$. Then, there are possibly $(\ell_c - \ell_b)^{+}$ levels (those between $\ell_b$ and $\ell_c$) whose edge weight is $\xi \cdot \|a-b\|_1 \cdot 2^k$, but appear with probability at most $\Theta(2^{-k})$ (since $a,b$ must be split); the total contribution is thus $O(\xi) \cdot \|a - b\|_1 (\ell_c - \ell_b)^+$. Lastly, we have the levels above $\ell_c$ and $\ell_c$ whose weight is $O(\xi) \|a - b\|_1 \cdot 2^{(\ell_c - \ell_b)^+ + k}$ but appear with probability at most $\Theta(2^{-((\ell_c-\ell_b)^+ +k)})$ (from splitting $a,b)$) times $\Theta(2^{-k})$ (from splitting $a,c$). This final term is also $O(\xi) \cdot \|a-b\|_1$, since we obtain another geometric sum. 
This gives the argument for the first segment; the analogous argument for the third and fourth segments gives an expected bound for $d_{\bT}(a,b)$ where the bottleneck are the first and third segments. Using the setting of $\xi = \tilde{O}(\log s)$, our bound becomes
\begin{align}
    \Ex_{\bT}\left[ d_{\bT}(a,b) \right] &= \tilde{O}(\log s) \cdot \|a-b\|_1 \cdot \left( 1 + (\ell_b - \ell_c)^+ \right) \nonumber \\
        &= \tilde{O}(\log s) \cdot \|a-b\|_1 \left(1 + \log\left(\frac{\|a-c\|_1}{\|a-b\|_1} + 1 \right) \right) 
    \label{eq:segments}
\end{align}

The last part, which combines the individual bounds for $\Ex_{\bT}[d_{\bT}(a,b)]$ uses Jensen's inequality: let $\pi \colon x\to y$ denote the bijection realizing $\EMD(x,y)$ and $\sigma \colon x \to \bOmega$ the mapping realizing $\Chamfer(x, \bOmega)$. Then, consider the distribution $\calD$ over $x$ which samples $\ba$ with probability proportional to the contribution of $\ba$ in $\EMD(x,y)$, i.e., $\|a - \pi(a)\|_1 / \EMD(x,y)$:
\begin{align*}
    \Ex_{\bT}\left[ \dfrac{\EMD_{\bT}(x,y)}{\EMD(x,y)}\right] \leq \Ex_{\bT}\left[ \Ex_{\ba\sim\calD}\left[\dfrac{d_{\bT}(\ba,\pi(\ba))}{\|\ba-\pi(\ba)\|_1} \right] \right] &\leq \Ex_{\ba\sim\calD}\left[ \tilde{O}(\log s) \cdot \left(1 + \log\left(\frac{\|\ba - \sigma(\ba)\|_1}{\|\ba - \pi(\ba)\|_1} +1\right) \right) \right] \\
%    &\leq \tilde{O}(\log s) \left(1 + \log\left( \Ex_{\ba}\left[ \dfrac{\| \ba - \sigma(\ba)\|_1}{\|\ba - \pi(\ba)\|_1}\right] + 1\right)\right) \\
%    &\leq \tilde{O}(\log s) \left(1 + \log\left( \sum_{a \in x} \frac{\|a - \pi(a)\|_1}{\EMD(x, y)} \cdot \dfrac{\| a - \sigma(a)\|_1}{\|a - \pi(a)\|_1} + 1\right)\right) \\
    &\leq \tilde{O}(\log s) \left(1 + \log\left( \dfrac{\Chamfer(x, \bOmega)}{\EMD(x,y)} + 1\right) \right),
\end{align*}
which completes (\ref{eq:better-chamfer}).

\textbf{Step 2(c): Locally Dense and non-Locally Dense Points. } Given the analysis of $\SampleTree(\mu, m)$, we may compose (\ref{eq:extension-intro}) with a LSH for $\ell_1$ to obtain a hash family with the following properties (see Lemma~\ref{lem:ingred-3}). For an arbitrary choice of threshold $r > 0$, and $0 < p_2 < p_1 < 1$, and any distribution $\mu$, the hash family $\calH$ (which depends on $\mu$) has approximation $c = \tilde{O}(\log s)$. It always satisfies the ``$p_2$-property'' (i.e., that far points separate) because the $\SampleTree(\mu,m)$ is non-contracting, but only satisfies the ``$p_1$-property'' on close pairs points $x,y$ where $x$ is locally-dense with respect to $\mu$. As mentioned, the important property of ``locally-dense'' is that, if we consider $\by_1,\dots, \by_m \sim \mu$ (where $m$ is only $\poly(\log s)$), then setting $\bOmega = \bigcup_{i=1}^m \by_i$ satisfies $\Chamfer(x, \bOmega) \leq r \cdot \log^{10} s$ in expectation (we used $10$ as an arbitrary setting of the $\poly(\log s)$ to illustrate the point-to-come).
%\begin{itemize}
%\item \textbf{Far Points Separate}: $\Prx_{\bh \sim \calH}\left[ \bh(x) = \bh(y)\right] \leq p_2$ for every $x,y \in \EMD_s(\{0,1\}^d)$ with $\EMD(x,y) \geq cr$. 
%\item \textbf{Close Points Collide}: $\Prx_{\bh \sim \calH}\left[ \bh(x) = \bh(y) \right] \geq p_1$ for every $x, y \in \EMD_s(\{0,1\}^d)$ where $\EMD(x,y) \leq r$ and $x$ is locally-dense with respect to $\mu$.
%\end{itemize}

Now divide $\mu$ into two regions: the locally-dense points, and the remainder. The $\SampleTree(\mu,m)$ embedding composed with an LSH for $\ell_1$ handles the locally-dense region. The remaining region is handled by the following observation. %If  $x \in \EMD_s(\{0,1\}^d)$ is \emph{not locally-dense}, then it is Chamfer Distance (and therefore $\EMD$) at least $r \log^{10} s$ from the majority of points $y \sim \mu$. Since $x$ is so far, when attempting to handle $x$ at scale $r$ (i.e., when considering the $p_2$ property), we can get away with a much weaker approximation. In particular, a $O(\log^2 s)$ approximation would suffice to satisfy the $p_2$ property for such points. 
We consider a point $x$ and sample from a (weak) data-independent LSH of~\cite{AIK08}, $\calH$, which is $(r, \tilde{c}r, p_1, p_2)$-sensitive with $\tilde{c} = O(\log^2 s)$. Then, the ``$p_1$-property'' still holds for any pair of points $x, y$, since $\EMD(x,y) \leq r$ implies that $\bh(x) = \bh(y)$ with probability at least $p_1$. Moreover, the ``$p_2$-property'' on points which are \emph{not locally-dense for} $\mu$ follow from the following %will follow by the above discussion. Specifically, our key observation can be summarized as follows.
\begin{quote}
\textbf{Key Idea 2:} Suppose $x$ is \emph{not locally-dense} for $\mu$. Then if we sample $\by \sim \mu$, the point $\by$ is likely to satisfy $\EMD(x, \by) \geq r \cdot \log^{10} s$; otherwise,  taking $m - 1$ additional samples to define $\bOmega$ (which includes $\by$) would satisfy $\Chamfer(x, \bOmega) \leq \EMD(x, \by) \leq r \cdot \log^{10} s$. Thus if $\by \sim \mu$ satisfies $\EMD(x, \by) \geq r \cdot \log^{10} s$, then we can use the (weaker) data-independent LSH $\calH$. Note that, $\EMD(x,\by) \geq \log^{10} s \cdot r$ is much larger than $\tilde{c} r$, so $x$ and $\by$ collide in $\calH$ with probability at most $p_2$, since $\log^{10} s \gg \tilde{c} = \log^2 s$. % which obtains approximation $O(\log^2 s)$ at threshold $r$, and we are asking it to separate a point at distance $r \cdot \log^{10} s$.
\end{quote}
In summary, the $\SampleTree(\mu, m)$ embedding captures locally-dense regions of $\mu$, and, in the remainder, it suffices to handle randomly sampled points $\by \sim \mu$ (which suffice for the ``$p_2$-property'' in data-dependent LSH). By definition, the uniform samples are expected to be \emph{very far} from locally-dense regions, so it suffices to utilize data-independent LSH which achieve weaker approximations.

In Section~\ref{sec:crucial-ingred}, we execute the above plan. We define a collection of (data-independent) LSH families which appear to be weak (and lead to the $O(\log^2 s)$-approximation). These LSH families always satisfy the ``$p_1$-property'' (Lemma~\ref{def:data-ind-hash}), but not a good ``$p_2$-property.'' Then, we connect failure of the $p_2$-property on these LSH families to the expected Chamfer distance to a randomly sampled collection $\bOmega$. Namely, we consider a point $x \in \EMD_s(\{0,1\}^d)$, and we assume that the hash families from Lemma~\ref{def:data-ind-hash} fail to separate randomly sampled points from $\mu$. In Section~\ref{sec:crucial-ingred}, we call these points ``locally-dense'' (Definition~\ref{def:ingred-2}), and show in Lemma~\ref{lem:dense-conse} that these are points whose expected Chamfer distance to $\bOmega$ is at most $r \cdot \poly(\log s)$. 

\textbf{Step 3: Gluing LSH for Locally-Dense and Non-Locally Dense Regions.}
The final step involves a ``gluing'' operation, which uses various hash families (for different regions of $\mu$) to define a single data-dependent LSH family for all $\mu$. Up to now, we have constructed:
\begin{itemize}
    \item A hash family coming from $\SampleTree(\mu, m)$, which always has a good ``$p_2$-property,'' but only has a good ``$p_1$-property'' on points $x$ which are locally-dense for $\mu$.
    \item A collection  of data-independent LSH families, $\calH(\tau, \ell)$ for (fixed) threshold $\tau>0$ and each $\ell \in \{0, \dots, L \}$ for $L = O(\log d)$ in Lemma~\ref{lem:ingred-1}. Here, the level $\ell$ corresponds to a level of the (data-independent) tree embedding, which is then embedded into $\ell_1$, and thereafter hashed via a $\ell_1$ LSH  (see Definition \ref{def:data-ind-hash} for full details). 
    
    % The hash families always satisfy a ``$p_1$-property'' (with a slightly smaller $p_1$), but only satisfy the ``$p_2$-property with approximation $c= O(\log^2(n))$. 
\end{itemize}

In Section~\ref{sec:glue}, we glue these hash families together, and obtain a data-dependent LSH which is $(r, cr, p_1,p_2)$-sensitive for $\mu$ (proving Theorem~\ref{thm:LSH-main}). The gluing proceeds as follows: for a fixed threshold $\tau > 0$ (which depends on the parameters $r, p_1$ and $p_2$ which we wish to obtain), we sample hash functions $\bh_1,\dots, \bh_{L}$ where $\bh_{\ell} \sim \calH(\tau, \ell)$ for each $\ell \in \{0,\dots, L\}$ and $L=O(\log d)$, as well as a hash function $\bh_{*}$ resulting from $\SampleTree(\mu,m)$. Importantly, the hash families $\calH(\tau, \ell)$ are initialized to be $(r, \tilde{c} r, p_1 / L, \tilde{p}_2)$-sensitive for an approximation $\tilde{c}$ (which will be a large $\poly(\log s)$), and an appropriate value of $\tilde{p}_2$ for Step 2 to go through (i.e., failure of the ``$\tilde{p}_2$-property'' for $\bh_1,\dots,\bh_L$ implies a bounded Chamfer distance to $\bOmega$).
Our final key observation is as follows:
\begin{quote}
    \textbf{Key Idea 3}: For a hash family $\calH$, distribution $\mu$, point $x$, and a draw $\bh \sim \calH$, the point $x$ can check whether (a stronger version of) its own ``$p_2$-property'' holds given  $\bh$.   In particular, one hashes the point $\bh(x) = u$, and for the (now fixed) $\bh$, one can computes the probability that $\by \sim \mu$ satisfies $\bh(\by) = u$ by simply looking at the probability mass of points which hash to the bucket $u$ (if $\mu$ is the uniform distribution, this is just proportional to the size of the hash bucket). If this probability mass is at most  $p_2$, then the ``$p_2$-property'' necessarily holds for $x$ conditioned on $\bh$. 
\end{quote}
The above check is for a stronger ``$p_2$-property'', since 
we are not also checking whether  $\by \sim \mu$ is far from $x$. Note that if this `$p_2$-property'' holds for some $\ell \in \{0,\dots,L\}$, then we can hash $x$ to this bucket and make significant progress by reducing the size of the dataset. 
Given the above observation, the gluing proceeds by letting
\[ \bh(x) = (\bell(x), \bh_{\bell(x)}(x)),\]
where $\bell(x)$ is the smallest $\ell \in \{0,\dots, L\}$ where the above ``$p_2$-property'' check succeeds for $x$ with the hash function $\bh_{\ell}$. If it always fails, then $\bell(x) = *$, thereby signifying that the hash output will be determinined by the output of the  $\SampleTree$ data-dependent LSH. Since each $\calH(\tau,\ell)$ collides close points with probability $p_1/L$, we can union bound over the $L$ levels to ensure that a close pair of points collide in all $L$ draws with probability at least $p_1$, thus $\bell(x) = \bell(y)$ for a close pair $(x,y)$ with probability at least $p_1$. Using this, the ``$p_1$-property'' follows immediately whenever $\ell \neq *$. On the other hand, if $\ell = *$, this indicates a failure of the $p_2$ property for each of the data-independent families, which as we have shown implies a bounded Chamfer distance from $x$ to a random sample $\Omega$, which in turn implies that $x$ is locally dense and therefore the $p_1$ property holds for $x$ under the $\SampleTree$ LSH $\bh_{*}$ (and thus holds for the full ``glued'' hash function). Finally, for the ``$p_2$ property'', if $\ell(x) \neq *$ then by definition of $\ell(x)$ we have split $x$ from all but a $p_2$ fraction of $\mu$, and otherwise the hash of $x$ is determined by $\SampleTree$, which always satisfies the desired ``$p_2$ property''. Putting together the above arguments will complete the proof of the Theorem \ref{thm:data-dep-hashing}.

%% file: related-work.tex
\subsection{Other Related Work}
The computational aspects of EMD date back over 70 years to the Hungarian algorithm of~\cite{K55}. Since then, significant work has gone into investigating the computational complexity of EMD in many settings. In what follows, we address two other settings of relevance, and refer the reader to \cite{peyre2019computational} for a more in depth survey on EMD and its modern applications.

\textbf{Approximation Algorithms for EMD.} The problem of approximating the EMD between two sets of size $s$ in a metric space has recieved significant attention. One of the most popular methods is the Sinkhorn algorithm~\cite{C13} (also see \cite{altschuler2017near,le2021robust,pham2020unbalanced}), which gives additive error approximations in quadratic $O(s^2)$ time. 
For computing the Euclidean EMD between two point sets, even though the input is size $O(s \cdot d)$, it is known that no $o(s^2)$ time exact algorithm can exist unless well-known fine grained complexity conjectures are false~\cite{rohatgi2019conditional}. Nevertheless, techniques from the sublinear algorithms community, such as locality sensitive hashing, have also been used for faster offline algorithms to approximately compute EMD.  For instance, 
Andoni and Zhang~\cite{andoni2023sub} recently gave the first $n^{2-\poly(\eps)}$ time algorithm for computing $(1+\eps)$ approximations to high-dimensional Euclidean EMD, based on constructing sub-quadratic spanners via LSH. Furthermore, \cite{beretta2023approximate} gave a subquadratic additive approximation for \emph{any} metric space. 

%Add back later
%Using tree-embedding techniques, ~\cite{bakshi2023near} gave a near-linear algorithm to compute for the \textit{Chamfer Distance}, which is a relaxation of EMD. 

\textbf{Low-Dimensional Space.} While the focus of this work is on EMD over \textit{high-dimensional} spaces $(\R^d,\ell_p)$, which is the common setting in many modern ML applications where the inputs are embeddings in a high-dimensional latent space (see e.g. \cite{KSKW15}), EMD over lower dimensional spaces has also received attention from the sublinear algorithms community, such as in sketching \cite{I04b,ABIW09} and parallel algorithms \cite{andoni2014parallel}.  In this setting, one can often obtain much better approximations, such as $(1+\eps)$ approximations, if one is okay with the runtime depending exponentially on the dimension---this allows for a new set of techniques such as $\eps$-nets to be employed, which would be too costly in high dimensional settings. An importance case is that of the plane (i.e., $d=2$), where sketches that achieve constant factor approximations are known \cite{ABIW09}. More generally, when the dimension is a constant, there are also offline approximation algorithms with near-linear runtime \cite{sharathkumar2012near,fox2022deterministic,choudhary2020extremal}

\begin{remark}[ANN for EMD with small $s,d$]\label{remark:lowd}
For the problem of nearest neighbor search, we remark that if we are allowed runtime that is exponential in both $s$ and $d$, then $(1+\eps)$ approximations are possible in sublinear in $n$ time. Specifically, it is straightforward to prove that the doubling dimension of the space $\EMD_s(\R^d,\ell_p)$ is at most $O(sd\log s)$. Thus it is possible to obtain $\exp(s,d)$ query time nearest neighbor search algorithms by using techniques such as navigating nets \cite{krauthgamer2004navigating}. Note that such general techniques employ linear scans over $\eps$-net like objects, and do not use structure specific to the Earth Mover Distance metric beyond its doubling dimension. %For a more practical context, this parametrization aligns nicely with practical applications in machine learning and natural language processing, where $n$ may be a growing collection of ``descriptors'' or ``documents'' (on the order of millions or even billions), but the number of tokens $s$ and the dimensionality $d$ of the individual tokens is on the order of hundreds (see also~\cite{KSKW15, BDIRW20, babenko2016efficient, bakshi2023near}). 
\end{remark}

%% file: Prelims.tex
\section{Preliminaries}

\newcommand{\EMDsd}{\mathsf{EMD}_s(\R^d)}
\newcommand{\EMDsh}{\mathsf{EMD}_s(\{0,1\}^d)}
\textbf{Notation.} For any integer $n \geq 1$, we write $[n] = \{1,2,\dots,n\}$, and for two integers $a,b \in \Z$, write $[a:b] = \{a,a+1,\dots,b\}$.  For $a,b \in \R$ and $\eps \in (0,1)$, we use the notation $a = (1 \pm \eps) b$ to denote the containment of $a \in [(1-\eps)b , (1+\eps)b]$. 
We will use boldface symbols to represent random variables and functions, and non-boldface symbols for fixed values
(potentially realizations of these random variables) for instance $\boldf$ vs, $f$. 

We denote the metric space consisting of multi-sets of $s$ points in a metric space $(X,d)$, where the distance between sets is the Earth Mover's Distance metric, by $\EMD_s(X,d)$.  For instance, $\EMD_s(\R^d,\ell_1)$ denotes the Earth Mover's Distance metric over sets of $s$ points living in $\R^d$ with the $\ell_1$ metric. We refer to the metric space $(X,d)$ as the ground metric of $\EMD_s(X,d)$. When the distance over $X$ is understood by context (e.g. over the Hamming cube $\{0,1\}^d$), we can drop the distance and simply write $\EMD_s(X)$.  Moreover, we use $\EMD(x,y)$ to denote the real-valued metric function of $\EMD_s(\R^d,\ell_1)$ or $\EMD_s{\{0,1\}}$, where the choice of the aforementioned two ground metric is understood from context via the type of the input parameters $x,y$. 

Given a rooted tree $T = (V(T),E(T))$, all edges of $T$ will be directed from parent to child, so an edge $(u,v) \in E(T)$ denotes an edge from the parent $u$ to the child $v$. We will often abuse notation and write $u \in T$ to denote that $u \in V(T)$. Given a rooted tree with weighted edges $T = (V(T), E(T), W(T))$, we abuse the notation $T$ to denote the tree metric $(V(T), d_T)$ where, for $u,v \in V(T)$, $d_T(a,b)$ is defined as the length of the shortest weighted path between $u,v$. We use $\EMD_T$ to denote the metric function of $\EMD_s(V(T), d_T)$.

\begin{remark}[On Embedding $\ell_p$ into $\ell_1$]\label{remark:ell1}
  For the remainder of the paper, we will prove all our upper bounds for the case that the ground metric is $(\R^d,\ell_1)$. To extend this to general $p \in [1,2]$, we can use well-known (data-independent)  embeddings from $(\R^d,\ell_p)$ into $(\R^{d'},\ell_1)$ with $d' = O(d)$ \cite{johnson1982embedding}, which preserve all distances up to $(1 \pm \eps)$ for any arbitrary constant $\eps>0$. This embedding is a randomized linear function, thus applying it to each vector $x \in (\R^d,\ell_p)$ will only increase the runtime by a multiplicative factor of $O(d)$, and increase the space by an additive $O(d^2)$, which will therefore not effect the stated complexity in our theorems.
\end{remark}

%% file: data-dep-hashing.tex
%!TEX root = main.tex

\section{Nearest Neighbors, Embeddings, and Data-Dependent Hashing}

In this section, we define the approximate \emph{near} neighbor search problem and data-dependent hashing, and also demonstrate how, given a data-dependent hashing family for a metric space, we can obtain a data structure for approximate near neighbor search with overhead analogous to that of (data-independent) locality-sensitive hashing.
%We build data structures for approximate near neighbor search for $\EMD$ using two techniques.
%First, by a dynamic embedding approach, which will give rise to a (weaker) $O(\log n)$-approximation for approximate near neighbor via the dynamic and data-dependent probabilistic tree embedding. Then, we give our main technical tool, which will be data-dependent locality-sensitive hashing for $\EMD$, which will achieve approximation $\tilde{O}(\log s)$. We will show how, given a data-dependent hashing family for a metric space, we obtain a data structure for approximate near neighbor search with overhead analogous to that of (data-independent) locality-sensitive hashing.

\begin{definition}[Approximate Near Neighbor]
Let $(X, d_X)$ be a metric space, $r > 0$ be a threshold, and $c > 1$ be an approximation. The $(c,r)$-approximate near neighbor problem is the following data structure problem:
\begin{itemize}
\item \textbf{\emph{Preprocessing}}: We receive a dataset $P \subset X$ of $n$ points to preprocess into a data structure.
\item \textbf{\emph{Query}}: A query is specified by any point $q \in X$, and a query is \emph{correct} whenever the following occurs. If there exists a point $p \in P$ with $d_X(p, q) \leq r$, the data structure outputs a point $\hat{p} \in P$ with $d_X(\hat{p}, q) \leq cr$. 
\end{itemize}
A data structure solves the $(c, r)$-approximate near neighbor problem if, for every (fixed) dataset $P \subset X$ and query $q \in X$, following preprocessing of $P$, the data structure answers correctly on $q$ with probability at least $9/10$ over the construction of the data structure.
\end{definition}

We remark that, by a standard reduction (see \cite{HIM12}), it will suffice to solve the $(c,r)$-approximate near neighbor problem above. 

\subsection{Approximate Nearest Neighbor via Data-Dependent Hashing}\label{def:ann-from-dd}

We remark that he definition of data-dependent hashing (Definition~\ref{def:data-dep}) that we obtain in this paper is slightly more stringent than the one presented in~\cite{AR16} (requiring that for any point $x \in X$, a randomly drawn ``far'' point is separated)---\cite{AR16} focused on lower bounds, so a less-stringent definition gives a stronger lower bound result; since we will show upper bounds, a more stringent definition gives a stronger result. We state the definition and show how a data-dependent hashing family implies a data structure for approximate near neighbor search. The proof itself is similar in spirit to that of~\cite{IM98} and deferred to Section~\ref{sec:dd-hash-to-ann}. The one subtlety is that, because our hashing family \emph{depends} on the dataset, one must instantiate it to the desired dataset before using it.

\begin{definition}[Data-Dependent Hashing]\label{def:data-dep}
For a metric $(X, d_X)$, a distribution $\mu$ over $X$, and a threshold $r > 0$, we say that a distribution $\calD$ over maps $h \colon X \to U$ is $(r, c r, p_1,p_2)$-sensitive for distribution $\mu$ if
\begin{itemize}
\item \textbf{\emph{Close Points Collide}}: For any two points $x, y \in X$ with $d_X(x, y) \leq r$, we have 
\[ \Prx_{\bh \sim \calD}\left[ \bh(x) = \bh(y)\right] \geq p_1. \]
\item \textbf{\emph{Far Points Separate on Average}}: For any point $x \in X$, we have
\[ \Prx_{\substack{\bh \sim \calD \\ \by \sim \mu}}\left[ \begin{array}{c} d_X(x, \by) > c\cdot r,\\
 \bh(x) = \bh(\by) \end{array}\right] \leq p_2. \]
\end{itemize}
\end{definition}

\begin{definition}[Data Structure for Data-Dependent Hashing]\label{def:ds-data-dep}
For a metric $(X, d_X)$, a data structure for data-dependent hashing with a $(r, cr,p_1,p_2)$-sensitive family satisfies the following:
\begin{itemize}
    \item \emph{\textbf{Preprocessing}}: The data structure preprocesses the description of a distribution $\mu$ supported on $X$, and maintains a draw of $\bh$ from a $(r, cr, p_1, p_2)$-sensitive family for $\mu$.
    \item \emph{\textbf{Query}}: Given any point $q \in X$, the data structure outputs the value of $\bh(q)$.
\end{itemize}
We let $I_{\sfh}(n)$ denote the time of instantiating the data structure with a distribution $\mu$ supported on $n$ points, and let $Q_{\sfh}(n)$ denote the worst-case query time.
\end{definition}

\begin{theorem}[Data-Dependent Hashing to Approximate Near Neighbors]\label{thm:hashing-to-nn}
Let $(X, d_X)$ be a metric, $r > 0$ be a threshold, $c > 1$ be an approximation, and $p_1, p_2 \in (0, 1)$ be two parameters, where $\rho \in \R$ is the parameter
\[ \rho = \dfrac{\log(1/p_1)}{\log(1/p_2)}. \]
Suppose there is a data structure for data-dependent hashing with a $(r, cr, p_1,p_2)$-sensitive family with preprocessing time $I_{\sfh}(n)$ and query time $Q_{\sfh}(n)$. %there are two algorithms $\GenHash, \EvalHash$ with the following guarantees:
%\begin{itemize}
%\item $\GenHash$ takes as input a set of at most $n$ points and in time $T_{\GH}(n)$, takes a sample $\bh \sim \calD$ from a distribution over maps $h \colon X \to U$ which is $(r, cr, p_1,p_2)$-sensitive for the uniform distribution over the set of input points, and outputs a description of $\bh$.
%\item $\EvalHash$ takes as input a point $x \in X$ and the description of a hash function $h \colon X \to U$ and evaluates $h(x)$ in time $T_{\sfh}$.
%\end{itemize}
Then, there exists a data structure for the $(c,r)$-approximate near neighbor problem which satisfies:
\begin{itemize}
\item \textbf{\emph{Preprocessing Time}}: The data structure preprocesses a size-$n$ dataset in time at most 
\[ O\left( n^{\rho} / p_1 \cdot \log_{1/p_2} n \cdot \left( I_{\sfh}(n) + n \cdot Q_{\sfh}(n) \right) \right),\]
and therefore its space complexity is at most that amount.
\item \textbf{\emph{Query Time}}: A query to the data structure is answered in time at most
\[ O\left( n^{\rho} / p_1 \cdot \log_{1/p_2} n \cdot Q_{\sfh}(n) \right).\]
\end{itemize}
\end{theorem}

%% file: DynamicEmbedding.tex
\section{Dynamic and Data-Dependent Probabilistic Tree Embeddings}
\label{sec:Dynamic Embedding}

%\enote{I started making a pass. In general, I'm a big fan of writing random variables in bold, as this makes it much nicer to read and follow what is randomized and what is deterministic. I've started making a pass on this, only changing the Figure since this is used in later section. I've also made a few notational changes to be consistent with later sections: $\sigma$ became $u$, $H$ became $T$, and the levels are indexed by $\ell$ instead of $i$.}

In this Section, we describe the dynamic, data-dependent probabilistic tree embedding from Theorem~\ref{thm:warm-up}. Even though Theorem~\ref{thm:warm-up} is not directly necessary for the proofs of the main results of this work (Theorems~\ref{thm:data-dep-hashing} and~\ref{thm:ann-main}), they elucidate the benefits and challenges of using the~\cite{CJLW22} probabilistic tree embedding for nearest neighbor search. For simplicity in this Section, we consider vectors which have integer coordinates $x \in [\Delta]^d = \{1,2,\dots,\Delta\}^d$ and our dependence will be logarithmic in $\Delta$. Note that given an upper bound on the aspect ratio $\Phi$ of the dataset (the ratio of the maximum distance to the minimum distance), one can always enforce this assumption by a re-scaling and discretization which introduces a minor constant-factor loss in the distortion. % with a constant loss in the bit-complexity of the algorithm.  

%Specifically, we prove the following:
\begin{theorem}[Dynamic and Data-Dependent Probabilistic Tree Embedding] \label{thm:dynamic-main}
    For a fixed $d \in \N$ and $p \in [1,2]$, there is a data structure supporting the following: %that initializes in time $O(d \log \Delta)$. amd supporting the following operations:
    \begin{itemize}
        \item \emph{\textbf{Maintenance}}: The data structure maintains a set $\Omega \subset [\Delta]^d$ of $m$ vectors, as well as a rooted probabilistic tree metric $\bT$ (whose distribution depends on $\Omega$), along with a non-contracting embedding $\varphi: (\Omega,\ell_p) \to \bT$, such that for any $x,y \in \Omega$:
        \[              \Ex_{\bT}\left[ d_\bT(\varphi(x),\varphi(y))\right] \leq \tilde{O}(\log (m d \Delta)) \cdot \|x-y\|_p.    \] 
        \item \emph{\textbf{Query}}: In time $O(d \log (d \Delta))$, we may query a vector $x \in \Omega$ and obtain the weighted path from the root to $\varphi(x)$ in $\bT$.
        \item \emph{\textbf{Insertions/Deletions}}: In expected time $O(d \log (d \Delta) + \log^2(d\Delta))$, we may add or remove vectors from the set $\Omega$. Since the updated $\bT$ depends on (the updated set) $\Omega$, the algorithm also returns (without additional computational overhead) the updated weighted paths of every vector in $\Omega$ whose path changed from the insertion/deletion.
        \end{itemize}
\end{theorem}

\ignore{
\begin{theorem}\label{thm:dynamic-main}
    Fix any set $S \subset \EMD_s(\{0,1\}^d)$ of $|S| = n$ points. Then there is an algorithm that computes in time $\tilde{O}(n s d)$ an embedding $\varphi: S \to (\R^m, \ell_1)$, for $m = \poly(n,s,\log d)$, such that $\varphi(x)$ is $\tilde{O}(s d^2)$-sparse for all $x \in S$, and such that for all $i \in [n]$ and all $x,y \in S$, with probability $1-1/\poly(n)$ we have
    \begin{equation}\label{eqn:embeddingbound}
    \EMD(x,y)   \leq  \|\varphi(x) - \varphi(y)\|_1  \leq  \tilde{O}(\log(m) + \log d) \cdot \EMD(x,y)    
    \end{equation}
    Furthermore, given an additional point $Q \in \EMDsh$ and a pre-computed embedding $\varphi: S \to (\R^m, \ell_1)$, one can compute an embedding $\varphi': S \cup \{Q\} \to (\R^m,\ell_1)$ that satisfies Equation \ref{eqn:embeddingbound} for all $x,y \in S \cup \{Q\}$ with probability $1-1/\poly(n)$. The expected time to compute $\varphi'$ given $(\varphi,q)$ is  $\tilde{O}(s^2 d^2 \log(m))$, and the expected number of $x \in S$ such that $\varphi(x) \neq \varphi'(x)$ is at most $\tilde{O}(s d^2 \log(m))$. 
\end{theorem}
}

%\begin{remark}
 %   Note that one can obtain a weaker embedding with the property that $   \EMD(x,y)   \leq  \|\varphi(x) - \varphi(y)\|_1 $ and $\ex{ \|\varphi(x) - \varphi(y)\|_1 } \leq  \tilde{O}(\log(m)) \cdot \EMD(x,y) $ with a faster update time of $O($
%\end{remark}

\subsection{Embedding for Subsets of the Hamming Cube}\label{sec:hamming-cube-dynamic}

We begin by proving Theorem~\ref{thm:dynamic-main} for the Hamming cube $\{0,1\}^d$ with $\ell_1$ metric (where we note that this sets $\Delta = 2$). This proof will already contain the major ideas, and subsequent sections will utilize the main definition of the $\quadtree$ sub-routine specified below. 
%defining a (static) data-dependent probabilistic tree embedding of a subset of the Hamming cube $\{0,1\}^d$. 
We later extend these ideas to $([\Delta]^d,\ell_p)$ for $p \in [1,2]$ in Appendix~\ref{sec:dynamic-l1}.

%\paragraph{Construction of Data-Dependent Probabilistic Trees.} 
We consider a subset $\Omega \subset \{0,1\}^d$ of $n$ vectors in the Hamming cube (we later show how to make this subset dynamic). For any (multi-)set of indices $\vec{i} = (i_1,i_2,\dots,i_t)  \in [d]^t$, define the projection $p_{\vec{i}}\colon \{0,1\}^d \to \{0,1\}^t$ which maps a vector $x \in \{0,1\}^d$ to $p_{\vec{i}}(x) = (x_{i_1},x_{i_2},\dots,x_{i_t})$. For any $t \in \N$, we consider the hash family $\calH_{t, d}$ given by
\begin{equation}\label{eqn:proj-hash-def}
\calH_{t,d}= \{p_{\vec{i}} \; : \; {\vec{i}} \in [d]^t \}.    
\end{equation}
Equivalently, a draw $\bphi$ from the hash family $\calH_{t,d}$ is given by sampling $t$ indices $\bi_1,\dots, \bi_t \sim [d]$ uniformly at random and letting $\bphi$ be the projection $p_{\vec{\bi}}$. Whenever $d$ is known from context, we drop the subscript and simply write $\calH_t$. %Thus, drawing a function $\bphi \sim \mathcal{H}_t$ is equivalent to sampling a tuple  $\vec{i} = (i_1,i_2,\dots,i_t)  \in [d]^t$ uniformly, and then setting $\phi = p_{\vec{i}}$. 
The construction of the (static) data-dependent probabilistic tree metric $\bT$ is described by the algorithm $\quadtree$ (in Figure~\ref{fig:DDquadtree-prelims}), which receives as input a set of vectors $\Omega \subset \{0,1\}^d$ and generates a random tree $\bT$ and a natural mapping from $\{0,1\}^d$ to leaves of $\bT$ (in Definition~\ref{def:quadtree-map}). We also allow $\quadtree$ to take an additional scaling parameter $\xi$. This scaling will not be needed in this section, and we can set it as $\xi=1$ (in fact, it will not effect the behavior of the algortihm in this section). However, we will need to set it carefully later on in Section \ref{sec:7}.
%
%$\sim \calT_{\Omega}$ depending on an input set of points $\Omega \subset \{0,1\}^d$. Afterwards, we will describe a natural mapping from $\{0,1\}^d$ to leaves of this tree $\bT$. 
%Our (static) data-dependent tree metric construction is now described below in Figure \ref{fig:DDquadtree-prelims}. 

    \begin{figure}[H]
	\begin{framed}
		\noindent Subroutine $\quadtree(\Omega, \xi)$
		\begin{flushleft}
			\noindent {\bf Input:} A subset of vectors $\Omega \subset \{0,1\}^d$, and a scaling parameter $\xi$ (if unspecified, set $\xi=1$).\\
			\noindent {\bf Output:} A probabilistic weighted tree $\bT$, as defined below. % drawn from a distribution $\calT_{\Omega}$, as defined below.% and a metric embedding $\varphi: \Omega \to \bT$.
			\begin{enumerate}
				\item Initialize a root node $v_0$ at depth $0$. We will let $L = O(\log d)$ (for a large enough constant, say $2$) denote the depth of the tree, and we define the notation which will indicate, for a node $v$,
                \[ \bElms(v) = \text{ subset of $\{0,1\}^d$ which will embed into the subtree at $v$. }\]
                Initially, we let $\bElms(v_0) = \{0,1\}^d$.
                \item For each $\ell=0,1,\dots,L-1$, sample a random hash function $\bphi_\ell \sim  \calH_{2^\ell}$, and let $\bphi_{L}:\{0,1\}^d \to \{0,1\}^d$ be the identity mapping $\bphi(x) = x$.  
			\item We initialize nodes $v$ at depths $1, \dots, L$, by the following inductive procedure which begins with $\ell = 0, \dots, L$: %perform the following:
				\begin{itemize}
					\item For every node $v$ at depth $\ell$, and every $u \in \{0,1\}^{2^{\ell}}$, we initialize a child node $v_{u}$ to $v$ (at depth $\ell+1$). We create the edge $(v,v_u)$, and set 
     \[\bElms(v_{u}) = \bElms(v) \cap \{x \in \{0,1\}^d \; | \; \bphi_{\ell}(x) = u \}. \]
                    The nodes at depths $L+1$ are leaves of $\bT$.
       %  \item If $\ell = L$, then $\bElms(v) = \{x\}$ for some single $x \in \{0,1\}^d$, so we define the embedding $\varphi(x) = v \in \bT$. 
				\end{itemize}
                    \item\label{en:quadtree-weight-setting} For every edge $(v,v_{u}) \in \bT$ where $v$ is at depth $\ell$, we assign the weight  \[ \bw(v, v_u) = \begin{cases} \vspace{0.25cm}
                        \Ex\limits_{\substack{\bc \sim \bElms(v) \cap \Omega \\ \bc' \sim \bElms(v_{u}) \cap \Omega}}\left[\|\bc - \bc'\|_1 \right] & \; \;  \bElms(v_{u}) \cap \Omega \neq \emptyset  \\
                         d/ 2^{\ell} \cdot \xi& \; \; \text{otherwise.}\footnotemark
                    \end{cases}\]
                
			\end{enumerate}
		\end{flushleft}
	\end{framed}
	\caption{The Data-Dependent $\quadtree$ Embedding.}\label{fig:DDquadtree-prelims}
\end{figure}

\footnotetext{As noted earlier, the symbol $\xi$ in the description of $\bw(v, v_{u})$ can be arbitrary in this section. Specifically, in this section, an edge of whose weight depends on $\xi$ will never be evaluated. Looking ahead, we have placed the parameter $\xi$ as it will become important in Section~\ref{sec:7}, where $\xi$ will be set to $O(\log s)$.}

\begin{definition}%[Embedding of $\Omega \subset \{0,1\}^d$ into $\bT$ from $\quadtree(\Omega)$]
\label{def:quadtree-map}
    For any subset $\Omega \subset \{0,1\}^d$ and any draw of $\bT$ generated from an execution of $\quadtree(\Omega)$ (in Figure~\ref{fig:DDquadtree-prelims}), we have the following:
    \begin{itemize}
        \item For every vector $x \in \{0,1\}^d$, there is a unique root-to-leaf path in $\bT$, given by the sequence of nodes $v_0(x), \dots, v_{L+1}(x)$, inductively defined by $v_{0}(x) = v_0$ and
        \[ v_{\ell}(x) \text{ is unique child $v_u$ of $v_{\ell-1}(x)$ with $x \in \bElms(v_u)$.}\]
        \item The mapping $\varphi \colon \{0,1\}^d \to \bT$ sends $x \in \{0,1\}^d$ to $v_{L+1}(x)$, and since the path for each $x \in \{0,1\}^d$ is unique, we abuse notation and associate $x \in \{0,1\}^d$ with its leaf $x = v_{L+1}(x) \in \bT$.
        \item The tree metric $(\bT, d_{\bT})$ is specified by the edge weights in $\bw(\cdot, \cdot)$, and for any $x, y \in \{0,1\}^d$, the distance $d_{\bT}(\varphi(x), \varphi(y))$ is the sum of edge-weights $\bw$ on the path from $\varphi(x)$ to $\varphi(y)$ in $\bT$.
    \end{itemize}
%    
%    
%    Fix any set $\Omega \subset \{0,1\}^d$ of $|\Omega| = n$ points, and let $\bT \sim \calT_\Omega$ be drawn via  Figure \ref{fig:DDquadtree-prelims}. We define the mapping $\varphi:\{0,1\}^d \to \bT$ via $\varphi(x) = \ell$, where  $\ell \in V(\bT)$ is the unique leaf node such that $\Elms(\ell) = \{x\}$
\end{definition}

%After fixing a the hash functions $\phi_1,\dots,\phi_L$ which are sufficent to define the value of $\varphi(x)$ for all $x \in \{0,1\}^d$, we will sometimes think of a point $x \in \{0,1\}^d$ interchangably with the leaf noode $\varphi(x) \in \bT$.
Whenever we generate $\bT$ from $\quadtree(\Omega)$ and we consider $x, y \in \Omega$, the edge weights along the path between $\varphi(x)$ and $\varphi(y)$ in $\bT$ do not depend on the parameter $\xi$. In particular, every node $v_{\ell}(x)$ on the root-to-leaf path in $\bT$ has $x \in \bElms(v_{\ell}(x)) \cap \Omega$, so that $\bw(v_{\ell-1}(x), v_{\ell}(x))$ falls into the first case in Step~\ref{en:quadtree-weight-setting}, where $\bw(\cdot,\cdot)$ is the expected distance of vectors sampled from $\bElms(\cdot) \cap \Omega$---we will call these edges ``data-dependent,'' since these weight depends on the vectors in $\Omega$ and may change when $\Omega$ changes. When one of $x$ or $y$ is not in $\Omega$, then at least one edge along the path $\varphi(x)$ to $\varphi(y)$ in $\bT$ falls in the second case of Step~\ref{en:quadtree-weight-setting} and has $\bw(\cdot,\cdot)$ set to $d / 2^{\ell} \cdot \xi$---we will call these edges ``data-independent.''

%Note that whenever $\Elms(v,\Omega) \neq \emptyset$, the weight of the edge $w(u,v)$ is \textit{data-dependent}, as it depends on the set of points in $\Omega$. However, the mapping $\varphi$ from the Hamming cube to vertices of $\bT$ is \textit{data-independent}.  For any $i \in [H]$, define $v_i(x)$ to be the unique vertex at depth $i$ in $\bT$ such that $x \in \Elms(v_i(x),\Omega)$. 

\begin{fact}[Distances in $\bT$ from $\quadtree(\Omega)$]\label{fact:tree-dist}
Let $\Omega \subset \{0,1\}^d$ be any subset and let $\bT$ be drawn from $\quadtree(\Omega)$. For any $x, y \in \{0,1\}^d$ and $\ell \in \{0,\dots, L+1\}$, we let $\Split_{\ell}(x,y)$ denote the indicator variable
\[ \Split_{\ell}(x,y) = \ind\{ v_{\ell}(x) \neq v_{\ell}(y) \}, \] %to be the indicator variable for the event that $v_i(x) \neq v_i(y)$. Then note that for any $x,y$ we have:
and note that we may write
\[ d_{\bT}(x, y) = \sum_{\ell=1}^{L+1} \Split_{\ell}(x, y) \cdot \Big( \bw(v_{\ell-1}(x), v_{\ell}(x)) + \bw(v_{\ell-1}(y), v_{\ell}(y)) \Big).\]
\end{fact}

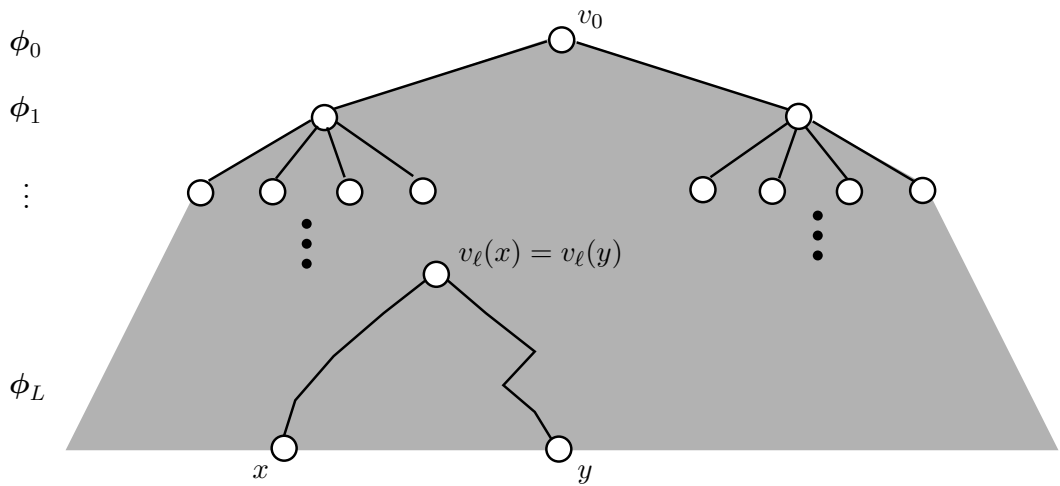
\begin{figure}[t]
\centering
\begin{picture}(370, 160)
\put(0,0){\includegraphics[width=.8\linewidth]{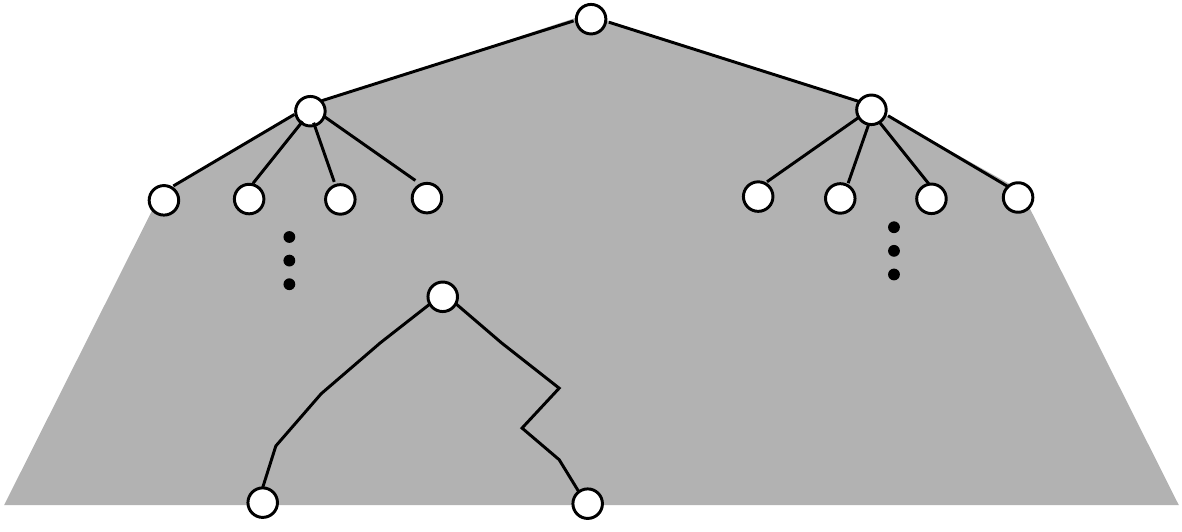}}
\put(195, 170){$v_0$}
\put(-20, 160){$\bphi_0$}
\put(-20, 135){$\bphi_1$}
\put(-15, 100){$\vdots$}
\put(-20, 30){$\bphi_{L}$}
\put(72, -2){$x$}
\put(195, -2){$y$} 
\put(150, 80){$v_{\ell}(x) = v_{\ell}(y)$}
\end{picture}
\caption{Tree Embedding $\bT$ Sampled from $\quadtree(\Omega)$. The root node is $v_0$ and the tree is generated by the maps $\bphi_0,\dots, \bphi_L$. Displayed are two vectors $x, y$ which map to the leaves of the tree, and their path (whose lowest common ancestor is $v_{\ell}(x) = v_{\ell}(y)$ is displayed. The distance $d_{\bT}(x,y)$ is given by the sum of weights along the path from $x$ to $v_{\ell}(x) = v_{\ell}(y)$, and then back to $y$.}\label{fig:tree}
\end{figure}

%\[ d_{\bT}(x,y) = \sum_{i=0}^{H-1}\Split_{i+1}(x,y) \cdot\Big(w( v_i(x) , v_{i+1}(x))+w( v_i(y) , v_{i+1}(y)) \Big)\]

Recall our goal in~Theorem~\ref{thm:dynamic-main}, the tree metric $\bT$ should be non-contracting for vectors $x, y \in \Omega$ while, at the same time, minimizing the expectation of $d_{\bT}(x, y)$. We use the following lemma from \cite{CJLW22}, which upper bounds the expected distance $d_{\bT}(x, y)$ whenever all weights along the path between $x$ and $y$ are data-dependent. Importantly, the lemma applies only to $x, y \in \Omega$, and extending it to vectors $x,y$ which are not necessarily in $\Omega$ will be the main technical challenge of the next sections. %Specifically, \cite{CJLW22} studied the same tree embedding, although they did not define weights for edges $(u,v)$ where $\Elms(v,\Omega) = \emptyset$. Thus, the following Lemma only applies to points $a,b \in \Omega$.

\begin{lemma}[Follows from Lemma 3.6 (with $i_0 = 0$) and Lemma 3.4 from \cite{CJLW22}]
\label{lem:cjlw}
    For any  set $\Omega \subseteq \{0,1\}^d$ of $m$ vectors, and any two $a, b \in \Omega$, we have that, whenever $\bT$ is generated from $\quadtree(\Omega)$, we have
    \[ \Ex_{\bT}\left[ d_{\bT}(a,b)\right] \leq \tilde{O}(\log(m) + \log(d) ) \cdot \|a - b\|_1 \]
 Moreover, we have $d_{\bT}(a,b) \geq \|a-b\|_1$ deterministically. 
\end{lemma}

We note that Lemma~\ref{lem:cjlw} immediately implies that a single draw of a tree metric $\bT$ from $\quadtree(\Omega)$ satisfies the distortion guarantees we desired: it is non-contracting for vectors in $\Omega$, and has a bounded expected expansion. In what follows, we will show how to maintain a data structure for $\bT$ dynamically, and for this purpose, it is useful to modify the way in which the ``data-dependent'' weights $\bw(v, v_u)$ are defined in Step~\ref{en:quadtree-weight-setting}. % as setting 
%\[ \bw(v, v_u) \eqdef \| \bc - \bc' \|_1 \qquad \text{ for $\bc \sim \bElms(v) \cap \Omega$ and $\bc' \sim \bElms(v_u) \cap \Omega$,}\]
%whenever $\bElms(v_u) \cap \Omega$, as opposed to evaluating the value expectation. In particular, we consider the following simple modification, which we give as a definition below to define notation.
\begin{definition}\label{def:mod-tree-weights}
    For any set $\Omega \subset \{0,1\}^d$, let $\bT$ be generated from an execution of $\quadtree(\Omega)$. We let $\bT'$ denote the tree metric whose vertex set, edge set, and mapping $\varphi \colon \{0,1\}^d \to \bT$ is the same as in $\bT$; however, we modify the weights as follows:
    \begin{itemize}
        \item For each node $v \in \bT'$, if $\bElms(v) \cap \Omega \neq \emptyset$, we sample what we call a representative $\bRep(v) \sim \bElms(v_u) \cap \Omega$.
        \item For each edge $(v, v_u) \in \bT'$ where $v$ is at depth $\ell$, we let
        \[ \bw'(v, v_u) = \left\{ \begin{array}{cc} \| \bRep(v) - \bRep(v_u) \|_1 & \bElms(v_u) \cap \Omega \neq \emptyset \\
                                                        d / 2^{\ell} \cdot \xi & \text{otherwise} \end{array} \right. .\]
    \end{itemize}
    We similarly consider the tree metric $(\bT', d_{\bT'})$, and we have
    \begin{align} 
    d_{\bT'}(x, y) = \sum_{\ell=1}^{L+1} \Split_{\ell}(x,y) \cdot \Big( \bw'(v_{\ell-1}(x), v_{\ell}(x)) + \bw'(v_{\ell-1}(y), v_{\ell}(y)) \Big). \label{eq:mod-tree-dist}
    \end{align}
\end{definition}

\paragraph{Data Structure for Dynamic, Data-Dependent Probabilistic Trees.}\label{en:ds-dynamic} We can now describe the data structure which maintains the tree $\bT'$, which samples $\bT$ from $\quadtree(\Omega)$ and uses the modified edge weights in~Definition~\ref{def:mod-tree-weights}. The data structure will maintain the following information:
\begin{itemize}
    \item We store the sampled functions $\bphi_0,\dots, \bphi_{L}$ (by storing the set of indices sampled for each $\ell \in \{0, \dots, L\}$), and note that it suffices to store the \emph{set} of indices samples, which has size at most $d$ always. This takes time $O(Ld)$ during initialization. 
    \item We also maintain the set $\Omega$, as well as the subtree of $\bT'$ of nodes $v$ for which $\bElms(v) \cap \Omega$ is non-empty. For each such node $v$, we maintain the set $\bElms(v) \cap \Omega$, as well as the sample $\bRep(v) \sim \bElms(v) \cap \Omega$. 
\end{itemize}
This completes the description of the data structure. Note that, each vector $x \in \Omega$ is naturally mapped to a leaf $\varphi(x)$ which may easily be found in $O(dL)$ time by walking down the (stored) subtree of $\bT'$. Given two leaves $\varphi(x)$ and $\varphi(y)$ for $x,y \in \Omega$, the required information is available to compute (\ref{eq:mod-tree-dist}) in time $O(dL)$. Then, when updating the set $\Omega$ by inserting or deleting a vector $x \in \{0,1\}^d$, we proceed by:
\begin{itemize}
    \item\label{en:insertion-dynamic} \textbf{Insertion:} We insert $x$ to $\Omega$ and find the leaf $\varphi(x)$, considering the root-to-leaf path given by nodes $v_{0}(x), \dots, v_{L+1}(x)$, where one may need to initialize new nodes if $v_{\ell}(x)$ was not stored in the stored subtree. For each node $v = v_{\ell}(x)$, with probability $1 / |\bElms(v) \cap \Omega|$ (note that $\Omega$ now includes one more vector), we update $\bRep(v)$ to $x$; otherwise, do not update $\bRep(v)$. If the data structure updates $\bRep(v)$, every vector in $\bElms(v) \cap \Omega$ has its weighted path modified and its change is reported.
    \item\label{en:deletion-dynamic} \textbf{Deletion:} We delete $x$ to $\Omega$ and find the leaf $\varphi(x)$, considering the root-to-leaf path given by nodes $v_{0}(x), \dots, v_{L+1}(x)$, where one may need to initialize new nodes if $v_{\ell}(x)$ was not stored in the stored subtree. For each node $v = v_{\ell}(x)$, if $\bRep(v)=x$, we update $\bRep(v)$ by re-sampling from $\bElms(v) \cap \Omega$ or removing $v$ if empty. If the data structure updates $\bRep(v)$, every vector in $\bElms(v) \cap \Omega$ has its weighted path modified and its change is reported.
\end{itemize}

\paragraph{Analysis.} From the metric perspective, the change from $\bT$ to $\bT'$ does not affect the distortion analysis. The fact that the embedding to $(\bT', d_{\bT'})$ is non-contracting follows from the triangle inequality since the path from $\phi(x)$ to $\phi(y)$ define some path of the form 
\[ v_{L+1}(x), v_{L}(x), \dots, v_{\ell}(x) = v_{\ell}(y), v_{\ell+1}(y), \dots, v_{L+1}(y), \]
and we have that, for any setting of the randomness, 
\begin{align*}
    d_{\bT'}(x, y) &= \sum_{j=\ell+1}^{L+1} \Big( \bw'(v_{j-1}(x), v_{j}(x)) + \bw'(v_{j-1}(y), v_{j}(y)) \Big) \\
                    &= \Big\| \bRep(v_{L+1}(x)) - \bRep(v_{L}(x)) \Big\|_1 + \Big\| \bRep(v_{L}(x)) - \bRep(v_{L-1}(x)) \Big\|_1 + \dots +\\
                    &\qquad + \Big\| \bRep(v_{\ell+1}(x)) - \bRep(v_{\ell}(y)) \Big\|_1 + \dots + \Big\| \bRep(v_{L}(y)) - \bRep(v_{L+1}(y))\Big\|_1 \\
                    &\geq \|x - y\|_1.
\end{align*}
In addition, the sampling procedure to modify the weights $\bw'$ is defined such that
\begin{align*}
    \Ex_{\bT'}\left[ d_{\bT'}(x, y) \right] = \Ex_{\bT}\left[ d_{\bT}(x, y)\right].
\end{align*}

\begin{claim}\label{cl:remain-uni}
    Consider any sequence of updates $u_1,\dots, u_k$ (specifying insertions and deletions of vectors) for the set $\Omega \subset \{0,1\}^d$. If $\bT'$ denotes the (randomized) tree maintained by the data structure, then, for every node $v \in \bT'$ where $\bElms(v) \cap \Omega$ is non-empty, the random variable $\bRep(v)$ is distributed as a uniform draw from $\bElms(v) \cap \Omega$.
\end{claim}

\begin{proof}
    We prove the claim by induction on the length of the sequence $u_1,\dots, u_k$. For $k=0$, $\Omega$ is empty therefore the claim is vacuously true. Assume for inductive hypothesis that for updates $u_1,\dots, u_k$, every non-empty $\bElms(v) \cap \Omega$ satisfies $\bRep(v)$ is uniformly distributed among $\bElms(v) \cap \Omega$. We consider the update $u_{k+1}$.
    \begin{itemize}
        \item \textbf{Insertion}: If $u_{k+1}$ is the insertion of a vector $x \in \{0,1\}^d$ to $\Omega$, we let $\Omega' = \Omega \cup \{x\}$. The data structure modifies the distribution of $\bRep(v_{\ell}(x))$ for $\ell=0,\dots, L+1$, and any other node $v'$ rest remain uniform over $\bElms(v') \cap \Omega' = \bElms(v') \cap \Omega$ by induction. For a node $v = v_{\ell}(x)$, the random variable $\bRep(v)$ is now updated to be \textbf{(i)} equal to $x$ with probability $1 / |\bElms(v) \cap \Omega'|$ and \text{(ii)} uniform over $\bElms(v) \cap \Omega$ by induction with the remaining probability. Thus, for any $y \in \bElms(v) \cap \Omega$, 
        \begin{align*}
            \Prx\left[ \bRep(v) = y\right] &= \left(1 - \frac{1}{|\bElms(v) \cap \Omega'|} \right) \cdot \frac{1}{|\bElms(v) \cap \Omega|} \\
            &= \left( \frac{|\bElms(v) \cap \Omega|}{|\bElms(v) \cap \Omega'|}\right) \frac{1}{|\bElms(v) \cap \Omega|} = \frac{1}{|\bElms(v) \cap \Omega'|}.
        \end{align*}
        \item \textbf{Deletion}: If $u_{k+1}$ is the deletion of a vector $x \in \{0,1\}^d$ from $\Omega$, we let $\Omega' = \Omega \setminus \{x\}$. The node of any $v$ whose value $\bRep(v)$ is not $x$ remains the same. By induction, the node was uniformly distributed over $\bElms(v) \cap \Omega$, and it is now uniformly distributed over $\bElms(v) \cap \Omega'$. Any node $v$ with $\bRep(v) = x$ is re-randomized, so uniform over $\bElms(v) \cap \Omega'$.
    \end{itemize}
    This completes the proof, as the draws of $\bphi_1,\dots, \bphi_{L+1}$, and hence the graph structure remains unchanged. The weights $\bw'$ depend on the draws of $\bRep(v)$, but these are uniform as needed.
\end{proof}

\begin{claim}
    Consider any fixed sequence of updates and consider a final update, the expected time of the update is $O(d L)$. 
\end{claim}

\begin{proof}
    First, we note that it takes $O(d L)$ time to find the root-to-leaf path of a vector $x \in \{0,1\}^d$ which is being added or removed from $\Omega$. Then, we note that for each depth $\ell$, on an insertion, the expected running time resulting from updates to the embeddings as a result from a weight-change is
    \[ \sum_{\ell=0}^{L} \frac{1}{|\bElms(v_{\ell}(x)) \cap \Omega|} \cdot |\bElms(v_{\ell}(x)) \cap \Omega| \cdot O(d) = O(dL). \]
    Similarly for deletions, the distribution of each $\bRep(v_{\ell}(x))$ is uniform among $\bElms(v_{\ell}(x)) \cap \Omega$, so that the probability that the deletion of $x$ means that any $\bRep(v_{\ell}(x))$ is updated is $1/|\bElms(v_{\ell}(x)) \cap \Omega|$---since this results in a change to $|\bElms(v_{\ell}(x)) \cap \Omega|$ weighted paths, the similar bound of $O(dL)$ follows.
\end{proof}

%% file: towards-hash-fam.tex
%!TEX root = main.tex

\section{Locality Sensitive Hash Family for $\EMD$}

We will now show how to construct data-dependent hash families for $\EMD$, assuming some technical lemmas which we will prove in the later sections. Formally, we consider the metric space whose objects are size-$s$ tuples, where each entry is a vector in $\R^d$; for any $p \in [1,2]$, the distance will be the Earth Mover's distance with ground metric $\ell_p$. We will argue by \textbf{(i)} reducing data-dependent hashing over $(\R^d, \ell_p)$ to that of the hypercube $\{0,1\}^d$, and then \textbf{(ii)} giving a data-dependent hashing scheme for $\EMD$ over the hypercube.

%%%%%%%%%%%

\subsection{Reduction to Data-Dependent LSH over the Hypercube}
%The main result of this section is to prove a reduction for Data-Dependent LSH over $(\R^d,\ell_1)$ to Data-Dependent LSH over the Hamming Cube $\{0,1\}^t$. Formally, we prove the following:

By Remark \ref{remark:ell1}, it suffices to consider data-dependent hashing for $\EMD$ over $(\R^d,\ell_1)$. In the following Lemma, we reduce the problem further to data-depedent hashing for $\EMD$ over the Hamming Cube $\{0,1\}^t$. 

%%%%%%%%%%%

\begin{lemma}[Reduction to the Data-Dependent Hashing on $\EMD_{s}(\{0,1\}^t)$]
    \label{lem:reduction-to-hypercube}
For any parameters $s,\tau  \geq 0$ and $c > 3$, $\delta \in (0,1)$:% suppose the following holds: %and any distribution $\mu$ over $\EMD_s(\R^d,\ell_1)$, then the followings holds:
\begin{itemize}    
    \item Suppose that there exists a data structure for data-dependent hashing over $\EMD_s(\{0,1\}^t)$ which is $(r, cr/3, p_1, p_2)$-sensitive for the parameter settings
    \[ t = \Theta(s^2c^2\log(1/\delta)) \qquad\text{and}\qquad r = \frac{t}{1.99c} \geq \omega(s),\]
    which has initialization time $I_{\sfh}(n)$ and query time $Q_{\sfh}(n)$.
    %For any distribution $\mu'$ over $\EMD_s(\{0,1\}^t)$, there is a data-dependent hash family $\Phi_1$ for $\EMD_s(\{0,1\}^t)$ is which is $(r, c r / 3, p_1, p_2)$-sensitive for $\mu'$, where $t = O(s^2 c^2 \log(1/\delta))$ and $r = t / (1.99 c)$. 
    \item Then, there exists a data structure for data-dependent hashing over $\EMD_s(\R^d,\ell_1)$ which is $(\tau, c\tau , p_1-\delta , p_2 + \delta)$-sensitive with initialization time $I_{\sfh}(n) + n \cdot \poly(sd)$ and query time $Q_{\sfh}(n) + \poly(sd)$.
\end{itemize}
\end{lemma}

%%%%%%

We prove the above lemma by showing how to use a locality-sensitive hash function to give a threshold embedding from $\EMD_s(\R^d, \ell_1)$ to $\EMD_s(\{0,1\}^t)$. Lemma~\ref{lem:reduction-to-hypercube} will then simply follow, and the remainder of the section proves the next lemma. % we first prove a intermediate proposition that reduces approximate near neighbor search over $\ell_1$ to approximate near neighbor search over the hamming cube $\{0,1\}^t$. 

\begin{lemma}\label{lem:reduction-to-hypercube-l1}
For any parameters $s,\tau,\geq 0$, $c > 3$, as well as $\delta \in (0,1)$, there exists a distribution $\Gamma$ over functions $f:\EMD_s(\R^d,\ell_1) \to \EMD_s(\{0,1\}^t)$, where $t = O(s^2 c^2\log(\delta^{-1} ))$ and $r = t/(1.99 c)$, such that the following holds: 
\begin{itemize}
    \item For every $x,y \in \EMD_s(\R^d,\ell_1) $ with $\EMD(x,y) \leq \tau$, $\EMD(\boldf(x),\boldf(y)) \leq r$ with probability at least $1-\delta$ over $\boldf \sim \Gamma$.
    \item For every $x,y \in \EMD_s(\R^d,\ell_1) $ with $\EMD(x,y) \geq  c \tau$, $\EMD(\boldf(x),\boldf(y)) \geq  cr/3$ with probability at least $1-\delta$ over $\boldf \sim \Gamma$.
\end{itemize}
Furthermore, there exists a data structure which maintains a draw of $\boldf \sim \Gamma$ and supports queries of $\boldf(x)$ for $x \in \EMD_s(\R^d, \ell_1)$ which has an initialization and query time of $\poly(sd)$.
\end{lemma}
%First, for $p \in (1,2]$, we map $X$ to $(\R^{d'},\ell_1)$ with distortion $(1 \pm \eps)$ and $d'= (?)$....Thus in what follows we can assume $p=1$.

 \begin{proof}[Proof of Lemma~\ref{lem:reduction-to-hypercube} assuming Lemma~\ref{lem:reduction-to-hypercube-l1}]
%We prove the proposition for $p=1$.    
We define the data structure which maintains a draw from a hash family $\Phi_2$ over $\EMD_s(\R^d, \ell_1)$ which will be $(r, cr, p_1, p_2)$-sensitive. In order to initialize the data structure upon receiving a distribution $\mu$, we perform the following:
\begin{enumerate}
    \item First, we initialize a data structure of Lemma~\ref{lem:reduction-to-hypercube-l1} in order to maintain a sample $\boldf \sim \Gamma$. We query the function $n$ times for each point in the support of $\mu$, and consider the push-forward distribution $\mu'$ over $\EMD_s(\{0,1\}^t)$ given by sampling $\bx \sim \mu$ and outputting $\boldf(\bx)$ (which is also supported on at most $n$ points).
    \item Then, we initialize the data structure to consider the hash family $\Phi_1$, which is assumed to be $(r, cr/3, p_1, p_2)$-sensitive for $\mu'$, and will maintain a sample $\bphi_1 \sim \Phi_1$.
    \item The data structure then maintains a draw $\bphi_2 \sim \Phi_2$ which is defined by letting $\bphi_2(x) = \bphi_1(\boldf(x))$.
\end{enumerate}
Note that the running time of the initialization is $n \cdot \poly(sd) + I_{\sfh}(n)$. Then, upon receiving a query $x \in \EMD_s(\R^d, \ell_1)$, we may query $\bphi_2(x)$ by first querying $\boldf(x)$ and then querying $\bphi_1(\boldf(x))$ in time $\poly(sd) + Q_{\sfh}(n)$. This completes the description of the data structure and we now check the two properties of Definition~\ref{def:data-dep} for $\Phi_2$.
%Let $\Phi_1$ be the assumed hash family for $\EMD_{s}(\{0,1\}^t)$ which is $(r, cr/3, p_1, p_2)$-sensitive for . Also let $\Gamma$ be the family of mapping $f:\EMD_s(\R^d,\ell_p) \to \EMD_s(\{0,1\}^t)$ from Lemma \ref{lem:reduction-to-hypercube-l1}. We define a draw $\bphi_2 \sim \Phi_2$ via \textbf{(1)} drawing $\bphi_1 \sim \Phi_1$, \textbf{(2)} drawing $\boldf \sim \Gamma$, and outputting $\bphi_2(x) = \bphi_1(\boldf(x))$. 
%We check both properties of Definition~\ref{def:data-dep} for $\Phi_2$. 

Consider two points $x,y \in \EMD_s(\R^d, \ell_1)$ with $\EMD(x,y) \leq \tau$. By Lemma~\ref{lem:reduction-to-hypercube-l1}, with probability $1-\delta$ over $\boldf \sim \Gamma$,  $\EMD(\boldf(x),\boldf(y)) \leq r$, for $r = t/(1.99c)$. Therefore, the guarantee that $\Phi_1$ is $(r, cr/3, p_1, p_2)$-sensitive means that $\bphi_2(x) = \bphi_2(y)$ with probability at least $p_1 - \delta$ as needed. 

For the second property of Definition~\ref{def:data-dep}, consider any points $x,y \in \EMD_s(\R^d, \ell_1)$ with $\EMD(x,y) \geq c \tau$. By Lemma \ref{lem:reduction-to-hypercube-l1}, with probability at least $1-\delta$ over $\boldf \sim \Gamma$, $\EMD(\boldf(x),\boldf(y)) \geq cr/3$. Therefore,
%It follows that, even taking $\by \sim \mu$ as distributed according to $\mu$, we have 

\begin{equation*}
    \begin{split}
      \Prx_{\substack{\bphi_2 \sim \Phi_2 \\ \by \sim \mu}}\left[ \begin{array}{c} \EMD(x,\by)\geq  c\tau,\\
 \bphi_2(x) = \bphi_2(\by) \end{array}\right] & \leq    \Ex_{\boldf \sim \Gamma}\left[ \Prx_{\substack{\bphi_1 \sim \Phi_1 \\ \by \sim \mu}}\left[ \begin{array}{c} \EMD(\boldf(x),\boldf(\by))\geq  cr/3,\\
 \bphi_1(\boldf(x)) = \bphi_1(\boldf(\by)) \end{array}\right]\right] + \delta \\
 & \leq p_2 + \delta.
    \end{split}
\end{equation*}
 \end{proof}

\subsubsection{Proof of Lemma~\ref{lem:reduction-to-hypercube-l1}} 
We now give the proof of Lemma~\ref{lem:reduction-to-hypercube-l1}, where we begin by introducing a hash family over $d$-dimensional $\ell_1$ space $(\R^d,\ell_1)$. Even though the LSH properties require that there be a gap between close and far pairs of points, we will use the stronger property that, for close enough points, the probability they are divided by this hash family is proportional to their distance. 
\begin{proposition}\label{prop:l1hash}
For any threshold $\sfR>0$, there is a hash family $\calG$ mapping $(\R^d,\ell_1)$ to a universe $U$ with the property that, for any $a,b \in \R^d$ with $\|a-b\|_1 \leq \sfR$,
\[       \Prx_{\bg \sim \calG}\left[ \bg(a) \neq \bg(b)\right] = \frac{\|a-b\|_1}{d \sfR},   \]
and for any $a, b \in \R^d$ with $\|a-b\|_1 > \sfR$,
\[       \Prx_{\bg \sim \calG}\left[ \bg(a) \neq \bg(b)\right] > \frac{1}{d}.   \]
In addition, a hash function from the family may be stored in $O(1)$ space and evaluated on an element in $O(1)$ time.
\end{proposition}
\begin{proof}
We describe how to generate a draw $\bg \sim \calG$. We sample a random coordinate $\bi^* \sim [d]$, and we impose a randomly shifted by offset $\bb\sim [\sfR]$ grid of size length $\sfR$ for the coordinate $\bi^*$. The hash function $\bg$ is then given by:
    \[ \bg(x) =  \left\lceil \frac{x_{\bi^*} + \bb}{\sfR}\right\rceil.  \]
    We have 
    \begin{equation*}
        \begin{split}
             \Prx_{\bg \sim \calG}\left[ \bg(a) \neq \bg(b) \right] &= \frac{1}{d}\sum_{i=1}^d \Prx_{\bg \sim \calG}\left[ \bg(a) \neq \bg(b) \mid \bi^* = i\right] = \frac{1}{d}\sum_{i=1}^d  \min\left\{1, \frac{|a_i - b_i|}{\sfR} \right\}. %  = \frac{\|a-b\|_1}{d\sfR}\\
        \end{split}
    \end{equation*}
    If $\|a-b\|_1 \leq R$, then $|a_i - b_i| \leq \sfR$ for all $i \in [d]$, and the above sum is $\|a-b\|_1 / (d\sfR)$ as desired. Otherwise, if any single $i$ satisfies $|a_i - b_i| \geq R$, then the above at least $1/d$ as desired.
    \end{proof}

    \begin{corollary}\label{cor:l1hashmodified}
For any threshold $\sfR>0$, there is a hash family $\calG^d$ mapping $(\R^d,\ell_1)$ to a universe $U$ with the property that for any $a,b \in \R^d$ with $\|a-b\|_1 \leq \sfR$ we have
\[ \frac{\|a-b\|_1}{2\sfR} \leq    \Prx_{\bg \sim \calG^d}\left[ \bg(a) \neq \bg(b)\right] \leq \frac{\|a-b\|_1}{\sfR}   \]
And if $\|a-b\|_1 > \sfR$, $\prb{\bg \sim \calG^d}{ \bg(a) \neq \bg(b)} > 1/2$. A hash function from the family may be stored in $O(d)$ space and may be evaluaated on an element in $O(d)$ time.
\end{corollary}
\begin{proof}
We take the hash family $\mathcal{G}$ used above in Proposition \ref{prop:l1hash}, and output its concatenation $d$ times. In particular, we use the hash family 
\[\mathcal{G}^d = \Big\{\bg \; \Big| \; h(x) = (\bg_1(x),\bg_2(x),\dots,\bg_d(x)), \; \bg_1,\dots,\bg_d \sim \mathcal{G}\Big\}\] 
Note then that if $\|a-b\|_1 \leq \sfR$,
\[   \Prx_{\bg \sim \calG^d}\left[ \bg(a) \neq \bg(b)\right]  = 1 - \left(1 - \frac{\|a-b\|_1}{d\sfR}\right)^d\]
Using the inequalities that \textbf{(1)} $(1-x)^d \geq 1-xd$ for $d \geq 1,x \geq -1$, \textbf{(2)} $(1-x)^d \leq 1/(1+xd)$ for $x \in [-1,1/d)$ and $d \geq 0$, and \textbf{(3)} that $x/2 \leq 1-1/(1+x)$ for $x \in [0,1]$, we have
\[ \frac{\|a-b\|_1}{2\sfR} \leq 1-  \frac{1}{1+\frac{\|a-b\|_1}{\sfR}}\leq 1 - \left(1 - \frac{\|a-b\|_1}{d\sfR}\right)^d\ \leq \frac{\|a-b\|_1}{\sfR}\]
For the second case, using \textbf{(1)} and Proposition \ref{prop:l1hash} yields the desired result. The running time bounds then simply follow.
\end{proof}

 We now define the distribution $\calH_{t}^d(\sfR)$ over functions $\R^d \to \{0,1\}^t$. Then, we will show how to specify $\sfR$, and set the distribution $\Gamma$ in Lemma~\ref{lem:reduction-to-hypercube-l1} to be the function which applies $\boldf \sim \calH_t^d(\sfR)$ to each of the $s$ elements in $\EMD_s(\R^d,\ell_1)$. We let $\boldf \sim \calH_{t}^d(\sfR)$ be given by:
 \begin{enumerate}
     \item For each $i \in [t]$, we sample a function $\bg_i \sim \calG^d$ (from Corollary~\ref{cor:l1hashmodified}), as well as a random function $\bchi_i \colon \Z \to \{0,1\}$.
     \item For each $i \in [t]$, we let $\boldf_i(x) = \bchi_i(\bg_i(x))$, and then we set
     \[ \boldf(x) = (\boldf_1(x), \boldf_2(x), \dots, \boldf_t(x)) \in \{0,1\}^t.\]
 \end{enumerate}

\paragraph{Data Structure Guarantees for $\boldf \sim \Gamma$.} Even though we described the above procedure which samples $\boldf$, the data structure does not explicitly sample the random functions $\bchi_i$. Rather, the data structure instantiates $t$ data structures which sample $\bg_i \sim \calG^d$ from Corollary~\ref{cor:l1hashmodified}; but does not explicitly generate the random functions $\bchi_i \colon \Z \to \{0,1\}$---rather, it generates them in a lazy fashion. Whenever there is a query point $x \in \EMD_s(\R^d, \ell_1)$, it queries the $t$ data structures to compute $\bg_i(x)$ and checks whether or not it had already generated (the random draw of) $\bchi_i(\bg_i(x))$; it uses it if it did, and generates it and stores it if it did not. This way, the total running time of the initialization procedure is $O(dt)$ and the running time of the querying $x \in \EMD_s(\R^d, \ell_1)$ is $O(dts)$, which becomes $\poly(sd)$ for the setting of $t$.
 
\paragraph{Expansion and Contraction Guarantees for $\boldf \sim \Gamma$.} First, note that, by Proposition~\ref{prop:l1hash}, if $\|a-b\|_1 \leq \sfR$, any $i \in [t]$ satisfies 
\begin{align}
\frac{\|a-b\|_1}{4 \sfR} \leq \Prx_{\boldf_i}\left[ \boldf_i(a) \neq \boldf_i(b)\right] \leq 
    \frac{\|a-b\|_1}{2 \sfR}. \label{eq:prob-1-a}
    \end{align}
Moreover, if $ \|a-b\|_1 > \sfR$ then 
\begin{align}\Prx_{\boldf_i}\left[ \boldf_i(a) \neq \boldf_i(b)\right] \geq  
    \frac{1}{4}. \label{eq:prob-2-a}
\end{align}
The following two claims are simple applications of Chernoff Bounds, using the expectations obtained from expressions (\ref{eq:prob-1-a}) and (\ref{eq:prob-2-a}). %Proposition \ref{prop:l1hash}
\begin{claim}\label{lem:boundNear}
    Fix any $a,b \in \R^d$ with $\|a-b\|_1 \leq \sfR$, and fix any $\eps,\delta \in (0,1)$. Then, as long as $t = O(\eps^{-2}\log \delta^{-1})$ (for sufficiently high constant),
    \[   \pr{t \left(\frac{\|a-b\|_1}{4 \sfR} - \eps\right) \leq \|\boldf(a) - \boldf(b)\|_1  \leq t \left(\frac{\|a-b\|_1}{2\sfR} + \eps\right)  } > 1 - \delta \]
\end{claim}

\begin{claim}\label{lem:boundFar}
    Fix any $a,b \in \R^d$ with $\|a-b\|_1 > \sfR$, and fix any $\eps,\delta \in (0,1)$. Then, as long as $t = O(\eps^{-2}\log \delta^{-1})$ (for sufficiently high constant),
    \[ \pr{ \|\boldf (a) - \boldf(b)\|_1  < t\left(\frac{1}{4}  - \eps\right)  } <\delta \]
\end{claim}

Given a function $\boldf \sim \calH_{t}^d(\sfR)$, there is a natural application of that function $s$ many times, in order to map $(\R^d)^s \to (\{0,1\}^t)^s$. Namely, given a tuple $x \in \EMD_s(\R^d,\ell_1)$ of $s$ elements in $\R^d$, we may apply $\boldf$ to each of the $s$ elements individually and obtain a tuple of $s$ elements in $\{0,1\}^t$. Thus, we will abuse notation and denote $\boldf \colon \EMD_s(\R^d,\ell_1) \to \EMD_s(\{0,1\}^t)$ given by mapping
\[ x = (x_1,\dots, x_s) \in \EMD_s(\R^d, \ell_1)\quad \mathop{\longmapsto}^{\boldf}\quad (\boldf(x_1), \dots, \boldf(x_s)) \in \EMD_s(\{0,1\}^t).\]
The (randomized) function $\boldf \colon \EMD_s(\R^d, \ell_1) \to \EMD_s(\{0,1\}^t)$ defines our desired mapping of Lemma~\ref{lem:reduction-to-hypercube-l1}. We now fix $\eps_0,r,t,\sfR$ as follows:
\begin{equation}\label{eqn:rtR}
    \eps_0 = \Theta(\eps/(s c)), \qquad    t = O\left(\eps_0^{-2}\log \left(\frac{s}{\delta}\right)\right), \qquad r =  \frac{t}{1.99c}, \qquad \ \sfR = c \cdot \tau.
\end{equation}
%We also fix the dataset $X \subset \EMD_s(\R^d, \ell_1)$ of $|X| =n$ points. 

%In what follolw, we define $\mathcal{E}$ to be the event that, for all $x,y \in X$ and any $a \in x, b \in y$, if $\|a-b\|_1 \leq R$,  then $\|\boldf(a) - \boldf(b)\|_0  = t \cdot\frac{\|a-b\|_1}{2d \sfR} \pm \eps t$, and otherwise $\|\boldf (a) - \boldf(b)\|_0  \geq 1/d - \eps$. By the above Lemmas and a union bound over all $\binom{sn}{2}$ pairs $(a,b)$, we have: $\pr{\mathcal{E}} \geq 1-1/\poly(sn)$. 

To verify the two properties of Lemma~\ref{lem:reduction-to-hypercube-l1}, consider any $x, y \in \EMD_s(\R^d,\ell_1)$ and let $\pi \colon [s] \to [s]$ denote the bijection realizing $\EMD(x,y)$ (i.e., that which satisfies $\EMD(x,y) = \sum_{i=1}^s \|x_i - y_{\pi(i)}\|_1$).
%In what follows, for any $x,y \in \EMD_s(\R^d,\ell_1)$, define the mapping $\pi_{x,y}:[s] \to [s]$ be the bijection which attains $\EMD(x,y)$, namely $\EMD(x,y) = \sum_{i=1}^s \|x_i - y_{\pi_{x,y}(i)}\|_1$. 
%Also, for any $x,y \in \EMD_s(\R^d,\ell_1)$, 
Consider the event $\calbE(x,y)$ (defined with respect to the randomness in the draw of $\boldf$) that, for all vectors $a\in x, b \in y$, the events in Claims~\ref{lem:boundNear} and \ref{lem:boundFar} hold with failure probability $\delta / s^2$, i.e., whenever $\|a-b\|_1 \leq \sfR$,
\begin{align*}
    t\left( \frac{\|a-b\|_1}{4\sfR} - \eps_0 \right) \leq \|\boldf(a) - \boldf(b)\|_1 \leq t \left( \frac{\|a-b\|_1}{2\sfR} + \eps_0 \right),
\end{align*}
and whenever $\|a-b\|_1 > \sfR$, $\|\boldf(a) - \boldf(b)\|_1$ is at least $t(1/4-\eps_0)$. Note that, by a union bound over $s^2$ pairs of elements (and the settings of $t$ and $\eps_0$), event $\calbE(x,y)$ holds with probability at least $1-\delta$. 
%if \textbf{(1)} $\|a-b\|_1 \leq \sfR$ then  $\left(\frac{\|a-b\|_1}{4 \sfR} - \eps_0\right) \leq \|\boldf(a) - \boldf(b)\|_1  \leq t \left(\frac{\|a-b\|_1}{2\sfR} + \eps_0\right)$, and if \textbf{(2)} $\|a-b\|_1 > \sfR$ then $\|\boldf(a) - \boldf(b)\|_1 \leq t/4 - t\eps_0$. By Lemmas \ref{lem:boundNear} and \ref{lem:boundFar} using failure probability $\Theta(\delta/s^2)$, and a union bound over $O(s^2)$ pairs,  we have $\pr{\mathcal{E}(x,y)} \geq 1-\delta$. We can now prove the following Lemmas.

\begin{lemma}\label{lem:reduction-to-hypercube-l1-upper} Let $t,\sfR,r$ be fixed as in Equation \ref{eqn:rtR}. For any $x,y \in \EMD_s(\R^d,\ell_1)$ with $\EMD(x,y) \leq \tau$, with probability at least $1-\delta$ over the draw of $\boldf \sim \calH_t^d(\sfR)$, we have $\EMD(\boldf(x),\boldf(y)) \leq r$.
\end{lemma}
\begin{proof} 
First, note that if $\EMD(x,y) \leq \tau$, clearly it holds that for all  $i \in [s]$ we have $\|x_i-y_{\pi(i)}\|_1 \leq \tau \leq \sfR$.  Thus, whenever $\calbE(x,y)$ holds,
\begin{equation*}
    \begin{split}
          \sum_{i=1}^s \|\boldf(x_i) - \boldf(y_{\pi(i)})\|_1 &   \leq t \left(\sum_{i=1}^s \frac{\|x_i - y_{\pi(i)}\|_1}{2\sfR} + \eps_0\right)  \leq t \left(\frac{\tau}{2\sfR} + \eps_0 s\right)  \leq \frac{t}{ 1.99 c}  = r
    \end{split}
\end{equation*}
Thus, $\EMD(\boldf(x),\boldf(y)) \leq r$ with probability at least $1-\delta$ as desired. 
%\Raj{Can save factor of $s$ here by doing Chernoff over all $x_i,y_{\pi(i)}$ together.}
\end{proof}

\begin{lemma} \label{lem:reduction-to-hypercube-l1-lower} Let $t,\sfR,r$ be fixed as in Equation \ref{eqn:rtR}. For any $x,y \in X$ with $\EMD(x,y) \geq c \tau$, with probability at least $1-\delta$ over the draw of $\boldf \sim \calH_t^d(\sfR)$, we have $\EMD(\boldf(x),\boldf(y)) \geq  c r/3$.
\end{lemma}
\begin{proof}
Consider the case $\calbE(x,y)$ holds. Then, for any matching $\sigma:[s] \to [s]$,
    \begin{equation*}
    \begin{split}
          \sum_{i=1}^s \|\boldf(x_i) - \boldf(y_{\pi(i)})\|_1 &   \geq t \sum_{i=1}^s \left(\min\left \{ \frac{1}{4}, \;  \frac{\|x_i - y_{\pi(i)}\|_1}{4\sfR} \right\} - \eps_0\right)   \geq t \left( \min \left \{ \frac{1}{4}, \frac{\EMD(x,y) }{4\sfR}    \right \} - s  \eps_0\right) \\
     &   \geq t/4 - st \eps_0 \geq t / 4.001 \geq cr/3.
    \end{split}
\end{equation*}
%as needed.
\end{proof}

The proof of Lemma \ref{lem:reduction-to-hypercube-l1} the follows immediately from Lemmas \ref{lem:reduction-to-hypercube-l1-upper} and \ref{lem:reduction-to-hypercube-l1-lower}.

\subsection{Three Crucial Ingredients for LSH for $\EMD$ over the Hypercube}\label{sec:crucial-ingred}

From now on, we will build a data-dependent hash family for $\EMD$ on size-$s$ tuples in the hypercube with respect to the Hamming distance, where the dimension $d = \poly(s)$ and our required threshold $r = \omega(s)$. We will refer to ``points'' as the size-$s$ tuples of vectors in $\{0,1\}^d$, and ``elements'' to the points in $\{0,1\}^d$ which will be in the tuples. Hence, a ``point'' is in reference to a point in the metric space $\EMD_s(\{0,1\}^d)$, and each point is a size-$s$ tuple of ``elements'' in $\{0,1\}^d$. We will need three ingredients, where it is useful to keep in mind the ``$p_1$''- and ``$p_2$''-properties for hash families in Definition~\ref{def:data-dep}; the ``$p_1$''-property guarantees that a query and its near neighbor oftentimes collide, and the ``$p_2$''-property guarantees far-apart points sampled from $\mu$ oftentimes do not collide.
\begin{enumerate}
\item The first ingredient is Lemma~\ref{lem:ingred-1}, which specifies a sequence of hash families whose hash functions maps points (i.e., size-$s$ tuples of $\{0,1\}^d$) to buckets. These hash families are parametrized by a so-called ``level'' $\ell$, and we will let $\ell$ vary among $L$ possible levels,\footnote{These levels will correspond to the depth of the quadtree embeddings.} for $L = \Theta(\log d)$.
As we will see, these data-independent hash families have a good ``$p_1$''-property---for each of these hash families, the probability that we divide any two points (i.e., size-$s$ tuples of $\{0,1\}^d$) is at most proportional to $\EMD(\cdot,\cdot)$, and this allows us to say that close points collide often (since their $\EMD(\cdot,\cdot)$ is small).
\item The second ingredient is Definition~\ref{def:ingred-2}, a point being ``locally-dense'' with respect to a distribution $\mu$ over points. We define ``local-density'' as all hash families defined in Lemma~\ref{lem:ingred-1} failing to have the ``$p_2$''-property; however, the important consequence (and the reason for the name ``local-density'') will be the following (see Lemma~\ref{lem:dense-conse}). A point $x$ (which recall is a tuple of vectors $a_1,\dots, a_s$ in $\{0,1\}^d$) will be locally-dense if ``many'' of its elements $a_i$ have a non-trivial fraction of ``nearby'' elements from points in $\mu$ (where ``many'' and ``nearby'' vary such that the sum-of-nearest-neighbors---known as the Chamfer distance---to a random subset of $\mu$ is bounded). By Ingredient 1, the hash families $\calH(\tau, \ell)$ always satisfy the ``$p_1$''-property, but not necessarily the ``$p_2$''-property; locally-dense points with respect to $\mu$ are exactly those whose desired ``$p_2$''-property with all $\calH(\tau, \ell)$'s does not hold. 
%Roughly speaking, the definition of local-density will say that none of the $T$ possible hash families (parametrized by the level $\ell \in [T]$) succeed at dividing many points from the underlying distribution $\mu$. This means that, for these points, the data-independent hash families do not have a desired ``$p_2$''-property.
\item The final ingredient will be the data-dependent hash family which fills in the gap. For all points, the data-dependent hash family always has the desired ``$p_2$''-property, but it may not have the ``$p_1$''-property. However, we will prove that the data-dependent hash family has the property ``$p_1$''-property whenever a point is locally-dense.
\end{enumerate}
In the remainder of the section, we formally state the lemmas which capture the three ingredients and show how these imply a data-dependent hash family. 

\begin{restatable}{lemma}{firstingred}\label{lem:ingred-1}
For any parameter $\ell \in \{0, \dots, L \}$ and any $\tau > 0$, we define a hash family $\calH(\tau, \ell)$ (in Definition~\ref{def:data-ind-hash}). The hash family $\calH(\tau, \ell)$ satisfies that, for any two $x, y \in \EMD_s(\{0,1\}^d)$,
\begin{align*}
\Prx_{\bh \sim \calH(\tau,\ell)}\left[ \bh(a) \neq \bh(b)\right] &\leq \frac{\EMD(x,y)}{\tau}.
\end{align*}
In addition, there is a data structure which maintains a draw of $\bh \sim \calH(\tau, \ell)$ while supporting queries of $\bh(x)$ in initialization and query time $O(sd)$.
\end{restatable}

\begin{restatable}{definition}{secondingred}\label{def:ingred-2}
Let $\mu$ denote a distribution supported on $\EMD_s(\{0,1\}^d)$. For parameters $\alpha, \tau > 0$, we say that a point $x \in \EMD_s(\{0,1\}^d)$ is $(\alpha, \tau)$-locally-dense with respect to $\mu$ if for all $\ell \in \{0,\dots, L\}$, %\Tian{Locally-dense also depends on the distribution $\mu$, should we make $\mu$ one of the parameters?}
\begin{align*}
\Prx_{\substack{\bu \sim \mu \\ \bh \sim \calH(\tau, \ell)}}\left[ \bh(x) = \bh(\bu)\right] \geq \alpha. 
\end{align*}
\end{restatable}

\begin{restatable}{lemma}{thirdingred}\label{lem:ingred-3}
Let $\mu$ denote a distribution supported on $\EMD_s(\{0,1\}^d)$, and fix any $\alpha, \tau > 0$. Then for
%and let $\sfD > 1$ denote a parameter which is set to 
%^\[ \sfD = \tilde{O}\left(\log\frac{sd}{\delta \alpha}\right)\right). \]
%\tilde{\Theta}\left(\log\left(\frac{sd}{\delta \alpha}\right) \cdot \log\log\left(\frac{sd}{ \alpha}\right)\right). \]
any $\gamma > 0$ and $\delta \in (0, 1)$, there exists a hash family $\calH(\mu, \tau, \gamma, \delta)$ with the following properties: 
\begin{itemize}
\item \textbf{\emph{Close Points Collide}}: For any pair of points $x, y \in \EMD_{s}(\{0,1\}^d)$. If $x$ is $(\alpha, \tau)$-locally-dense, then
\[ \Prx_{\bh \sim \calH(\mu,\tau, \gamma,\delta)}\left[ \bh(x) \neq \bh(y) \right] \leq  \frac{ \EMD(x,y)  }{\gamma} \cdot   \lambda   \cdot\left(1+ \log \left( \frac{\tau + s}{\EMD(x,y)} + 1 \right)\right) \]
Where $\lambda = C_1 \log \left(\frac{sd}{\delta \alpha}\right) \left( \log \log \left(\frac{sd}{\delta \alpha}\right)\right)^{C_2}$ for absolute constants $C_1,C_2$. %\raj{would be good to know the exact value of $C_2$.}

\item \textbf{\emph{Far Points Separate}}: For any pair of points $x, y \in \EMD_s(\{0,1\}^d)$,
\[ \Prx_{\bh \sim \calH(\mu,\tau,\gamma,\delta)}\left[ \bh(x) = \bh(y) \right] \leq \exp\left( - \frac{\EMD(x, y)}{\gamma}\right) + \delta.\]
\end{itemize}
In addition, there is a data structure which maintains a draw $\bh \sim \calH(\mu, \tau,\gamma, \delta)$ while supporting queries of $\bh(x)$ which has initialization time $n \cdot \poly(sd/\alpha)$ (where $\mu$ is supported on $n$ points) and query time $\poly(sd)$.
\end{restatable}

\subsubsection{Main Theorems for Data Dependent Hashing and Nearest Neighbor Search}
With the above ingredients set in place, we are ready to state the data-dependent hash family for $\EMD$. The remainder of the section is devoted to proving the main theorem below.

\begin{theorem}[Data-Dependent Hashing for $\EMD$]\label{thm:data-dep-hashing} Fix any $0 < p_2 < p_1 < 1$. There exists a data structure for data-dependent hashing with a $(r, cr, p_1,p_2)$-sensitive family for $\EMD_s(\{0,1\}^d)$ where $r > s$ for an approximation $c > 1$ which is
%\[ c = O\left( \dfrac{\log(3/p_2)}{1-p_1} \cdot \log\left(\frac{sd}{(1-p_1)\cdot p_2^2}\right) \cdot \log \log \left(\frac{sd}{(1-p_1)\cdot p_2}\right)\right). \]
%\raj{How were we stating the exact value of $c$, which depends on $\gamma$ and therefore the (old value of $D$), when we were hiding $\poly \log \log$'s in the $\tilde{O}(\cdot)$ notation for $D$? Can we replace it with the below:}
\[  c = \tilde{O}\left( \frac{1}{1-p_1} \cdot \log\left(\frac{1}{p_2}\right) \cdot \log\left( \frac{sd}{p_2}\right)\right)\]

The data structure has initialization time $I_{\sfh}(n) \leq n \cdot \poly(sd/((1-p_1)p_2))$ and query time $\poly(sd)$.
%In addition, there is a data structure which maintains a draw of $\bh \sim \calD$ while supporting queries to $\bh(x)$ which has initialization time $n \cdot \poly(sd/((1-p_1)p_2))$ (where $\mu$ is supported on $n$ points) and query time $\poly(sd)$.
\end{theorem}

Our main result, for Data-Dependent LSH for $\EMD_s(\R^d,\ell_p)$ for any $p \in[1,2]$ (Theorem \ref{thm:LSH-main}) follows from combining Theorem \ref{thm:data-dep-hashing} with Lemma \ref{lem:reduction-to-hypercube}.
Setting $p_1 = 1-\eps$ and $p_2 = \Theta(1)$ in Theorem \ref{thm:data-dep-hashing}, and then applying Theorem \ref{thm:hashing-to-nn}, we obtain our main result on nearest neighbor search under the Earth Mover's Distance. 
\begin{theorem}[Approximate Nearest Neighbor Search for $\EMD$]\label{thm:ann-main}
    For any $s, d \in \N$, $p \in [1,2]$ and $\eps \in (0,1)$, there exists a data structure for approximate nearest neighbor search over $\EMD_s(\R^d, \ell_p)$ with approximation $c = \tilde{O}( \frac{\log s}{\eps}) $ satisfying the following guarantees:
    \begin{itemize}
        %\item \emph{\textbf{Approximation}}: The approximation factor $c$ is $\tilde{O}(\log s) / \eps$.
        \item \emph{\textbf{Preprocessing Time}}: The data structure preprocesses a dataset $P$ of $n$ points in $\EMD_s(\R^d, \ell_p)$ in time $n^{1+\eps} \cdot \poly(sd \eps^{-1})$.
        \item \emph{\textbf{Query Time}}: For a vector $q \in \EMD_s(\R^d, \ell_p)$, we output a $c$-approximate nearest neighbor of $q$ in $P$ in time $n^{\eps} \cdot \poly(sd)$. 
    \end{itemize}
\end{theorem}

\subsection{The Hash Family $\calD$ and Proof of Theorem~\ref{thm:data-dep-hashing}}  \label{sec:glue}

%{\color{blue} Erik: Would be useful to make another pass here. Since it looks hard to read.}

Now that we have stated all of the preliminary ingredients, we show how to construct the data-dependent hash family $\calD$ stated in Theorem \ref{thm:data-dep-hashing}. Let $\lambda = \tilde{O}\left(\log(\frac{sd}{\delta \alpha})\right)$ be the parameter defined in Lemma \ref{lem:ingred-3}.
We now instantiate the following parameters
\[ \tau = \frac{4 \cdot (L+1) \cdot r}{1 - p_1}, \qquad \alpha = \frac{(1 - p_1) \cdot p_2 }{6}, \qquad \gamma = \frac{\left(\lambda \log\left( \frac{20(L+1)}{1-p_1 } \right)  +\log (\frac{1}{1 - p_1})\right) \cdot \tau}{ L+1} \] \[ \delta = \frac{p_2}{3}, \qquad c = \log\left(\frac{3}{p_2}\right) \cdot \frac{\gamma}{r}\]
To sample sample a hash function $\bh \sim \calD$, we first sample a hash functions $\bh_0,\dots, \bh_L$, where $\bh_{\ell} \sim \calH(\tau, \ell)$ for each $\ell \in \{0, \dots, L \}$, and we next sample a hash function $\bh_* \sim \calH(\mu,\tau, \gamma,\delta)$. In order to evaluate $\bh$ on a point $z \in \EMD_s(\{0,1\}^d)$, we first check whether there exists an index $\ell \in \{0,\dots, L\}$ for which %\raj{Describe algorithmically how we do this check?}
\begin{align} 
\Prx_{\bu \sim \mu}\left[\bh_{\ell}(z) = \bh_{\ell}(\bu) \right] \leq \frac{p_2}{3}. \label{eq:ell-def}
\end{align}
If so, then we define $\bell(z)$ to be the smallest index $\ell \in \{0,\dots, L \}$ where (\ref{eq:ell-def}) holds. If no such index exists, we set $\bell(z) = *$. The final hash function $\bh \sim \calD$ then evaluates:
\[ \bh(z) = \left(\bell(z), \bh_{\bell(z)}(z) \right).\]

\paragraph{Running Time for Initializing and Querying $\bh \sim \calD$.} For the initialization, we must initialize and sample $\bh_{\ell} \sim \calH(\tau, \ell)$ for each $\ell \in \{0,\dots, L\}$, as well as $\bh_{*} \sim \calH(\mu, \tau, \gamma, \delta)$. The initialization as well as query time for the draws to $\bh_{\ell} \sim \calH(\tau, \ell)$ take time $O(sd)$ (from Lemma~\ref{lem:ingred-1}). From Lemma~\ref{lem:ingred-3}, the initialization time of $\bh_{*}$ takes time $n \cdot \poly(sd/\alpha)$, which is $n \cdot \poly(sd/\eps)$ as claimed. The time to query a single $\bh_{*}$ is $\poly(sd)$.

It remains to show how to compute $\bell(z)$, as this determines which hash function to evaluate in Equation \ref{eq:ell-def}. We proceed as follows during the initialization. First, we draw the hash functions $\bh_0,\dots, \bh_L$ and apply them to the dataset of all points $x$ in the support $\text{supp}(\mu)$ of $\mu$ (taking time $n \cdot \poly(sd)$). Then, we compute, for each bucket, the probability mass under $\mu$ which lies in that bucket (which is simply the number of dataset points hashed to that bucket, divided by $n$). Maintaining this additional information allows one to quickly determine the value of $\bell(z)$, for any $z$. Thus, the total query time is $(L+1) \cdot \poly(sd) + \poly(sd)$.

\paragraph{Analysis.} We first upper bound the probability that points which are close have different hash values. Suppose $x, y \in \EMD_s(\{0,1\}^d)$ with $\EMD(x, y) \leq r$, and notice that in order for $\bh(x) \neq \bh(y)$, there must exists a hash function (either $\bh_{\ell}$ for $\ell \in \{0,\dots, L\}$ or $\bh_{*}$) where they disagree (Otherwise, if all the hash functions agree, it follows that all $\bell(.)$ values agree, thus the composite hash functions $\bh(.)$ agree). Suppose first that $x$ is $(\alpha,\tau)$-locally-dense. Then, we may apply Lemma~\ref{lem:ingred-1} and the first item of Lemma~\ref{lem:ingred-3} to say
\begin{align*}
\Prx_{\bh \sim \calD}\left[ \bh(x) \neq \bh(y) \right] &\leq (L+1) \cdot  \sup_{\ell} \Prx_{\bh_{\ell} \sim \calH(\tau, \ell)}\left[ \bh_{\ell}(x) \neq \bh_{\ell}(y) \right] + \Prx_{\bh_{*} \sim \calH(\mu,\tau,\gamma,\delta)}\left[ \bh_{*}(x) \neq \bh_{*}(y)\right]
\end{align*}
where we can bound the first term by
\begin{align*}
(L+1) \cdot  \sup_{\ell} \Prx_{\bh_{\ell} \sim \calH(\tau, \ell)}\left[ \bh_{\ell}(x) \neq \bh_{\ell}(y) \right] &\leq (L+1) \cdot \frac{r \cdot (1-p_1)}{4 \cdot (L + 1) \cdot r} \leq \frac{1 - p_1}{4}.
\end{align*}

%{\color{red} Raj: working on this part}

For the remaining term, via Lemma~\ref{lem:ingred-3}  we have

\[\Prx_{\bh_{*} \sim \calH(\mu,\tau,\gamma,\delta)}\left[ \bh_{*}(x) \neq \bh_{*}(y)\right] \leq\frac{ \EMD(x,y)  }{\gamma} \cdot   \lambda   \cdot\left(1+ \log \left( \frac{\tau + s}{\EMD(x,y)} + 1 \right)\right) \]
Using that $\tau > 4r \geq 4s$ (by assumption of the Theorem statement), that $r > \EMD(x,y)$, and we can upper bound $(1+\log(z + 1))$ by $\log(4z)$ whenever $z \geq 1$, we have
%$we have $\left(1+ \log \left( \frac{\tau + s}{\EMD(x,y)} + 1 \right)\right) \leq \log \left( \frac{4 \tau}{\EMD(x,y)} \right)$, thus

\begin{align*}
\Prx_{\bh_{*} \sim \calH(\mu,\tau,\gamma,\delta)}\left[ \bh_{*}(x) \neq \bh_{*}(y)\right] &\leq \frac{(L+1) \cdot  \EMD(x,y)  }{\tau \cdot \log\left( \frac{20(L+1)}{1-p_1 } \right) } \cdot  \log\left( \frac{5 \tau }{\EMD(x,y)} \right) \\
 &\leq \frac{  \EMD(x,y) \cdot (1-p_1) }{ 4 r \cdot \log\left( \frac{20(L+1)}{1-p_1 } \right) } \cdot  \log\left( \frac{20 r (L+1) }{\EMD(x,y) \cdot (1-p_1)} \right) \\
  &\leq \frac{  \EMD(x,y) \cdot (1-p_1) }{ 4 r } \cdot \log\left( \frac{r }{\EMD(x,y) } \right) \\
 & \phantom{hi } \qquad +  \frac{  1-p_1 }{ 4  \cdot \log\left( \frac{20(L+1)}{1-p_1 } \right) } \cdot \log\left( \frac{20(L+1)}{1-p_1 } \right) \\
   &\leq \frac{   1-p_1 }{ 2} \\
\end{align*}
which concludes that for $x,y$ with $\EMD(x,y) \leq r$, and such that $x$ is $(\alpha,\tau)$ locally dense, we have:
\[\Prx_{\bh \sim \calD}\left[ \bh(x) \neq \bh(y) \right]\leq \frac{3(1-p_1)}{4}.\]

On the other hand, suppose that $\EMD(x,y) \leq r$ and $x$ is not $(\alpha,\tau)$-locally-dense, and let $\ell_0$ denote the smallest index which certifies that $x$ is not $(\alpha,\tau)$-locally-dense (recall Definition~\ref{def:ingred-2}). 
Then, whenever $\bh(x) \neq \bh(y)$, one of the two cases must occur:
\begin{enumerate}
\item At least one of $\bell(x)$ or $\bell(y)$ lies in $\{0,\dots, L\}$. First observe that whenever $\bh_\ell(x) = \bh_\ell(y)$ for all $\ell \in \{0,\dots,L\}$, it must be the case that $\ell(x) = \ell(y)$, since whether (\ref{eq:ell-def}) holds for a point $x$ at level $\ell$ is a deterministic function of the hash bucket that $x$ lands in under $\bh_\ell$. If $\ell(x) = \ell(y)$, then clearly $\bh_{\ell(x)}(x) \neq \bh_{\ell(y)}(y)$. Thus, for this case to occur, it must be that $\bh_\ell(x) \neq \bh_\ell(y)$ for at least one $\ell \in \{0,\dots,L\}$. 
 %When this occurs, we must have that there exists some $\ell \in \{0,\dots, L\}$ where $\bh_{\ell}(x) \neq \bh_{\ell}(y)$. If all are equal, then since one of $\bell(x), \bell(y) \in \{0,\dots, L\}$ and is smallest, either $x$ or $y$ was hashed by $\bh_{\ell}$ where (\ref{eq:ell-def}) holds, but the other point was not. 
As before, we can bound this probability by a union bound over the $L$ levels:
\begin{align*}
\Prx\left[ \begin{array}{c} \bh(x) \neq \bh(y) \\ \{ \bell(x) , \bell(y) \} \neq * \end{array} \right] \leq (L+1) \cdot \sup_{\ell} \Prx_{\bh_{\ell} \sim \calH(\tau, \ell)}\left[ \bh_{\ell}(x) \neq \bh_{\ell}(y) \right] \leq \frac{1 -p_1}{4}.
\end{align*}
\item Both $\bell(x) = \bell(y) = *$. In this case, we can upper bound the probability that $\bell(x) \neq \ell_0$ by Markov's inequality. Namely, Definition~\ref{def:ingred-2} implies that the expectation, over the draw of $\bh_{\ell_0} \sim \calH(\tau, \ell_0)$, of the probability over $\bu \sim \mu$ that $\bh_{\ell_0}(\bu) = \bh_{\ell_0}(x)$ is at most $\alpha$; but when we sampled $\bh_{\ell_0} \sim \calH(\tau, \ell_0)$, the fact that $\bell(x) = *$ implies that this probability was larger than $p_2/3$. Thus:
\begin{align*}
\Prx\left[ \begin{array}{c} \bh(x) \neq \bh(y) \\ \bell(x) = \bell(y) = * \end{array} \right] \leq \Prx_{\bh_{\ell_0} \sim \calH(\tau, \ell_0)}\left[ \Prx_{\bu \sim \mu}\left[ \bh_{\ell_0}(\bu) = \bh_{\ell_0}(x)\right] \geq \frac{p_2}{3} \right] \leq \frac{3 \alpha}{p_2} \leq \frac{1 - p_1}{2}.
\end{align*}
\end{enumerate}
By a union bound, the probability $\bh(x) \neq \bh(y)$ is at most $1-p_1$. That concludes the condition that points which are closer than $r$ in $\EMD$ tend to collide. We now upper bound the probability that points which are far collide. Suppose that $x \in \EMD_s(\{0,1\}^d)$ and we think of sampling $\by \sim \mu$ and $\bh \sim \calD$, and evaluating the probability that they are separated. Then, we can upper bound the probability that $\bh(x) = \bh(\by)$ and $\EMD(x,\by) \geq c \cdot r$ by first sampling $\bh_1,\dots, \bh_{L}$ and $\bh_{*}$, and then sampling $\by$. Note that 
\begin{align*}
\Prx_{\substack{\bh \sim \calD \\ \by \sim \mu}}\left[\begin{array}{c} \EMD(x,\by) \geq c\cdot r \\ \bh(x) = \bh(y) \end{array}\right] &\leq \Prx\left[ \begin{array}{c} \bell(x) \in \{0,\dots, L\} \\ \bh_{\bell(x)}(x) = \bh_{\bell(x)}(\by) \end{array} \right] + \Prx\left[ \bh_{*}(x) = \bh_{*}(\by) \mid \EMD(x, \by) \geq c \cdot r\right]\\
			&\leq \frac{p_2}{3} + \exp\left( - \frac{cr}{\gamma}\right) + \frac{p_2}{3} \leq p_2.
\end{align*}
where the first inequality uses, by definition of $\bell(x) \in \{0,\dots, L\}$, that whenever this occurs, the probability over $\by \sim \mu$ that $\bh_{\bell(x)}(x) = \bh_{\bell(x)}(\by)$ is at most $p_2/3$, and otherwise, if $\bell(x) = *$, we apply the second item of Lemma~\ref{lem:ingred-3} with our setting of $c$. 

%% file: data-ind-hash-and-SampleTree/st_main.tex
\section{Ingredients 1 and 2: the Hash Family $\calH(\tau, \ell)$ and Locally-Dense Points}\label{sec:local-density-and-hashing}

 In this section, we give the first ingredient and prove Lemma~\ref{lem:ingred-1}. We will first define the hash family $\calH(\tau, \ell)$, and derive the main consequence of locally-dense points.

\input{data-ind-hash-and-SampleTree/data-ind-hash}

\input{data-ind-hash-and-SampleTree/structure-of-locally-dense-elements}

\section{Ingredient 3: $\SampleTree$ and Proof of Lemma~\ref{lem:ingred-3}}
\label{sec:7}
In this section, we show the proof of Lemma~\ref{lem:ingred-3}, which gives the final ingredient of the data-dependent hashing scheme for $\EMD_s(\{0,1\}^d)$, and concludes the proof of Theorem~\ref{thm:data-dep-hashing}. We reproduce the lemma below and proceed by first describing the $\SampleTree$ embedding, and then giving two lemmas which state the expansion and contraction properties of $\SampleTree$ that give rise to Lemma~\ref{lem:ingred-3}. The remainder of the section is then devoted to showing the expansion and contraction lemmas.

%{\color{red} Erik: The lemma below looks hard to read since there is a $\Theta$ in the statement of the inequality---may be better to have all asymptotic notation be w.r.t. $\sfD$?}

\thirdingred*

\input{data-ind-hash-and-SampleTree/st-embedding-and-hashing}

\subsection{Proof of Lemma~\ref{lem:sample-tree-expan}}\label{sec:proofofexpand}
\input{data-ind-hash-and-SampleTree/st-expansion}

\subsection{Proof of Lemma~\ref{lem:sample-tree-contr}}
\input{data-ind-hash-and-SampleTree/st-contraction}

%% file: data-ind-hash-and-SampleTree/data-ind-hash.tex
\subsection{Hash Family $\calH(\tau, \ell)$ and Proof of Lemma~\ref{lem:ingred-1}}\label{sec:proof-ingred-1}

As in Subsection~\ref{sec:crucial-ingred}, the term ``points'' is used to denote size-$s$ tuples of vectors in $\{0,1\}^d$. Each of the $s$ vectors in $\{0,1\}^d$ is referred to as an ``element'' of the point. We let $L = O(\log d)$, and we will refer to $\ell \in \{0,\dots, L\}$ as the ``levels.''

\begin{definition}[The Hash Family $\calH(\tau, \ell)$]\label{def:data-ind-hash}
    For $\tau > 0$ and $\ell \in \{0,\dots, L\}$, the hash family $\calH(\tau, \ell)$ is specified by the following sampling procedure. A draw of a hash function $\bh \sim \calH(\tau, \ell)$ proceeds by:
    \begin{enumerate}
        \item First, we sample $\bphi \sim \calH_{2^{\ell}}$ (as in Section~\ref{sec:Dynamic Embedding}) by sampling $2^{\ell}$ coordinates $\bi_1,\dots, \bi_{2^{\ell}} \sim [d]$ and letting $\bphi \colon \{0,1\}^d \to \{0,1\}^{2^\ell}$ be
        \[ \bphi(a) = (a_{\bi_1}, a_{\bi_2}, \dots, a_{\bi_{2^{\ell}}}) \in \{0,1\}^{2^{\ell}}. \]
        \item Then, for each $u \in \{0,1\}^{2^{\ell}}$ and each $k \in [s]$, we let $\bC_{u, k} \sim \Ber(d / ( \tau 2^{\ell+1}))$. 
%\raj{Did we define $\Ber()$? Also this should probably be $\Ber(\max\{1,d / ( \tau 2^{\ell+1})\})$}
%\footnote{Note that naively generating all of these Bernoulli random variables $\bC_{u, k}$ would incur exponential-in-$d$ time (which is a cost we want to avoid). Even though it is useful for the analysis to consider all of these Bernoulli random variables as being pre-generated, algorithmically, they are only generated when accessed (see Fact~\ref{fact:compute-h}).}
        \item For a point $x \in \EMD_s(\{0,1\}^d)$, and $u \in \{0,1\}^{2^{\ell}}$ and $k \in [s]$, we let $\bchi(x, u, k) \in \{0,1\}$ be
        \[\bchi(x, u, k) = \ind\{ \text{at least $k$ elements $a \in x$ satisfy $\bphi(a) = u$} \}. \]
        With those definitions, we let
        \[ \bh(x) = \left( \bC_{u, k} \cdot \bchi(x, u, k) : u \in \{0,1\}^{2^{\ell}}, k \in [s] \right) \in \{0,1\}^{\{0,1\}^{2^{\ell}} \times [s]}. \]
    \end{enumerate}
\end{definition}

\paragraph{Data Structure Guarantees for $\bh \sim \calH(\tau, \ell)$.} It is important to note that, for each $x \in \EMD_s(\{0,1\}^d)$, we may compute $\bh(x)$ in time $O(sd)$. This is because, even though the vector $\bh(x)$ lies in hypercube of dimensionality as high as $s \times 2^{2^L}$, the vectors $\bh(x)$ have at most $s$ non-zero coordinates. We may identify the at-most-$s$ non-zero entries of $\bchi(x, u, k)$ in $O(sd)$ time, and we can generate and store the corresponding Bernoulli random variables $\bC_{u,k}$ with a constant-time overhead per access. If we always store the values of $\bC_{u,k}$ generated after each query $\bh(x)$ (since there are at most $s$ such Bernoulli random variables being generated), we may implement evaluations to $\bh \sim \calH(\tau, \ell)$ as a data structure, whose initialization and query time is $O(sd)$. %which we encapsulate in the following fact.

%\begin{fact}\label{fact:compute-h}
%For any $\tau > 0$ and $\ell \in \{0,\dots, L\}$, there exists a data structure which maintains a sample $\bh \sim \calH(\tau, \ell)$, and has the following guarantees:
%\begin{itemize}
%    \item \emph{\textbf{Queries}}: A query is specified by a point $x \in \EMD_{s}(\{0,1\}^d)$ and in time $O(sd)$ (which is linear in the size of $x$), outputs sparse representation of $\bh(x)$.
%    \item \emph{\textbf{Space Complexity}}: If the data structure has received $m$ queries, the total space complexity of the data structure is $O(d + t sd)$.
%\end{itemize}
%\end{fact}

%\paragraph{Probability of Separation Guarantees for.} With Definition~\ref{def:data-ind-hash}, we may simply prove the Lemma~\ref{lem:ingred-1}, which establishes the ``$p_1$''-property of $\calH(\tau, \ell)$. 

\firstingred*

\begin{proof}
In order for $\bh(x) \neq \bh(y)$, there must exists at least one $u \in \{0,1\}^{2^{\ell}}$ and $k \in [s]$ where $\bC_{u, k} = 1$, and $\bchi(x, u, k) \neq \bchi(y, u, k)$. Thus, we can upper bound the probability that $\bh(x) \neq \bh(y)$ by 
\begin{align*}
\Prx_{\bh \sim \calH(\tau, \ell)}\left[ \bh(x) \neq \bh(y) \right] &= \Prx\left[ \exists u, k \text{ s.t } \bchi(x_0, u, k) \neq \bchi(y, u, k) \text{ and } \bC_{u, k} = 1\right]\\
&\leq \frac{d}{\tau \cdot 2^{\ell+1}} \Ex_{\bphi \sim \calH_{2^{\ell}}}\left[\sum_{u \in \{0,1\}^{2^{\ell}}} \sum_{k=1}^s \ind\left\{ \bchi(x, u, k) \neq \bchi(y, u, k)\right\} \right].
\end{align*}
Suppose we let $a_1,\dots, a_s \in \{0,1\}^d$ denote the elements of $x$, and $b_1,\dots, b_s \in \{0,1\}^d$ denote the elements of $y$; where we re-index the elements so that $a_i$ is matched to $b_i$ in the matching which realizes $\EMD(x, y)$. Then, we deterministically satisfy (for all choices of $\bphi$),
\[ \sum_{u \in \{0,1\}^{2^{\ell}}} \sum_{k=1}^s \ind\{ \bchi(x, u, k) \neq \bchi(y, u, k) \} \leq \sum_{i=1}^s 2 \cdot \ind\{ \bphi(a_i) \neq \bphi(b_i) \}.  \]
We may thus upper bound
\begin{align*}
\Prx_{\bh \sim \calH(\tau, \ell)}\left[ \bh(x) \neq \bh(y)\right] \leq \frac{2d}{\tau \cdot 2^{\ell+1}} \sum_{i=1}^s \Prx_{\bphi \sim \calH_{2^{\ell}}}\left[ \bphi(a_i) \neq \bphi(b_i) \right] \leq \frac{d}{\tau \cdot 2^{\ell}} \sum_{i=1}^s \frac{2^{\ell} \|a_i - b_i\|_1}{d} = \frac{\EMD(x, y)}{\tau}.
\end{align*}
\end{proof}

%% file: data-ind-hash-and-SampleTree/structure-of-locally-dense-elements.tex
\subsection{Locally-Dense Points}

In this section, we derive the main consequence of locally-dense points, which will become a crucial ingredient in Lemma~\ref{lem:ingred-3}. We will let $\mu$ denote a distribution over points in $\EMD_s(\{0,1\}^d)$ and refer to the hash families $\calH(\tau, \ell)$ defined in Definition~\ref{def:data-ind-hash}. Recall the definition of locally-dense points (which we reproduce below).

\secondingred*

For any two subsets $x$ and $z$ of vectors in $\{0,1\}^d$, we let the Chamfer distance from $x$ to $z$ be given by
\[ \textsf{Chamfer}(x, z) = \sum_{a \in x} \min_{b \in z} \|a - b\|_1. \]
Notice that $\textsf{Chamfer}(\cdot,\cdot)$ is an asymmetric measure ($\textsf{Chamfer}(x,z)$ is not equal to $\textsf{Chamfer}(z, x)$). The main consequence of the above definition is that a point $x \in \EMD_s(\{0,1\}^d)$ which is locally-dense will have a small Chamfer distance to the union of elements in a (relatively) small sample from $\mu$. 
\begin{lemma}\label{lem:dense-conse}
Let $\mu$ denote a distribution supported on $\EMD_{s}(\{0,1\}^d)$. If, for parameters $\alpha, \tau > 0$, a point $x \in \EMD_{s}(\{0,1\}^d)$ is $(\alpha, \tau)$-locally dense with respect to $\mu$, then as long as $m = \omega(\log(sd)/\alpha)$,
\begin{align*}
\Ex_{\by_1,\dots, \by_m \sim \mu}\left[ \textsf{\emph{Chamfer}}\left(x, \bigcup_{i=1}^m \by_i\right) \right] \leq (\tau + s) \cdot \polylog(s d / \alpha).
\end{align*}
\end{lemma}

\subsubsection{Proof of Lemma~\ref{lem:dense-conse}}

We consider a fixed point $x$ which is $(\alpha, \tau)$-locally dense with respect to $\mu$, and we let $V_s \subset \EMD_{s}(\{0,1\}^d)$ denote the subset of points $y \in \EMD_{s}(\{0,1\}^d)$ which satisfy
\begin{align*}
    \Prx_{\bh \sim \calH(\tau, \ell)}\left[ \bh(x) = \bh(y)\right] \geq \alpha / 2.
\end{align*}
Notice that, from an averaging argument, $\Prx_{\by \sim \mu}\left[ \by \in V_s \right] \geq \alpha / 2$. So fix $y \in V_s$, and for any $\rho > 0$ let $\calE(y, \rho)$ denote the subset of elements of $x$ which do not contain any element of $y$ within distance $\rho$, i.e.,
\[ \calE(y, \rho) = \left\{ a \in x : \forall b \in y, \|a - b\|_1 > \rho \right\}. \]

\begin{claim}\label{cl:few-empty}
For any $y \in V_s$ and $\rho \geq 1$, the set $\mathcal{E}(y, \rho)$ has size 
\[ |\calE(y, \rho)| \leq \max\left\{ \frac{12 \tau \cdot \log^2(4s/\alpha)}{\rho}, s \right\} \]
\end{claim}

\begin{proof}
First, note that the above statement is trivial once $\rho = d$ since all elements are in $\{0,1\}^d$, so as long as $y$ is non-empty, there cannot be any elements in $x$ whose distance to all of $y$ is larger than $d$. In addition, the set $\calE(y, \rho)$ always contains at most $s$ elements, since it is a subset of $x \in \EMD_{s}(\{0,1\}^d)$. Thus, we consider $\rho$ between $1$ and $d$; here, we may consider the smallest setting of $\ell$ in $\{0, \dots, L \}$\footnote{We note that $L$ is a large enough factor of $O(\log d)$ so that $2^{\ell}$ may be as high as $3d \log(s/\alpha) / \rho$, since $s$ and $d$ are polynomially related, and $\alpha$ will be set to a small enough constant.} which satisfies
\begin{align} 
\frac{3d \cdot \log(s/\alpha)}{\rho} \leq 2^{\ell} \leq \frac{6d \cdot \log(s/\alpha)}{\rho},  \label{eq:ell-val}
\end{align}
and recall that the hash function $\bh \sim \calH(\tau,\ell)$, after sampling $\bphi \sim \calH_{2^{\ell}}$, will sub-sample, for each $u \in \{0,1\}^{2^{\ell}}$ and $k \in [s]$ an indicator from $\Ber(d/ (\tau 2^{\ell+1}))$ and consider the values of $\bchi(y, u, k)$ and $\bchi(x, u, k)$ for $u, k$ with $\bC_{u,k} = 1$. Thus, suppose that we define the random set $\bZ$ which depends on the draw $\bphi \sim \calH_{2^{\ell}}$, given by
\begin{align*}
\bZ = \Big\{ (u, k) \in \{0,1\}^{2^{\ell}} \times [s] : \bchi(x, u, k) \neq \bchi(y, u, k)\Big\},
\end{align*}
and note that in order for $\bh(x) = \bh(y)$, we must have avoided setting $\bC_{u, k} = 1$ for $(u,k) \in \bZ$---otherwise, the coordinate corresponding to $(u,k)$ in $\bh(x)$ differs from that of $\bh(y)$. Since $y \in V_s$, 
\begin{align*}
\frac{\alpha}{2} \leq \Prx_{\bh \sim \calH(\tau, \ell)}\left[ \bh(x) = \bh(y) \right] &= \Ex_{\bh \sim \calH(\tau,\ell)}\left[ \left(1 - \frac{d}{\tau \cdot 2^{\ell+1}} \right)^{|\bZ|}\right] \leq \Ex_{\bh \sim \calH(\tau, \ell)}\left[ \exp\left( - \frac{d}{ \tau \cdot 2^{\ell+1}} \cdot |\bZ|\right)\right] \\
	&\leq \Prx_{\bphi \sim \calH_{2^{\ell}}}\left[ |\bZ| \leq \frac{\tau \cdot 2^{\ell+1} \cdot \log(4/\alpha)}{d}\right] + \frac{\alpha}{4},
\end{align*}
which implies that $\bZ$ must have size smaller than $\tau \cdot 2^{\ell+1} \log(4/\alpha) / d$ with probability at least $\alpha / 4$. Suppose, for the sake of contradiction, that the set $\calE(y, \rho)$ has at size
\[ |\calE(y, \rho)| > \dfrac{12\tau \cdot \log^2(4s/\alpha)}{\rho} \geq \frac{\tau \cdot 2^{\ell+1} \cdot \log(4/\alpha)}{d}, \]
where the second inequality is by the upper bound in (\ref{eq:ell-def}). Then, we may lower bound $|\bZ|$ by considering the elements from $\calE(y, \rho)$ which never collide with any element from $y$ under $\bphi$. In particular, if all elements of $\calE(y, \rho)$ have no elements from $y$ colliding, these contribute to entries $(u, k)$ of $\bZ$. Thus,
\begin{align*}
\Prx_{\bphi \sim \calH_{2^{\ell}}}\left[ |\bZ| > \frac{\tau \cdot 2^{\ell+1}\cdot \log(4/p_2)}{d}\right] &\geq \Prx_{\bphi \sim \calH_{2^{\ell}}} \left[ \forall a \in \calE(y, r_{\ell}), \forall b \in y : \bphi(a) \neq \bphi(b) \right] \\
			&= 1 - \Prx_{\bphi \sim \calH_{2^{\ell}}}\left[ \exists a \in \calE(y, \rho), \exists b \in y : \bphi(a) = \bphi(b) \right]\\
			&\geq 1 - |\calE(y, \rho)| \cdot s \cdot \left(1 -\frac{\rho}{d}\right)^{2^{\ell}} \geq 1 - s^2 \cdot \exp\left(- \frac{2^{\ell} \cdot \rho}{d} \right)\\
            &\geq 1 - o(\alpha),
\end{align*}
by setting of $\ell$ (the lower bound in (\ref{eq:ell-def})). This is a contradiction, so we obtain a bound on $|\calE(y, \rho)|$.
\end{proof}

We now conclude the proof of Lemma~\ref{lem:dense-conse}. Below, we write $\ba \sim x$ to mean sampling an element $a$ from $x \in \EMD_s(\{0,1\}^d)$ uniformly at random. Then, we may write
\begin{align*}
\Ex_{\by_1,\dots, \by_m \sim \mu}\left[ \textsf{Chamfer}\left(x, \bigcup_{i=1}^m \by_i\right)\right] &= s \int_{\rho:0}^{\infty} \Prx_{\substack{\by_1,\dots, \by_m \sim \mu \\ \ba \sim x}}\left[ \forall i \in [m] : \ba \in \calE(\by_i, \rho) \right]  \\
            &\leq \int_{\rho:1}^{d} \frac{12 \tau \cdot \log^2(4s/\alpha)}{\rho} \cdot d\rho + sd \left( 1 - \frac{\alpha}{2}\right)^{m} + s \\
            &\leq 12 \tau \cdot \log^2(4s/\alpha) \cdot \log d + sd \left(1 - \frac{\alpha}{2} \right)^{m} + s
\end{align*}

%% file: data-ind-hash-and-SampleTree/st-embedding-and-hashing.tex
\subsection{The $\SampleTree$ Embedding and Hash Family Construction}\label{sec:sample-tree-def}

In this section, we specify the construction of the hash family $\calH(\mu, \tau, \gamma,\delta)$. We will do so by first specifying the $\SampleTree$ embedding, and then concatenating it with a locality-sensitive hash function in $\ell_1$. In particular, we first describe an algorithm, $\SampleTree$, which takes as input a distribution $\mu$ supported on $\EMD_{s}(\{0,1\}^d)$ and a parameter $m$ (which, as per Lemma~\ref{lem:dense-conse}, will be set to $\omega(\log(sd)/\alpha)$), and outputs a weighted tree $\bT$ from an execution to $\quadtree$ in Figure~\ref{fig:DDquadtree-prelims}.

%It is well-known that the Earth Mover's distance over weighted tree metrics embed isometrically into $\ell_1$, and that there exists locality-sensitive hash functions in $\ell_1$. Therefore, it suffices for us to analyze the following algorithm, and the following lemma.

\begin{figure}[h!]
\begin{framed}
    \noindent Subroutine $\SampleTree(\mu, m)$
    
    \begin{flushleft}
        \noindent {\bf Input:} A distribution $\mu$ supported on $\EMD_{s}(\{0,1\}^d)$, and positive integer $m$.
        
        \noindent {\bf Output:} A weighted tree $\bT$ obtained from an execution of $\quadtree$.
        
        \begin{enumerate}
                \item Take $m$ random i.i.d. samples $\by_1,\dots, \by_m \sim \mu$, and let $\bOmega = \bigcup_{i=1}^m \by_i \subset \{0,1\}^d$. 
            \item Let $\hat{\bOmega} \subset \{0,1\}^d$ denote the set of (at most $ms(d+1)$ elements)
            \[ \hat{\bOmega} = \nbr(\bOmega) \eqdef \left\{ b' \in \{0,1\}^d : \exists b \in \bOmega, \|b - b'\|_1 \leq 1 \right\}.\]
            \item Run and return $\quadtree(\hat{\bOmega}, \xi)$ (Figure \ref{fig:DDquadtree-prelims}) where we set $\xi = \Theta(\log (msd/\delta))$.
        \end{enumerate}
    \end{flushleft}
\end{framed}
\caption{The $\SampleTree$ Algorithm.}\label{fig:quadtree-embed}
\end{figure}

To describe $\SampleTree$ algorithm, we introduce the notations of neighborhood. For any element $e \in \{0,1\}^d$, let the neighborhood of $e$ be
$$\nbr(e) \coloneqq \left\{p \in \{0,1\}^d \mid \|e - p\|_1 \leq 1\right\}$$
We extend the above notation so that we can apply it to a set of elements as we do in Figure~\ref{fig:quadtree-embed}. For any set $\Omega \subseteq \{0,1\}^d$, let the neighborhood of $\Omega$ be 
$$\nbr(\Omega) \coloneqq \left\{ p \in \{0,1\}^d \mid \exists e \in \Omega, \|e - p\|_1 \leq 1 \right\}$$ 
The $\SampleTree$ sub-routine (in Figure~\ref{fig:quadtree-embed}) specifies a tree metric $\bT$, and a natural association of any element $a\in \{0,1\}^d$ to a leaf in $\bT$ (each element $a \in \{0,1\}^d$ maps to a unique leaf in $\bT$, since the final hash function $\bphi_{L+1}\colon \{0,1\}^d \to \{0,1\}^d$ is set to the identity). Thus, we let $\EMD_s(\bT)$ denote the metric space on size-$s$ tuples of leaves in $\bT$. We let $x,y \in \EMD_s(\bT)$, with $x = (x_1,\dots, x_s)$ and $y = (y_1,\dots, y_s)$ where $x_1,\dots, x_n, y_1,\dots, y_n$ are leaves in $\bT$, and
\[ \EMD_{\bT}(x,y) = \min_{\substack{\pi \colon [s]\to[s] \\ \text{bijection}}} \sum_{i=1}^s d_{\bT}(x_i, y_{\pi(i)}),\]
where $d_{\bT}(\cdot,\cdot)$ denotes the length of the shortest path between two leaves. We thus have the following (straight-forward) association of points $x \in \EMD_s(\{0,1\}^d)$ to points in $\EMD_s(\bT)$: if the point $x \in \EMD_s(\{0,1\}^d)$ is specified by the $s$ elements $x_1,\dots, x_s \in \{0,1\}^d$, we consider the point $x' \in \EMD_s(\bT)$ given by the $s$-tuple of mapped elements $x_1,\dots, x_s$ which are leaves in $\bT$. We abuse notation and refer to $x \in \EMD_s(\{0,1\}^d)$ and $x \in \EMD_{s}(\bT)$ for clarity---these are in bijective correspondence and should be clear from context whether we will use the sampled tree $\bT$, or the original representation in $\{0,1\}^d$.

\paragraph{Data Structure Guarantees for $\SampleTree$.} It is important to note (and similarly to Definition~\ref{def:data-ind-hash}) that the running time of naively executing $\SampleTree$ will incur exponential-in-$d$ factors, since Line~3 of $\quadtree$ iterates through $u \in \{0,1\}^{2^{\ell}}$ (where $\ell$ may be as high as $\poly(d)$). Therefore, the total number of edges in $\bT$ will incur exponential-in-$d$ factors. However, the number of edges of $\bT$ whose weight depends on the sample $\by_1,\dots, \by_m \sim \mu$ is only $ms\cdot L$, as there are at most $s$ elements in each of the $m$ points $\by_1,\dots, \by_m$ and these go down $L$ edges; the rest of the edges have weights are $\xi \cdot d / 2^{\ell}$, which only depend on the depth $\ell$ and thus be (implicitly) maintained.\footnote{Even though the parameter $\xi$ did not play a role in Section~\ref{sec:Dynamic Embedding}, it will be important for Lemma~\ref{lem:sample-tree-contr}.} Even though $2^{\ell}$ may be larger than $d$ (this was useful in the proof of Claim~\ref{cl:few-empty}), it suffices to maintain the subset of sampled coordinates from $[d]$ (which takes $O(d)$ space). We thus have the following two facts, which we will use to implicitly compute the embedding of points in $\EMD_{s}(\{0,1\}^d)$ into $\ell_1$. 

\begin{fact}[(Folklore) Isometric Embedding of a Tree Metric into $\ell_1$]\label{fact:ell-1-embed}
    Let $\bT$ be any (rooted) weighted tree with $k$ edges and depth $L+1$:
    \begin{itemize}
        \item  There exists a map $\psi_{\bT} \colon \EMD_{s}(\bT) \to \R^{k}$ which is an isometric embedding into $\ell_1$, i.e., for any $x, y \in \EMD_s(\bT)$, $\EMD_{\bT}(x,y) = \| \psi_{\bT}(x) - \psi_{\bT}(y)\|_1$~(implicit in Section~4 of \cite{C02}).
        \item For $x \in \EMD_s(\bT)$, the vector $\psi_{\bT}(x) \in \R^{k}$ has $(L+1) \cdot s$ non-zero entries.
    \end{itemize}
\end{fact}
Note that, the data structure may then provide access to the root-to-leaf path specified by an element $a \in \{0,1\}^d$ to the leaf of $\bT$ where it mapped to. In order to maintain a draw $\bT$ from $\SampleTree(\mu, m)$, the data structure may first read $\mu$ (supported on $n$ points) and take $m$ samples in $O(mn)$ time and then store the data-dependent weights in $O(msdL)$ time. Given a point $x \in \EMD_s(\{0,1\}^d)$, one may then evaluate the sparse representation of $\psi_{\bT}(x)$ by obtaining its root-to-leaf path in $\poly(sd)$ time as well.
 
%\begin{fact}
%There exists a data structure which maintains a sample $\bT$ generated from $\SampleTree(\mu, m)$, and has the following guarantees:
%\begin{itemize}
%\item \emph{\textbf{Queries}}: A query is specified by a point $x \in \EMD_s(\{0,1\}^d)$ and in time $O(sdL)$, can output a sparse representation of the vector $\psi_{\bT}(x)$, where $x \in \EMD_s(\bT)$ is the associated point and $\psi_{\bT}$ is the isometric embedding of Fact~\ref{fact:ell-1-embed}.
%\item \emph{\textbf{Space Complexity}}: Given a constant-time sampler for $\mu$, the data structure requires $O(msdL)$ time to preprocess and produces a data structure of space complexity $O(Ld) + O(msL)$, where the first term is to store the maps $\bphi_1,\dots, \bphi_L$ (sampled in $\quadtree$), and the second is to store the $ms(L+1)$ weights which depend on $\by_1,\dots, \by_m$ (sampled in $\SampleTree$).
%\end{itemize}
%\end{fact}

\paragraph{Expansion and Contraction of $\SampleTree$.} Given the above description of $\SampleTree$ and the corresponding embedding that it produces into $\ell_1$, we state two lemmas below which bound the expansion and contraction of the $\SampleTree$ embedding. The proof of these two lemmas will constitute the bulk of the remainder of the section, and assuming the two lemmas, the proof of Lemma~\ref{lem:ingred-3} follows by concatenation with an $\ell_1$ locality-sensitive hash function.

\begin{lemma}[Expansion of $\SampleTree$]\label{lem:sample-tree-expan}
Consider any pair of points $x, y \in \EMD_{s}(\{0,1\}^d)$, and suppose that $x$ is $(\alpha, \tau)$-locally dense with respect to $\mu$. Then, as long as $m = \omega(\log(sd)/\alpha)$,   % where $m = \omega(\log(sd)/\alpha)$, 
\begin{align*}
\Ex_{\bT}\left[ \EMD_{\bT}(x, y)\right] \leq \EMD(x, y) \cdot \tilde{O}(\log(msd/\delta)) \left(1 + \log\left(\frac{\tau + s}{\EMD(x,y)} + 1\right)\right). 
\end{align*}
over a draw of $\bT$ from $\SampleTree(\mu, m)$,
\end{lemma}

\begin{lemma}[Non-Contraction of $\SampleTree$]\label{lem:sample-tree-contr}
For any $\delta \in (0, 1)$, consider executing $\quadtree$ (in Figure~\ref{fig:DDquadtree-prelims}) with the parameter
\[ \xi = \Omega(\log(msd/\delta)). \]
Then, for any pair of points $x, y \in \EMD_{s}(\{0,1\}^d)$, over a draw of $\bT$ from $\SampleTree(\mu, m)$,
\[ \Prx_{\bT}\left[ \EMD_{\bT}(x, y) < \EMD(x, y)\right] \leq \delta. \]
\end{lemma}

\subsubsection{Proof of Lemma~\ref{lem:ingred-3} assuming Lemma~\ref{lem:sample-tree-expan} and Lemma~\ref{lem:sample-tree-contr}}

In order to prove Lemma~\ref{lem:ingred-3}, we make use of Lemmas~\ref{lem:sample-tree-expan} and~\ref{lem:sample-tree-contr} in order to embed into $\ell_1$, and utilize a locality-sensitive hash function in $\ell_1$. In particular, classic works on locality-sensitive hashing~\cite{IM98, HIM12} give, for any parameter $\gamma > 0$, a distribution over hash functions $\bphi \colon \R^{k} \to U$ which satisfies, for any $x, y \in \R^{k}$
\begin{align}
    \Prx_{\bphi}\left[ \bphi(x) \neq \bphi(y)\right] &\leq \frac{\|x - y\|_1}{\gamma} \label{eq:ell-1-lsh-close}\\
    \Prx_{\bphi}\left[ \bphi(x) = \bphi(y)\right] &\leq \exp\left(-\frac{\|x - y\|_1}{\gamma}\right). \label{eq:ell-1-lsh-far}
\end{align}
Furthermore, it is simple to construct a data structure which maintains a description of a hash function $\bphi$ which is generated ``on-demand,'' such that, if the vector $x \in \R^k$ is sparse and written as its sparse representation, the data structure can output $\bphi(x)$ in time which is linear in the description of $x$. Given these guarantees, check both required properties of Lemma~\ref{lem:ingred-3} whenever we let $\bh \sim \calH(\mu, \tau, \gamma,\delta)$ denote the concatenation of 
\[ \bh \quad:\quad x \in \EMD_s(\{0,1\}^d) \mathop{\longmapsto}^{\text{Id}} x \in \EMD_{s}(\bT) \mathop{\longmapsto}^{\psi_{\bT}} \psi_{\bT}(x) \in \R^k \mathop{\longmapsto}^{\bphi} \bphi(\psi_{\bT}(x)) \in U, \]
where the first (identity) map $x \in \EMD_{s}(\{0,1\}^d)$ to $x \in \EMD_s(\bT)$ is the natural association of the elements of $x$ as vectors in $\{0,1\}^d$ to elements of $x$ as leaves of $\bT$, the second map $\psi_{\bT}$ is the map from Fact~\ref{fact:ell-1-embed}, and the third is the LSH for $\ell_1$ specified in (\ref{eq:ell-1-lsh-close}) and (\ref{eq:ell-1-lsh-far}). We set $m = \poly(\log(sd)/\alpha)$, thus $\xi = \Theta(\log(sd / (\delta \alpha))$ when invoking Lemma~\ref{lem:sample-tree-expan} and Lemma~\ref{lem:sample-tree-contr}. 
\begin{itemize}
\item \textbf{Close Points Collide:} Given any pair of points $x, y \in \EMD_{s}(\{0,1\}^d)$, if $x$ is $(\alpha, \tau)$-locally dense with respect to $\mu$, we use Lemma~\ref{lem:sample-tree-expan} to evaluate:
\begin{align*}
    \Prx_{\bh \sim \calH(\mu, \tau, \gamma,\delta)}\left[ \bh(x) \neq \bh(y)\right] &= \Ex_{\bT}\left[ \Prx_{\bphi}\left[ \bphi(\psi_{\bT}(x)) \neq \bphi(\psi_{\bT}(y)) \right]\right] \stackrel{(\ref{eq:ell-1-lsh-close})}{\leq} \Ex_{\bT}\left[ \dfrac{\| \psi_{\bT}(x) - \psi_{\bT}(y)\|_1}{\gamma}\right] \\
        \overset{(\ref{fact:ell-1-embed})}&{\leq} \Ex_{\bT}\left[ \frac{\EMD_{\bT}(x, y)}{\gamma} \right] \\
        \overset{(\ref{lem:sample-tree-expan})}&{\leq} \dfrac{\EMD(x,y)}{\gamma} \cdot \tilde{O}\left( \log\left(\frac{sd}{\delta \alpha}\right) \right) \cdot \left( 1 + \log\left(\frac{\tau + s}{\EMD(x,y)} + 1 \right) \right),
\end{align*}
where the last inequality simplified $m = \poly(\log(sd)/\alpha)$.
\item \textbf{Far Points Separate:} For any pair of points $x, y \in \EMD_{s}(\{0,1\}^d)$, we use a union bound and Lemma~\ref{lem:sample-tree-contr} to upper bound the probability that $\bh(x) = \bh(y)$. Namely, we have
\begin{align*}
    \Prx_{\bh \sim \calH(\mu, \tau, \gamma,\delta)}\left[ \bh(x) = \bh(y) \right] &\leq \Ex_{\bT}\left[ \Prx_{\bphi}\left[ \bphi(\psi_{\bT}(x)) = \bphi(\psi_{\bT}(y)) \right] \mid \EMD_{\bT}(x,y) \geq \EMD(x,y) \right] \\
    &\qquad\qquad + \Prx_{\bT}\left[ \EMD_{\bT}(x,y) < \EMD(x,y)\right] \\
    &\leq \Ex_{\bT}\left[ \exp\left(-\dfrac{\EMD_{\bT}(x, y)}{\gamma} \right) \mid \EMD_{\bT}(x,y) \geq \EMD(x,y) \right] + \delta \\
    &\leq \exp\left(-\frac{\EMD(x,y)}{\gamma}\right) + \delta,
\end{align*}
where above, we similarly use Fact~\ref{fact:ell-1-embed} to embed $\EMD_s(\bT)$ into $\ell_1$ isometrically, the expression (\ref{eq:ell-1-lsh-far}) for $\bphi$, and finally Lemma~\ref{lem:sample-tree-contr}.
\end{itemize}

%% file: data-ind-hash-and-SampleTree/st-expansion.tex
We first introduce some notations and consequential observations, which will help  prove Lemma~\ref{lem:sample-tree-expan} by decomposing $\EMD_\bT(x,y)$ into ``data-independent" part and ``data-dependent" part. Some of the notations will also be used later in the proof of Lemma~\ref{lem:sample-tree-contr}. 

\paragraph{Basic notations for $\SampleTree$.} For any distribution $\mu$, integer $m \geq 0$, and any draw of $\bT$ from an execution of $\SampleTree(\mu,m)$ (see Figure~\ref{fig:quadtree-embed} and Figure~\ref{fig:DDquadtree-prelims}), we have the following:
\begin{itemize}
    \item For every element $a \in \{0,1\}^d$, there is a unique root-to-leaf path in $\bT$, given by the sequence of nodes $v_0(a), \dots, v_{L}(a)$, inductively defined by $v_0(a) = v_0$ and 
    \[v_\ell(a) \text{ is the child } v_u\text{ of } v_{\ell-1}(a) \text{ with } a \in \Elms(v_u).\]
    \item For a pair of elements $a,b \in \{0,1\}^d$, let $\Split_\ell(a,b)$ be the indicator variable of the event $v_\ell(a) \neq v_\ell(b)$, i.e., 
    \[\Split_\ell(a,b) = \ind\{v_\ell(a) \neq v_\ell(b)\}.\] 
    \item For a pair of elements $a,b \in \{0,1\}^d$, we can write $d_\bT(a,b)$ as
    \begin{align}
    \label{eq:rewrite-tree-metric-using-split}
        d_\bT(a,b) = \sum_{\ell = 0}^L \Split_{\ell+1}(a,b) \cdot (\bw(v_\ell(a),v_{\ell+1}(a)) + \bw(v_\ell(b),v_{\ell+1}(b)) ).
    \end{align}
\end{itemize}

\paragraph{Decompose $\EMD_\bT(x,y)$ into data-independent and data-dependent parts.} Recall that $\bhomg$ is the neighborhood of elements of points in $\bOmega$ sampled from $\mu$ (Figure~\ref{fig:quadtree-embed}). For $\SampleTree(\mu,m)$, we define the following notations for all $a \in \{0,1\}^d, \ell \in \{0,\dots, L\}$:
\begin{align*}
    % \DataInd(a, \ell) &= \ind{\left[\forall \be \in \bhomg, \Split_\ell(a,\be) = 1 \right]} \\[10 pt]
    \DataInd(a, \ell) &= \ind \left\{\bElms(v_{\ell}(a)) \cap \bhomg \neq \emptyset\right\} \\
    \bw_{\Dep}(v_\ell(a), v_{\ell+1}(a)) &=\begin{cases} \vspace{0.25cm}
                        \Ex\limits_{\substack{\be \sim \bElms(v_\ell(a)) \cap \bhomg \\ \be' \sim \bElms(v_{\ell+1}(a)) \cap \bhomg}}\left[\|\be - \be'\|_1 \right] & \; \;  \bElms(v_{\ell+1}(a)) \cap \bhomg \neq \emptyset  \\
                         0& \; \; \text{otherwise.}
    \end{cases}
\end{align*}

Note that $\bw(v_\ell(a), v_{\ell+1}(a)) $ by executing $\SampleTree(\mu,m)$ is set to
\[ \bw(v_\ell(a), v_{\ell+1}(a))  = \begin{cases} \vspace{0.25cm}
                        \Ex\limits_{\substack{\be \sim \bElms(v_\ell(a)) \cap \bhomg \\ \be' \sim \bElms(v_{\ell+1}(a)) \cap \bhomg}}\left[\|\be - \be'\|_1 \right] & \; \;  \bElms(v_{\ell+1}(a)) \cap \bhomg \neq \emptyset  \\
                         d/ 2^{\ell} \cdot \xi& \; \; \text{otherwise}
                    \end{cases}\]
where $\xi$ is a parameter that can be set up. If $\bw(.,.)$ of an edge evaluates to the first case, we call the edge ``data-dependent"; Otherwise, we call the edge ``data-independent". In the same spirit, for a pair of elements $a,b \in \{0,1\}^d$, we write $d_\bT(a,b)$ into data-dependent part and data-independent part:
\begin{align}
    d_\bT(a,b) \overset{(\ref{eq:rewrite-tree-metric-using-split})}&{=}  \sum_{\ell = 0}^L \Split_{\ell+1}(a,b) \cdot (\bw(v_\ell(a),v_{\ell+1}(a)) + \bw(v_\ell(b),v_{\ell+1}(b)) ) \nonumber \\
    & = \sum_{\ell=0}^{L}  \Split_{\ell+1}(a,b) \cdot (\DataInd(a,\ell+1)+\DataInd(b,\ell+1)) \cdot \frac{d}{2^\ell}\cdot \xi \nonumber \\
    & \quad +\sum_{\ell=0}^{L}  \Split_{\ell+1}(a,b) \cdot \big( \bw_{\Dep}(v_\ell(a), v_{\ell+1}(a)) + \bw_{\Dep}(v_\ell(b), v_{\ell+1}(b)) \big) \label{eq:decompose-dt}
\end{align}

Let $x = \{a_1,a_2,\dots,a_s\},$ $y = \{b_1,b_2,\dots,b_s\} \in \EMD_s(\{0,1\}^d)$ be a pair of points, and we write $\pi\colon [s] \to [s]$ to denote the minimum matching for $x,y$ in $\EMD_s(\{0,1\}^d)$. In order to decompose $\EMD_\bT(x,y)$ into data-dependent part and data-independent part, we further define the following:
\begin{align*}
    \Qind_\bT(x,y) &= \sum_{i=1}^s \sum_{\ell=0}^{L}  \Split_{\ell+1}(a_i,b_{\pi(i)}) \cdot (\DataInd(a_i,\ell+1) + \DataInd(b_{\pi(i)},\ell+1))\cdot \frac{d}{2^\ell}\cdot \xi \\
    \Qdep_\bT(x,y) &= \sum_{i=1}^s \sum_{\ell=0}^{L}  \Split_{\ell+1}(a_i,b_{\pi(i)}) \cdot \big( \bw_{\Dep}(v_\ell(a_i), v_{\ell+1}(a_i)) + \bw_{\Dep}(v_\ell(b_{\pi(i)}), v_{\ell+1}(b_{\pi(i)})) \big)
\end{align*}
Note that $\Qind_\bT(x,y) + \Qdep_\bT(x,y) = \sum_{i = 1}^s d_\bT(a_i,b_{\pi(i)}) $ by \ref{eq:decompose-dt}, which is the value of the matching $\pi$ in $\bT$. With the above notations, we can finally upper bound $\EMD_\bT(x,y)$ by considering the matching $\pi$ in $\bT$, and dividing the contribution of the cost of $\pi$ into two parts:
\begin{align}
\label{eq:decompose-into-ind-and-dep}
\EMD_\bT(x,y) \leq \Qind_\bT(x,y) + \Qdep_\bT(x,y).
\end{align} 
% Notice that for any $a,b \in \{0,1\}^d$, we can write $d_\bT(a,b)$ as 
% \begin{align*}
%     d_\bT(a,b) = & \sum_{\ell=0}^{L}  \Split_{\ell+1}(a,b) \cdot \DataInd(a_i,\ell+1) \cdot \frac{d}{2^\ell}\cdot \xi \\
%     & + \sum_{\ell=0}^{L}  \Split_{\ell+1}(a,b) \cdot \DataInd(b_i,\ell+1) \cdot \frac{d}{2^\ell}\cdot \xi \\
%     & +\sum_{\ell=0}^{L}  \Split_{\ell+1}(a_i,b_{\pi(i)}) \cdot \big( \bw_{\Dep}(v_\ell(a_i), v_{\ell+1}(a_i)) + \bw_{\Dep}(v_\ell(b_{\pi(i)}), v_{\ell+1}(b_{\pi(i)})) \big)
% \end{align*}
The above inequality follows from the facts that the left hand side above is the value of minimum matching in $\bT$, which is at most the value of the matching $\pi$ in $\bT$.  
The goal is reduced to upper bounding data-independent part $\Qind_\bT(x,y)$ and data-dependent part $\Qdep_\bT(x,y)$ respectively. In particular, it is easy to see that it suffices to prove the following two lemmas, in order to prove Lemma~\ref{lem:sample-tree-expan}.

\begin{lemma}
\label{lem:sample-tree-expan-ind}
    Let $\bT$ be drawn from $\SampleTree(\mu,m)$, and $x, y \in \EMD_s(\{0,1\}^d)$. %$ = \{a_1,a_2,\dots,a_s\},$ $ y = \{b_1,b_2,\dots,b_s\} \in \EMD_s(\{0,1\}^d)$ be two points. 
    If $x$ is $(\alpha, \tau)$-locally dense, then as long as $m = \omega(\log(sd)/\alpha)$,
    \[ \Ex_\bT \left[ \Qind_\bT(x,y) \right] \leq \EMD(x,y) \cdot O(\xi) \cdot \left( \log\left(\frac{\tau + s}{\EMD(x,y)} + 1 \right) + \log\log \frac{sd}{\alpha} \right) \]
\end{lemma}

\begin{lemma}
\label{lem:sample-tree-expan-dep}
    Let $\bT$ be drawn from $\SampleTree(\mu,m)$, and $x, y \in \EMD_s(\{0,1\}^d)$ % = \{a_1,a_2,\dots,a_s\},$ $ y = \{b_1,b_2,\dots,b_s\} \in \EMD_s(\{0,1\}^d)$ be two points. 
    Then, we have
    \[ \Ex_\bT \left[\Qdep_\bT(x,y) \right] \leq \tilde{O}(\log (msd)) \cdot \EMD(x,y) \] 
\end{lemma}

\begin{proof}[Proof of Lemma~\ref{lem:sample-tree-expan} assuming Lemma~\ref{lem:sample-tree-expan-ind} and Lemma~\ref{lem:sample-tree-expan-dep}]
    %Recall that, in Lemma~\ref{lem:sample-tree-expan},
    %\[\sfF = \Tilde{O}\left(\log(msd) + \xi \cdot \log\log \frac{sd}{\alpha}\right)\]
    Putting together the upper bounds in Lemma~\ref{lem:sample-tree-expan-ind} and Lemma~\ref{lem:sample-tree-expan-dep} gives
    \begin{align*}
        \EMD_\bT(x,y) \overset{(\ref{eq:decompose-into-ind-and-dep})}&{\leq} \Qind_\bT(x,y) + \Qdep_\bT(x,y) \\
        \overset{(\ref{lem:sample-tree-expan-ind},\ref{lem:sample-tree-expan-dep})}&{\leq} \quad \EMD(x,y) \cdot \left( O(\xi) \cdot \log\left(\frac{\tau + s}{\EMD(x,y)} + 1 \right) + O(\xi) \cdot \log\log \frac{sd}{\alpha}  + \tilde{O}(\log (msd))\right) \\
        & \leq  \EMD(x,y) \cdot \tilde{O}\left( \log(msd/\delta) \right)\left(1 + \log\left(\frac{\tau + s}{\EMD(x,y)} + 1 \right)\right),
    \end{align*}
    given that $T$ is drawn from $\SampleTree(\mu,m)$, $x,y\in \EMD_s(\{0,1\}^d)$, $x$ is $(\alpha,\tau)$-locally dense, and $m = \omega(\log(sd)/\alpha)$.
\end{proof}

The next two sections are devoted to proving the above two lemmas.

\subsubsection{Bounding the Data-Independents Part (Proof of Lemma~\ref{lem:sample-tree-expan-ind})}\label{sec:data-ind-part-bound-1}

In this section, we upper bound the expectation of the data-independent part, and show that 
\[ \Ex_{\bT} \left[\Qind_\bT(x,y)\right] \leq \EMD(x,y) \cdot O(\xi) \cdot \left( \log\left(\frac{\tau + s}{\EMD(x,y)} + 1 \right) + \log\log \left(sd/\alpha\right) \right),\]
when $x$ is $(\alpha,\tau)$-locally dense. %We consider, for $a_i \in x$, its closest point $\bp_i$ in the set $\bhomg$ (see Figure \ref{fig:quadtree-embed} for $\bhomg$) so that, $\bp_i = \argmin_{e\in \bhomg}\|a_i - \bp_i\|$, and note that it is a random variable. 
An execution of $\SampleTree(\mu,m)$ introduces two independent sources of randomness: (1) the $m$ samples $\by_1,\dots, \by_m \sim \mu$ drawn to generate $\bOmega$, and (2) the randomness used in a call to $\quadtree(\hat{\bOmega})$, which draws random coordinates $\bj_1,\dots, \bj_{2^{L+1}-1} \sim [d]$ for the random hash functions $\bphi_\ell \sim \calH_{2^\ell}$ (see Figure \ref{fig:DDquadtree-prelims}). Let $R = 2^{L+1}-1$ denote the total number of (random) coordinates sampled which define the tree $\bT$, and let $r_{\ell} = 2^{\ell+1} - 1$ denote the total number of (random) coordinates sampled up to (and including) depth $\ell$. We can write $\Ex_{\bT}\left[\Qind_\bT(x,y)\right]$ by expanding out both sources of randomness---using the expressions in Subsection~\ref{sec:proofofexpand}:
\begin{align}
  & \Ex_{\bT} \left[\Qind_\bT(x,y)\right] \nonumber \\
  &= \Ex_{\by_1,\dots, \by_m} \left[ \Ex_{\bj_1,\dots, \bj_R}\left[\sum_{i=1}^s \sum_{\ell=0}^{L} \left(\left\{ \begin{array}{c} \Split_{\ell+1}(a_i,b_{\pi(i)}) \times \\ \DataInd(a_i,\ell+1) \end{array} \right\} + \left\{ \begin{array}{c} \Split_{\ell+1}(a_{i}, b_{\pi(i)}) \times \\ \DataInd(b_{\pi(i)},\ell+1))\end{array} \right\} \right) \cdot \frac{d}{2^\ell}\cdot \xi \right] \right] \nonumber \\
  % =& \sum_{i=1}^s \sum_{\ell=0}^{L}  \Ex_{\bT} \left[\Split_{\ell+1}(a_i,b_{\pi(i)}) \cdot \DataInd(a_i,\ell+1)\right] \cdot \frac{d}{2^\ell}\cdot \xi \\
  &= \Ex_{\by_1,\dots, \by_m} \Bigg[\sum_{i=1}^s \sum_{\ell=0}^{L} \Prx_{\bj_1,\dots, \bj_R} \left[\begin{array}{c}
    \Split_{\ell + 1}(a_i,b_{\pi(i)}) = 1 \land \\
     \DataInd(a_i,\ell+1) = 1
    \end{array} \right] \cdot \frac{d}{2^{\ell}} \cdot \xi \Bigg] \label{eq:a_i-part-exp}\\
    &\qquad\qquad +  \Ex_{\by_1,\dots, \by_m}\Bigg[ \sum_{i=1}^s \sum_{\ell=0}^{L} \Prx_{\bj_1,\dots, \bj_R} \left[\begin{array}{c}
     \Split_{\ell + 1}(a_i,b_{\pi(i)}) = 1 \land \\
     \DataInd(b_{\pi(i)},\ell+1) = 1
    \end{array} \right] \cdot \frac{d}{2^\ell}\cdot \xi \Bigg]. \label{eq:b_i-part-exp}
\end{align}
Consider, first, the inner-most terms in the expression above, by fixing the draws $\by_1,\dots, \by_m$. For $i \in [s]$ and $\ell \in \{0, \dots, L\}$, we will now upper bound the inner-most probability over the draws of $\bj_1,\dots, \bj_{R}$,
\begin{align} \Prx_{\bj_1,\dots, \bj_{R}}\left[ \begin{array}{c} \Split_{\ell+1}(a_i, b_{\pi(i)}) = 1 \wedge \\
\DataInd(a_i, \ell+1)=1 \end{array} \right] \qquad\text{and}\qquad \Prx_{\bj_1,\dots, \bj_{R}}\left[ \begin{array}{c} \Split_{\ell+1}(a_i, b_{\pi(i)}) = 1 \wedge \\
\DataInd(b_{\pi(i)}, \ell+1)=1 \end{array} \right] . \label{eq:todo-now-probs}
\end{align}

%For all $i\in[s]$, $\bp_i$ is random over the draws of $\by_1,\dots,\by_m$. We will fix $\by_1,\dots,\by_m$ and deal with the randomness of $\bj_1,\dots, \bj_R$ first. Thus, we can think $\bhomg$ is fixed to $\hat{\Omega}$, and $\bp_i$ is fixed to $p_i$.  

The subsequent two claims will help us upper bound the above expression. The first claim is immediate from the definitions, and the second claim has a simple proof. 

\begin{claim}\label{cl:split-claim-1}
    For any two points $a, b \in \{0,1\}^d$, and any $\ell$, once we draw $\bj_1,\dots, \bj_{R}$, $v_{\ell}(a) = v_{\ell}(b)$ whenever $a_{\bj_k} = b_{\bj_k}$ for all $k \in [r_{\ell}]$. Thus, $a_i$ and $b_{\pi(i)}$ are split if the above does not occur, i.e.,
    \[ \Split_{\ell+1}(a_i, b_{\pi(i)}) = \ind\Big\{ \exists k \in [r_{\ell}] : (a_i)_{\bj_k} \neq (b_{\pi(i)})_{\bj_k}\Big\}.\]
\end{claim}

\begin{claim}\label{cl:data-ind-claim-2}
    Consider any $i\in[s]$ and any $\ell \in \{0,\dots, L\}$. Having fixed the draw $\Omega = \{ y_1,\dots, y_m \}$, let $p \in \Omega$ denote the nearest neighbor of $a_i$. 
    \[ \DataInd(a_i, \ell+1) \leq \ind\Big\{ \exists \text{ \emph{distinct} } k_1, k_2 \in [r_{\ell}] : (a_i)_{\bj_{k_1}} \neq p_{\bj_{k_1}} \wedge (a_i)_{\bj_{k_2}} \neq p_{\bj_{k_2}}\Big\}.\]
\end{claim}

\begin{proof}
In order to split $a_i$ from all points in $\hat{\Omega}$ during the execution of $\quadtree(\hat{\Omega})$, $a$ must be split from all elements in $\nbr(p)$ since $\nbr(p) \subseteq \hat{\Omega}$. In particular, there is an index $k_1 \in [r_{\ell}]$ which witnesses the split between $a_i$ and $p$ and satisfies $(a_i)_{bj_{k_1}} \neq p_{\bj_{k_1}}$---otherwise, $a_i$ and $p$ are not split. Furthermore, let $p' \in \nbr(p)$ be the point which agrees with $p$ in all but the $\bj_{k_1}$-th coordinate, and since $p' \in \hat{\bOmega}$, there must be an index $k_2 \in [r_{\ell}]$ which witnesses the split between $a_i$ and $p'$. Finally, $k_1 \neq k_2$ since $(a_i)_{\bj_{k_1}} \neq p_{\bj_{k_1}}$ and $p_{\bj_{k_1}} \neq p_{\bj_{k_1}}'$, since these are in the hypercube, $(a_i)_{\bj_{k_1}} = p_{\bj_{k_1}}'$. 
\end{proof}

Using Claims~\ref{cl:split-claim-1} and~\ref{cl:data-ind-claim-2}, for any $i \in [s]$ and $\ell \in \{0,\dots, L\}$, we can now upper bound (after fixing the randomness in $\Omega$ and hence $a_i$'s nearest neighbor $p$) the probability
\begin{align}
\label{eq:two-diff-coords}
    \Prx_{\bj_1,\dots, \bj_R} \left[\begin{array}{c}
    \Split_{\ell+1}(a_i,b_{\pi(i)}) = 1 \land \\
    \DataInd(a_i,\ell+1) = 1
    \end{array} \right] 
    \leq \Prx_{\bj_1,\dots, \bj_{r_\ell}}\left[\begin{array}{c}
         \exists \text{ distinct } k_1, k_2 \in [r_\ell] \text{ s.t } \\
        (a_i)_{k_1} \neq (b_{\pi(i)})_{k_1} \wedge (a_i)_{k_2} \neq p_{k_2}
    \end{array} \right].
\end{align} 
Based on the above, we have the following two claims which allows us to upper bound the left- and right-most inequalities in (\ref{eq:todo-now-probs}).
\begin{claim}
\label{claim:decouple-events}
For any $i \in [s]$ and $\ell \in \{0,\dots, L\}$, letting $p \in \Omega$ be the nearest neighbor of $a_i$, 
\[ \Prx_{\bj_1,\dots, \bj_R} \left[\begin{array}{c}
    \Split_{\ell+1}(a_i,b_{\pi(i)}) = 1 \land \\
     \DataInd(a_i,\ell+1) = 1
\end{array} \right] \leq \left(\frac{r_\ell}{d}\right)^2 \| a_i - b_{\pi(i)}\|_1 \cdot\|a_i - p\|_1 .\]
\end{claim}

\begin{proof}
    \begin{align*}
    \Prx_{\bj_1,\dots, \bj_R} \left[\begin{array}{c}
     \Split_{\ell+1}(a_i,b_{\pi(i)}) = 1 \land \\
     \DataInd(a_i,\ell+1) = 1
    \end{array} \right] 
    \overset{(\ref{eq:two-diff-coords})}&{\leq}\Prx_{\bj_1,\dots, \bj_{r_\ell}}\left[\begin{array}{c}
         \exists \text{ distinct } k_1, k_2 \in [r_\ell] \text{ s.t} \\
         (a_i)_{k_1} \neq (b_{\pi(i)})_{k_1} \wedge (a_i)_{k_2} \neq p_{k_2}
    \end{array} \right] \\
    &\leq \sum_{\substack{k_1,k_2 \in [r_\ell],\\k_1 \neq k_2}} \Prx_{\bj_{k_1}}\left[(a_i)_{\bj_{k_1}} \neq (b_{\pi(i)})_{\bj_{k_1}}\right] \cdot \Prx_{\bj_{k_2}}\left[(a_i)_{\bj_{k_2}} \neq p_{\bj_{k_2}}\right] \\
    &\leq (r_\ell)^2 \cdot \frac{\| a_i - b_{\pi(i)}\|_1}{d} \cdot \frac{\|a_i - p\|_1}{d} \\
    &\leq \left(\frac{r_\ell}{d}\right)^2 \cdot \| a_i - b_{\pi(i)}\|_1 \cdot \|a_i - p\|_1 
\end{align*}
The second inequality comes from the fact that the two events $a_{\bj_{k_1}} \neq b_{\bj_{k_1}}$ and $a_{\bj_{k_2}} \neq p_{\bj_{k_2}}$ are independent since $k_1,k_2$ are distinct. The second to last inequality follows from applying union bound over all possible distinct $k_1,k_2$. 
\end{proof}

\begin{claim}
\label{claim:decouple-events-2}
For any $i \in [s]$ and $\ell \in \{0,\dots, L\}$, letting $p \in \Omega$ be the nearest neighbor of $a_i$, 
\[ \Prx_{\bj_1,\dots, \bj_R} \left[\begin{array}{c}
    \Split_{\ell+1}(a_i,b_{\pi(i)}) = 1 \land \\
     \DataInd(b_{\pi(i)},\ell+1) = 1
\end{array} \right] \leq \left(\frac{r_\ell}{d}\right)^2 \| a_i - b_{\pi(i)}\|_1 \cdot \left( \|a_i - b_{\pi(i)}\|_1 + \| a_i - p\|_1 \right) .\]
\end{claim}

\begin{proof}
Notice that the above claim exchanges the notions of $a_i$ and $b_{\pi(i)}$ in Claim~\ref{claim:decouple-events}, as the $\DataInd(b_{\pi(i)}, \ell+1)$ event replaces $\DataInd(a_i, \ell+1)$. However, since $a_i$ and $b_{\pi(i)}$ are not entirely symmetric (as $p$ is denoted as the nearest neighbor of $a_i$), we will incur an extra additive $\|a_i - b_{\pi(i)}\|_1$ term. In particular, if we let $\tilde{p} \in \Omega$ denote the nearest neighbor of $b_{\pi(i)}$, Claim~\ref{claim:decouple-events} implies
\begin{align*}
    \Prx_{\bj_1,\dots, \bj_{R}}\left[\begin{array}{c} \Split_{\ell+1}(a_i, b_{\pi(i)})=1 \wedge \\ \DataInd(b_{\pi(i)}, \ell+1) \end{array} \right] \leq \left(\frac{r_{\ell}}{d}\right)^2 \| a_i - b_{\pi(i)}\|_1 \cdot \|b_{\pi(i)} - \tilde{p}\|_1.
\end{align*}
The claim follows since the fact $\tilde{p}$ is the nearest neighbor implies $\|b_{\pi(i)} - \tilde{p}\|_1 \leq \| b_{\pi(i)} - p\|_1$, and we apply the triangle inequality to say $\|b_{\pi(i)} - p\|_1 \leq \|a_i - b_{\pi(i)}\|_1 + \| a_i - p\|_1$
\end{proof}

With Claims~\ref{claim:decouple-events} and~\ref{claim:decouple-events-2}, we may now proceed towards upper bounding (\ref{eq:a_i-part-exp}) and (\ref{eq:b_i-part-exp}). 
\begin{lemma}\label{eq:probs-sum-ub}
    Consider a fixed set $\Omega = \{ y_1,\dots, y_m \} \subset \{0,1\}^d$, let $i \in [s]$ be any index, and $p \in \Omega$ be $a_i$'s nearest neighbor. Then, both
    \begin{align}
        &\sum_{\ell=0}^{L}\Prx_{\bj_1,\dots, \bj_R}\left[\begin{array}{c} \Split_{\ell+1}(a_i, b_{\pi(i)})=1 \wedge  \\ \DataInd(a_i, \ell+1)=1 \end{array} \right] \cdot \frac{d \cdot \xi}{2^{\ell}}, \qquad \text{and} \label{eq:a_i-term}\\
        &\sum_{\ell=0}^{L} \Prx_{\bj_1,\dots, \bj_R}\left[\begin{array}{c} \Split_{\ell+1}(a_i, b_{\pi(i)})=1 \wedge \\ \DataInd(b_{\pi(i)}, \ell+1) = 1\end{array} \right] \cdot \frac{d\cdot \xi}{2^{\ell}} \label{eq:b_i-term}
    \end{align}
are at most
\[ 8 \cdot \xi \cdot \| a_i - b_{\pi(i)} \|_1 \cdot \left( \log_2 \left( \dfrac{\|a_i - p\|_1}{\|a_i - b_{\pi(i)}\|_1} + 1\right) + 1\right). \]
\end{lemma}

\begin{proof}
    We begin with upper bounding the second term, as both will be symmetric arguments (using the fact that Claim~\ref{claim:decouple-events-2} is a weakening of the analogous inequality in Claim~\ref{claim:decouple-events}. We thus introduce a variable $\gamma > 1$ (which we will optimize later), and consider the setting
    \[ \ell^{*} = \min\left\{ \ell \in \{0, \dots, L \} : \frac{d}{2^{\ell}} \leq \gamma \cdot \|a_i - b_{\pi(i)}\|_1 \right\}.\]
    We break up the summation over $\ell \in \{0, \dots, L\}$ into three parts: \textbf{(i)} the settings of $\ell > \ell^* + \log_2(\gamma)$, \textbf{(ii)} the settings $\ell$ which are above $\ell^*$ but below $\ell^* + \log_2(\gamma)$, and \textbf{(iii)} the settings of $\ell \leq \ell^*$. 

    \textbf{Case (i)}. Cases (i) is the simplest, as it will suffice to upper bound the probabilistic event by one. Namely, case (i) considers settings where $\ell > \ell^* + \log_2(\gamma)$, and in this case,
    \begin{align*}
        \sum_{\ell = \ell^* + \lceil \log_2(\gamma)\rceil}^{L} \Prx_{\bj_1,\dots, \bj_{R}}\left[\begin{array}{c} \Split_{\ell+1}(a_i, b_{\pi(i)}) = 1 \wedge \\
        \DataInd(b_{\pi(i)}, \ell+1) \end{array} \right] \cdot \frac{d \cdot \xi}{2^{\ell}} &\leq \frac{2\cdot \xi \cdot d}{\gamma \cdot 2^{\ell^*}} \leq 2\cdot\xi \cdot \|a_i - b_{\pi(i)}\|_1,
    \end{align*}
    by the definition of $\ell^*$.

    \textbf{Case (ii).} The second case is only slightly more involved, as we will solely use Claim~\ref{cl:split-claim-1} to upper bound
    \begin{align*}
        \Prx_{\bj_1,\dots, \bj_{R}}\left[ \begin{array}{c} \Split_{\ell+1}(a_i, b_{\pi(i)}) = 1 \wedge \\
        \DataInd(b_{\pi(i)}, \ell+1) \end{array} \right] \leq \Prx_{\bj_1,\dots, \bj_R}\left[ \Split_{\ell+1}(a_i, b_{\pi(i)}) = 1\right] \leq \frac{r_{\ell}}{d} \cdot \| a_i - b_{\pi(i)}\|_1.
    \end{align*}
    Therefore, we upper bound:
    \begin{align*}
        \sum_{\ell=\ell^*}^{\ell^* + \lfloor \log_2(\gamma)\rfloor} \Prx_{\bj_1,\dots, \bj_{R}}\left[\begin{array}{c} \Split_{\ell+1}(a_i, b_{\pi(i)}) = 1 \wedge \\ \DataInd(b_{\pi(i)}, \ell+1)=1 \end{array} \right] \cdot \frac{d \cdot \xi}{2^{\ell}} \leq 2\xi \cdot \lceil \log_2(\gamma)\rceil \cdot \|a_i - b_{\pi(i)}\|_1
    \end{align*}

    \textbf{Case (iii).} This case is the most involved, where we use Claim~\ref{claim:decouple-events-2}. In particular, we may write
    \begin{align*}
        &\sum_{\ell=0}^{\ell^* - 1} \Prx_{\bj_1,\dots, \bj_R}\left[\begin{array}{c} \Split_{\ell+1}(a_i, b_{\pi(i)}) = 1 \wedge \\ \DataInd(b_{\pi(i)}, \ell+1) = 1 \end{array}
        \right] \cdot \frac{d \cdot \xi}{2^{\ell}} \\
        &\qquad \leq \sum_{\ell=0}^{\ell^*-1} \left( \frac{r_{\ell}}{d}\right)^2 \|a_i - b_{\pi(i)}\|_1 \left(\|a_i - b_{\pi(i)}\|_1 + \|a_i - p\|_1 \right) \cdot \frac{d \cdot \xi}{2^{\ell}} \\
        &\qquad \leq 4 \xi \cdot \frac{2^{\ell^*}}{d} \cdot \|a_i - b_{\pi(i)}\|_1 \left(\|a_i - b_{\pi(i)}\|_1 + \|a_i - p\|_1 \right) \leq \frac{4 \cdot \xi}{\gamma} \left(\| a_i - b_{\pi(i)}\|_1 + \|a_i - p\|_1 \right).
    \end{align*}
    Putting all cases together, we've upper bounded our desired quantity (\ref{eq:b_i-term}) by
    \begin{align*}
        2 \xi \left( 1 + \lceil \log_2(\gamma)\rceil \right) \cdot \|a_i - b_{\pi(i)}\|_1 + \frac{4 \xi}{\gamma} \left( \|a_i - b_{\pi(i)} \|_1 + \|a_i - p\|_1\right),
    \end{align*}
    where $\gamma > 1$ is unrestricted. Thus, we may set $\gamma = \|a_i - p\|_1 / \|a_i - b_{\pi(i)}\|_1 + 1$ in order to obtain our desired bound. We note that the upper bound for (\ref{eq:a_i-term}) is analogous, and may be upper bounded by the same term.
\end{proof}

\begin{proof}[Proof of Lemma~\ref{lem:sample-tree-expan-ind}]
Directly substituting (\ref{eq:a_i-part-exp}) and (\ref{eq:b_i-part-exp}), as well as Lemma~\ref{eq:probs-sum-ub}, we have
\begin{align}
\Ex_{\bT}\left[ \Qind_{\bT}(x,y)\right] &\leq 16 \xi \Ex_{\by_1,\dots, \by_m}\left[\sum_{i=1}^s \| a_i - b_{\pi(i)}\|_1 \left(\log_2\left( \frac{\|a_i - \bp_i\|_1}{\|a_i - b_{\pi(i)}\|_1} + 1\right) + 1\right) \right] , \label{eq:desired-11}
\end{align}
where we have taken $\bp_i$ to be the (random) nearest neighbor of $a_i$ among the elements of $\bOmega$, defined by points $\by_1,\dots, \by_m$. The final argument will be an application of Jensen's inequality twice. Consider the distribution over $\bi \in [s]$ which samples an index $i$ with probability $\|a_i - b_{\pi(i)}\|_1 / \EMD(x,y)$; because of concavity of the logarithm function, we may re-write the expression
\begin{align*}
    \sum_{i=1}^n \|a_i - b_{\pi(i)}\|_1 \left(\log_2\left(\frac{\|a_i - \bp_i\|_1}{\|a_i - b_{\pi(i)}\|_1} +1 \right) + 1 \right),
\end{align*}
within the expectation over $\by_1,\dots, \by_m$ as 
\begin{align*}
\EMD(x,y) \left(\Ex_{\bi}\left[\log_2\left(\frac{\|a_i - \bp_{i}\|_1}{\|a_i - b_{\pi(i)}\|_1} +1 \right) \right] +1\right) &\leq \EMD(x,y) \left(\log_2\left( \Ex_{\bi}\left[\dfrac{\|a_i - \bp_i\|_1}{\|a_i - b_{\pi(i)}\|_1} \right] + 1\right) + 1 \right) \\
        &= \EMD(x,y) \left(\log_2\left( \dfrac{\textsf{Chamfer}(x, \bOmega)}{\EMD(x,y)} + 1\right) + 1 \right).
\end{align*}
Taking the expectation with respect to $\by_1,\dots, \by_m$, and using Jensen's inequality once more, our expression (\ref{eq:desired-11}) is upper-bounded by
\[ 16 \xi \cdot \EMD(x,y) \left(\log_2\left(\dfrac{\Ex[\textsf{Chamfer}(x, \bOmega)]}{\EMD(x,y)} + 1 \right) + 1 \right),\]
where the inner-most expectation is over $\by_1,\dots, \by_m$, which define $\bOmega$. Finally, we use the fact that $x$ is $(\alpha,\tau)$-locally dense and apply Lemma~\ref{lem:dense-conse} to conclude that the above expression is at most
\[ O(\xi) \cdot \EMD(x,y) \left( \log_2\left( \dfrac{\tau + s}{\EMD(x,y)} + 1\right)+ \log \log \left( sd/\alpha \right) \right),\]
as claimed.
\ignore{
\begin{align*}
    \frac{1}{\EMD(x,y)} \sum_{i=1}^s \|a_i - b_{\pi(i)}\|_1 \left(\log_2\left(\frac{\|a_i - \bp_i\|_1}{\|a_i - b_{\pi(i)}\|_1} + 1\right)+1 \right) &= \Ex_{\bi}\left[\log_2\left( \dfrac{\|a_i - \bp_i\|_1}{\|a_i - b_{\pi(i)}\|_1} + 1 \right) \right] + 1 \\
    &\leq \log_2\left( \Ex_{\bi}\left[\frac{\|a_i - \bp_i\|_1}{\|a_i - b_{\pi(i)}\|_1}\right] + 1 \right) + 1,
\end{align*}
and we may substitute:
\[ \Ex_{\bi}\left[ \frac{\|a_i - \bp_i\|_1}{\|a_i - b_{\pi(i)}\|_1} \right] = \dfrac{\textsf{Chamfer}(x, \bigcup_{j=1}^m \by_j)}{\EMD(x,y)}.\]

\begin{align*}
    &\Ex_{\by_1,\dots, \by_m} \left[\sum_{i = 1}^s \|a_i - b_{\pi(i)} \|_1 \left(\log_2 \left( \frac{\|a_i -\bp_i\|_1}{ \|a_i - b_{\pi(i)}\|_1} + 1 \right) + 1 \right)\right]\\
    &\qquad= \EMD(x,y) \cdot \left( \Ex_{\by_1,\dots, \by_m} \left[ \sum_{i = 1}^s \frac{\|a_i - b_{\pi(i)} \|}{\EMD(x,y)} \cdot \log \left( \frac{\|a_i -\bp_i\|_1}{ \|a_i - b_{\pi(i)}\|_1} + 1 \right) \right] + 1 \right) \\
    &\qquad\leq \EMD(x,y) \cdot \Ex_{\by_1,\dots, \by_m} \left[ \log\left(\sum_{i = 1}^s \frac{\|a_i - b_{\pi(i)} \|}{\EMD(x,y)} \cdot \left( \frac{\|a_i -\bp_i\|_1}{ \|a_i - b_{\pi(i)}\|_1} + 1 \right)\right) \right] \tag{Jensen's Inequality}\\
    &\qquad \leq \EMD(x,y) \cdot  \log\left(\Ex_{\by_1,\dots, \by_m} \left[ \frac{\sum_{i = 1}^s\|a_i - \bp_i\|_1}{\EMD(x,y)} + 1\right]\right)
    \tag{Jensen's Inequality}
\end{align*}}
\end{proof}

\subsubsection{Bounding the Data-Dependent Part (Proof of Lemma~\ref{lem:sample-tree-expan-dep})}\label{sec:proof-sample-tree-exp-dep}

The goal of this section is to upper bound the expectation of the data-dependent part. Similarly to Section~\ref{sec:data-ind-part-bound-1}, we decompose the randomness of $\bT \sim \SampleTree(\mu, m)$ into two independent sources: the draw of $\by_1,\dots,\by_m \sim \mu$ and $\bT' \sim \quadtree(\bhomg)$, and we seek to show 
\[ \Ex_\bT\left[\Qdep_\bT(x,y)\right] = \Ex_{\by_1,\dots,\by_m}\left[\Ex_{\bT'}\left[\Qdep_{\bT'}(x,y)\right]\right] \leq \tilde{O}(\log (msd)) \cdot \EMD(x,y). \]
It will suffice to upper bound the inner-most expectation over $\bT'$, and treat the sampled points $\by_1,\dots,\by_m$ as deterministic variables $y_1,\dots, y_m$ (thereby removing the boldness). %, which implies $\bhomg$ is fixed to some $\hat{\Omega}$ and we can write $d_{\bT'}(.,.)$ to denote $d_\bT(.,.)$. 

%Recall from (\ref{eq:decompose-dt}) that we may expand $\Qdep_{\bT'}(x,y)$ and express it in terms of distances in the tree $\bT'$ as:
%\[ \Qdep_{\bT'}(x,y) = \sum_{i=1}^s \]

%Notice that $\Qdep_{\bT'}(x,y) \leq \sum_{i=1}^s d_{\bT'}(a_i,b_{\pi(i)})$ (See \ref{eq:decompose-dt}). Therefore, it is tempting to apply Lemma~\ref{lem:cjlw} with parameter $|\Omega| = |\bhomg| \leq m\cdot s\cdot(d+1)$ to claim $\sum_{i=1}^s d_{\bT'}(a_i,b_{\pi(i)}) \leq \sum_{i=1}^s\Tilde{O}(\log(msd))\cdot \|a_i-b_{\pi(i)}\|_1$. However, Lemma~\ref{lem:cjlw} only works for elements in $\bhomg$ and elements of interest ($a_1,\dots,a_s,b_1,\dots,b_s$) are not necessarily in $\bhomg$ unless $x,y$ are among $y_1,\dots,y_m$ (See Figure~\ref{fig:quadtree-embed}). Fortunately, it turns out that Lemma~\ref{lem:cjlw} can be modified into a favorable version which allows us to upper bound $d_{\bT'}(a,b)$ for any $a,b \in \{0,1\}^d$.
\begin{lemma}
\label{lem:cjlw-mod}
    Let $\Omega \subseteq \{0,1\}^d$ be any set of $m$ elements, and $\T'$ be drawn from $\quadtree(\hat{\Omega})$. \emph{For any two elements $a, b \in \{0,1\}^d $}, we have
    $$\E_{\bT'}\left[d_{\bT'}(a,b)\right] \leq \Tilde{O}(\log(m) + \log(d) ) \cdot \|a - b\|_1$$
\end{lemma}
\begin{proof}
Our analysis will proceed by considering two (correlated) trees $(\bT', \tilde{\bT})$ defined by draws to
\[ \bT' \sim \quadtree(\hat{\Omega}) \qquad \text{and}\qquad \tilde{\bT} \sim \quadtree(\hat{\Omega} \cup \{ a, b\}) \]
with the same hash functions (recall from Figure~\ref{fig:DDquadtree-prelims} that the draw of hash functions was independent of $\hat{\Omega}$ or $\hat{\Omega} \cup \{a, b\}$). Since we included $\{a, b\}$ into the generation of the tree $\tilde{\bT}$, we can safely apply Lemma~\ref{lem:cjlw}, where the number of elements which generate the tree is $m(d+1) + 2$,
\[ \Ex_{\tilde{\bT}}\left[ d_{\tilde{\bT}}(a, b) \right] \leq \tilde{O}(\log(m) + \log(d)) \cdot \|a - b\|_1.\]
It suffices, therefore, to show that for every $\ell$, letting $\bw'_{\Dep}(v_{\ell}(a), v_{\ell+1}(a))$ and $\tilde{\bw}_{\Dep}(v_{\ell}(a), v_{\ell+1}(a))$ denote the weights on $\bT'$ and $\tilde{\bT}$ on the $(\ell+1)$-th edge of the root-to-$a$ path, that
\begin{align} \bw'_{\Dep}(v_{\ell}(a), v_{\ell+1}(a)) \leq 9 \cdot \tilde{\bw}_{\Dep}(v_{\ell}(a), v_{\ell+1}(a)), \label{eq:desired-weight-comp}
\end{align}
and that the analogous expression holds for $b$. This would conclude the argument, as it would imply that $d_{\bT'}(a, b) \leq 9 \cdot d_{\tilde{\bT}}(a,b)$, because $\Split_{\ell+1}(a, b)$ depend solely on the hash functions, which are identical in $\bT'$ and $\tilde{\bT}$. 

So, consider a fixed setting of $\ell$, and let:
\begin{align*}
    A &= \bElms(v_{\ell}(a)) \cap (\hat{\Omega} \cup \{ a, b\}) \\
    B &= \bElms(v_{\ell+1}(a)) \cap (\hat{\Omega} \cap \{ a, b\})\\
    C &= A \setminus \{a,b\} \\
    D &= B \setminus \{a,b\}. 
\end{align*}
Recall that, with the above notation, we have that $\bw'_{\Dep}(v_{\ell}(a), v_{\ell+1}(a)) = 0$ in the case $D = \emptyset$, in which case, (\ref{eq:desired-weight-comp}) is trivially satisfied. Otherwise, $D \neq \emptyset$ which implies $B \neq \emptyset$, and both weights are determined by:
    \begin{align*}
        \bw'_{\Dep}(v_{\ell}(a), v_{\ell+1}(a)) &= \Ex_{\substack{\bc \sim C\\ \bc' \sim D}}\left[\|\bc - \bc'\|_1 \right]\qquad\text{and}\qquad
        \tilde{\bw}_{\Dep}(v_{\ell}(a), v_{\ell+1}(a)) = \Ex_{\substack{\tilde{\bc} \sim A \\ \tilde{\bc}' \sim B}}\left[\|\tilde{\bc} - \tilde{\bc}'\|_1 \right]. 
    \end{align*}
Note that, any $c \in C$ and $c' \in D$ which appears in the expectation on the left-hand side also appears in the right-hand side, where the term appearing is $\| c - c'\|_1 / (|C| \cdot |D|)$ on the left-hand side, and $\|c - c'\|_1 / (|A| \cdot |B|)$ on the right-hand side. However, we also have $1 \leq |C| \leq |A| + 2$ and $1 \leq |D| \leq |B|+2$, which means that 
\[ \frac{1}{|C| \cdot |D|} \leq \frac{9}{|A| \cdot |B|}, \]
and therefore, we obtain  (\ref{eq:desired-weight-comp}).
\end{proof}

With Lemma~\ref{lem:cjlw-mod}, we conclude the proof of Lemma~\ref{lem:sample-tree-expan-dep}.

\begin{proof}[Proof of Lemma~\ref{lem:sample-tree-expan-dep}]
First, notice from (\ref{eq:decompose-dt}) and the definition of $\Qdep_{\bT}(x,y)$, that it suffices to upper bound $\sum_{i=1}^s d_{\bT'}(a_i, b_{\pi(i)})$. Thus, we apply Lemma~\ref{lem:cjlw-mod} with $|\Omega| = |\bhomg| \leq m \cdot s \cdot (d+1)$, and finish the proof:
\begin{align*}
\Ex_{\by_1,\dots,\by_m}\left[\Ex_{\bT'}\left[\sum_{i=1}^s d_{\bT'}(a_i,b_{\pi(i)})\right]\right]
  \overset{(\ref{lem:cjlw-mod})}
  &{\leq} \Ex_{\by_1,\dots,\by_m}\left[\sum_{i=1}^s\Tilde{O}(\log(m\cdot s \cdot (d+1)) + \log d)\cdot \|a_i - b_{\pi(i)}\|_1\right] \\
  &=\Tilde{O}(\log(msd))\cdot \EMD(x,y).
\end{align*}
\end{proof}
%  Now, if we decompose the randomness of $\SampleTree$ into the randomness of $\hat{\bOmega}$ and randomness from the following $\quadtree$ subroutine, then we can get
% \begin{align*}
%   \E_{\bS \sim \calS_{\hat{\bOmega}}}\left[d_{\bS}(a,b)\right] 
%   = & \E_{\hat{\bOmega} \sim \Omega^m} \left[\E_{\bQ \sim \calQ_{\nbr\left(\hat{\bOmega}\right)}}\left[d_{\bQ}(a,b) \right] \right] \\
%   \leq & \Tilde{O}\left(\log(s) + \log(d) \right) \cdot \|a - b\|_1 \tag{use \ref{ineq:apply-cjlw} for all $\hat{\bOmega}$}
% \end{align*}

% Finally, we have
% \begin{align*}
%     \sum_{i = 1}^s \Ex_{\bS \sim \calS_{\Omega}}\left[ d_{\bs}(a_i,b_{\pi(i)}) \right] \leq &  \Tilde{O}(\log(s) + \lod(d)) \cdot \sum_{i = 1}^s \|a_i - b_{\pi(i)}\|_1 \\
%     = & \Tilde{O}(\log(s) + \log(d)) \cdot \EMD(x_0,y)
% \end{align*}

%% file: data-ind-hash-and-SampleTree/st-contraction.tex
%\subsection{Contraction of $\SampleTree$ Embedding}
% \subsection{Proof of Lemma \ref{lem:non-contraction}}
%In this section, we will show $\SampleTree$ has no contraction with high probability. Formally, we will show
%\begin{lemma}
%\label{lem:sampletree-contraction}
%For any pair $x,y \in \Omega \subset \EMD_s(\{0,1\}^d)$ and $m = \poly(s)$, we have
%\begin{align*}
%  \Pr_{\bS \sim \calS_{\Omega, m}}[D_{\bS}(x,y) \geq \EMD(x,y)] = 1 - O\left(\frac{1}{s}\right)
%\end{align*}  
%\end{lemma}
 
% It suffices to show the following 
% \begin{claim}
% \label{claim: non-contraction-for-any-two}
% For any $a,b \in \{0,1\}^d$, with probability at least $1 - O(\frac{1}{s^2})$,
% \begin{align*}
% &\sum_{h=0}^{L} \ind\left[a,b \text{ split at depth } h \right] \cdot ( &\DataDep(a,h) \cdot \Weight(v(a,h)) + \DataDep(b,h) \cdot \Weight(v(b,h)) \\
% &&+ \DataInd(a,h) \cdot \frac{d}{2^h} \cdot \log s + \DataInd(b,h) \cdot \frac{d}{2^h} \cdot \log s )  \\
% &\geq \|a - b\|_1 & 
% \end{align*}
% \end{claim}

In order to prove Lemma~\ref{lem:sample-tree-contr}, we claim that it suffices to prove the following lemma, which shows that for any two elements $a, b \in \{0,1\}^d$, the probability over $\bT$ that $d_{\bT}(a, b) \leq \|a - b\|_1$ is vanishingly small. Then, the desired lemma follows from a union bound and the proper setting of $\xi$.
\begin{lemma}
\label{lem:sample-tree-contr-2}
    Fix any $a ,b \in \{0,1\}^d$ and any $\delta \in (0, 1)$, and let $\bT$ be generated from $\SampleTree(\mu, m)$ with parameter $\xi = O(\log(msd/\delta))$ (for a large enough constant factor). Then, 
    \begin{align*}
        \Prx_{\bT}\left[ d_{\bT}(a, b) \leq \|a - b \|_1 \right] \leq \frac{\delta}{s^2}
    \end{align*}
\end{lemma}

% To prove Lemma \ref{lem:sampletree-contraction}, it suffices to show that, for any pair of elements $a,b \in \{0,1\}^d$ and $m = \poly(s)$, $\Pr_{\bS \sim \calS_{\Omega}, m}[d_\bS(a,b) \geq \|a-b\|_1] = 1 - \frac{1}{s^2}$. Then, we can apply union bound over all pairs $(a_i,b_i), i \in [s]$. 
\begin{proof} [proof of Lemma \ref{lem:sample-tree-contr} assuming Lemma \ref{lem:sample-tree-contr-2}]
    Let $x = \{a_1,a_2,\dots,a_s\}, y = \{b_1,b_2,\dots,b_s\} \in \EMD_s(\{0,1\}^d)$ be a pair of points. Let $\bsigma: [s] \to [s]$ be the matching such that $\EMD_\bT(x,y) = \sum_{i=1}^s d_\bT(a_i,b_{\bsigma(i)})$. 

    Applying lemma \ref{lem:sample-tree-contr-2} and union bound over all possible $s^2$ pairs $a.b$ such that $a\in x, b \in y$ gives 
    \[\Prx_\bT[\forall a \in x, b \in y, d_\bT(a,b) \geq \|a-b\|_1] \geq 1 - \sigma. \] 
    Then it suffices to show $\EMD_\bT(x,y) \geq \EMD(x,y)$ given that $\forall a \in x, b \in y, d_\bT(a,b) \geq \|a-b\|_1$. It is clear that  
    \begin{align*}
        \EMD_\bT(x,y) =\sum_{i = 1}^s d_\bT(a_i, b_{\bsigma(i)}) 
        \geq  \sum_{i = 1}^s \|a_i - b_{\bsigma(i)}\|_1 
        \geq \EMD(x,y). 
    \end{align*}
\end{proof}

The remainder of this section is devoted to the proof of Lemma \ref{lem:sample-tree-contr-2}. We first introduce the following helpful lemma, which indicates the probability over $\bT$ that the tree metric has contraction vanishes quickly as the constant factor in $\xi$ increases if all edge weights are data-independent. Recall that $\Split_\ell(x,y) \in \{0,1\}$ is the indicator variable for the event that $v_\ell(x) \neq v_\ell(y)$.
% h -> l; prove 7.4 assuming 7.5; remind what is split; prove 7.6; randomness: first apply union bound over \omega x \omega to remove the randomness of sampletree.  bold v's

% \begin{lemma}
%     $\forall a, b \in \{0,1\}^d$, let $h_{a,b} = \log \frac{d}{\|a-b\|_1}$. For any constant $c > 0$, we have 
%     $$\Pr\left[\Split_{h_{a,b} + \log\left(c \cdot \log s  \right)}(a,b) = 0\right] \leq \frac{1}{s^c}$$ 
% \end{lemma}

% \begin{corollary}
% \label{cor:unlike-to-contract}
%     For any $a,b \in \{0,1\}^d$ and any constant $c > 0$, with probability at least $1 - \frac{1}{s^c}$ 
%     $$\sum_{h=0}^{L} \Split_h(a,b) \cdot \frac{d}{2^h} \cdot c \cdot \log s \geq \|a - b\|_1$$
% \end{corollary}

\begin{lemma} 
\label{lem:unlike-to-contract}
    For any $a,b \in \{0,1\}^d$ and any $\rho > 0$, we have 
    $$\Prx_{\bT}\left[\sum_{\ell=0}^{L} \Split_{\ell + 1}(a,b) \cdot \frac{d}{2^{\ell}} \cdot \log\frac{1}{\rho} \leq \|a - b\|_1\right] \leq \rho$$
\end{lemma}

\begin{proof}
    Since the indicator $\Split_\ell(a,b)$ is non-decreasing with respect to $\ell$, the event that $\sum_{\ell=0}^{L} \Split_{\ell+1}(a,b) \cdot \frac{d}{2^{\ell}} \cdot \log\frac{1}{\rho} \leq \|a - b\|_1$ happens only if $\Split_{\ell_0 + 1} = 0$ where $\ell_0 = \lceil \log \frac{d}{\|a-b\|_1}+ \log\log\frac{1}{\rho} \rceil$ (so that $\sum_{\ell=0}^{L} \Split_{\ell+1}(a,b) \cdot \frac{d}{2^{\ell}}\log\frac{1}{\rho} 
    = \sum_{\ell=\ell_0 + 1}^{L} \frac{d}{2^{\ell}}\log\frac{1}{\rho} \leq \|a-b\|_1$). Therefore, we have
    \begin{align*}
        \Prx_{\bT}\left[\sum_{\ell=0}^{L} \Split_{\ell + 1}(a,b) \cdot \frac{d}{2^{\ell}} \cdot \log\frac{1}{\rho} \leq \|a - b\|_1\right] \leq & \Prx_{\bT}\left[\Split_{\ell_0 + 1} = 0\right] \\
    \end{align*}
    It suffices to upper-bound $\Prx_{\bT}\left[\Split_{\ell_0 + 1} = 0\right]$ by $\rho$. Notice that the number of coordinates that have been sampled by (including) depth $\ell \in \{0,1,\dots,L\}$ is $2^{\ell+1} - 1$. $\Split_{\ell_0 + 1} = 0$ is equivalent to that $a,b$ agree on all $2^{\ell_0 + 1} - 1$ coordinated sampled by (including) depth $\ell_0$. It holds that 
    \begin{align*}
        \Prx_{\bT}\left[\Split_{\ell_0 + 1} = 0\right] &\leq \left(1 - \frac{\|a-b\|}{d}\right)^{2^{\ell_0 + 1} - 1}\\
        &\leq \left(1 - \frac{\|a-b\|}{d}\right)^{2^{\ell_0}} \\
        &\leq \left(1 - \frac{\|a-b\|}{d}\right)^{2^{\log \frac{d}{\|a-b\|_1}+ \log\log\frac{1}{\rho}}}\\
        &\leq \rho 
    \end{align*}
\end{proof}

\begin{proof}[proof of Lemma \ref{lem:sample-tree-contr-2}]
    Recall that $\bhomg$ is the set of elements in the points sampled by $\bT$ and $|\bhomg| \leq ms(d+1)$ (see Figure \ref{fig:quadtree-embed}). We set $\xi = c \cdot \log \frac{msd}{\delta}$ where$c$ is a parameter to be set later. Fix any $a,b \in \{0,1\}^d$, let the shortest path between $a,b$ in $\bT$ be $\calP: a, \bv_1,\bv_2,\cdots, \bv_k, b$. Recall that we say an edge $(u,v)$ in $\bT$ is data-independent if its weight $w(u,v)$ of evaluates to $\frac{d}{2^\ell}\cdot \xi$ where $\ell $ is the depth of $u$, otherwise it is data-dependent. We prove the lemma by the following three cases:
    % (i.e., it evaluates to $\Ex_{\be,\be'}\left[\|\be - \be'|_1\right]$ where $\be \sim \Elms(u, \bhomg), \be' \sim \Elms(v,\bhomg)$). 

    \begin{itemize}
        \item  If for all $i \in [k-1]$, edge $(v_i , v_{i+1})$ is data-dependent, then by triangular inequality we have 
        \begin{align*}
            \|a - b\|_1 \leq & \Ex_{e_i \sim \Elms(v_i) \cap \bhomg, i \in [k]}\left[\|a - e_1\|_1 + \|e_1 - e_2\|_1 + \cdots + \|e_k - b\|_1\right] \\
            = & d_\bT(a,b)
        \end{align*}
        where the equality follows from the fact that the identity mapping is used at depth $L + 1$, thus all points at a leaf must be identical.
        \item  If for all $i \in [k-1]$, edge $(v_i , v_{i+1})$ is data-independent, we know 
        $$d_{\bT}(a,b) = 2\cdot \sum_{\ell=0}^{L} \Split_{\ell+1}(a,b) \cdot \left( \frac{d}{2^{\ell}} \cdot c \cdot \log \frac{msd}{\delta} \right) $$
        we are able to apply Lemma \ref{lem:unlike-to-contract} with $\rho = \left(\frac{\delta}{msd}\right)^{2c}$ and get
        \begin{align*} 
        \Prx_{\bT}\left[d_{\bT}(a,b) \leq \|a - b\|_1\right] \leq \left(\frac{\delta}{msd}\right)^{2c}
        \end{align*}
        As long as $c \geq 1$, the Lemma to prove holds for this case since $m,d \geq 1$ and $\delta \leq 1$. 
        \item If neither of the above is the case, there must exist $v_{i_1}$ and $v_{i_2}$ in path $\calP:v_1,\dots, v_{i_1}, \dots, v_{i_2}, \dots, v_k$ such that edges among $v_1,v_2,\dots, v_{i_1}$ and among $v_{i_2}, \dots,v_{k-1}, v_k$ are data-independent, and edges among $v_{i_1},\dots, v_{i_2}$ are data-dependent. Let $\ell_1,\ell_2$ be the depth of $v_{i_1},v_{i_2}$ respectively. 
        
        By applying lemma \ref{lem:unlike-to-contract} for $a, p$ as well as $b,p'$ with $\rho = \left(\frac{\delta}{msd}\right)^c$ and union bound over all possible pairs $(p,p') \in \bhomg \times \bhomg$, we have that for any $(p,p') \in \bhomg \times \bhomg$, it holds that
        \begin{align*}
            \Prx_{\bT}\left[\begin{array}{c}
                   \|a - p\|_1 \geq \sum_{\ell=0}^{L} \Split_{\ell+1}(a,p) \cdot \left( \frac{d}{2^{\ell}} \cdot c \cdot \log \frac{msd}{\delta} \right) \lor\\
                  \|b-p'\|_1 \geq \sum_{\ell=0}^{L} \Split_{\ell+1}(b,p') \cdot \left( \frac{d}{2^{\ell}} \cdot c \cdot \log \frac{msd}{\delta}\right)
            \end{array}  \right] \leq (m\cdot s\cdot (d+1))^2 \cdot \left(\frac{2\delta}{msd}\right)^c
        \end{align*}
        
        Therefore, for any $p_1 \in \Elms(v_{i_1}, \bhomg), p_2 \in \Elms(v_{i_2}, \bhomg)$, with probability at least $1 - (m\cdot s\cdot  (d+1))^2 \cdot \left(\frac{\delta}{msd}\right)^c$ over a draw of $\bT$, we have 
        \begin{align*}
            &\|a - p_1\|_1 \leq \sum_{\ell=0}^{L} \Split_{\ell+1}(a,p_1) \cdot \left( \frac{d}{2^{\ell}} \cdot c \cdot \log \frac{msd}{\delta} \right) \land \\
                  &\|b-p_2\|_1 \leq \sum_{\ell=0}^{L} \Split_{\ell+1}(b,p_2) \cdot \left( \frac{d}{2^{\ell}} \cdot c \cdot \log \frac{msd}{\delta}\right).
        \end{align*}
        Thus, also with probability at least $1 - (m\cdot s\cdot  (d+1))^2 \cdot \left(\frac{\delta}{msd}\right)^c$, it holds that            
        \begin{align*}
            \|a-b\|_1 \leq & \|a - p_1\|_1 + \|p_1- p_2\|_1 + \|p_2 - b\|_1 \\
            \leq &\sum_{\ell=0}^{L} \Split_{\ell+1}(a,p) \cdot \left( \frac{d}{2^{\ell}} \cdot c \cdot \log \frac{msd}{\delta} \right)  \\ 
            & + \Ex_{e_i \sim \Elms(v_{i_1 + i}) }\left[\|p_1 - e_1\|_1 + \|e_1 - e_2\|_1 + \dots + \|e_{i_2 - i_1 - 1} - p_2\|_1 \right] \\ & + \sum_{\ell=0}^{L} \Split_{\ell+1}(b,p') \cdot \left( \frac{d}{2^{\ell}} \cdot c \cdot \log \frac{msd}{\delta}\right) \\
            = & d_\bT(a,b)  
        \end{align*}
        The existence of $c$ such that $c \geq 1$ and $(m\cdot s \cdot (d+1))^2 \cdot \left(\frac{\delta}{msd}\right)^c \leq \frac{\delta}{s^2}$ completes the proof.
\end{itemize}

\end{proof}

%% file: data-dep-lb.tex
%!TEX root = main.tex

\section{Data-Dependent Hashing and Sketching Lower Bounds}\label{sec:dd-lb}

We will now show that the data-dependent LSH (Definition~\ref{def:data-dep}) construction from Theorem~\ref{thm:data-dep-hashing} has an approximation factor of $\tilde{O}(\log s)$ which is best possible (up to the $\poly(\log \log s)$ factors in the $\tilde{O}$) when $p_1$ and $p_2$ are constant. We do this by reducing data-dependent LSH to sketching lower bounds, and apply the lower bound on~\cite{AIK08}. Specifically, recall the set-up of communication complexity for sketching lower bounds.
\begin{definition}[EMD Sketching and Distributional EMD Sketching]\label{def:sketching}
For every $s, d \in \N$ and every $r > 0$ and $c > 1$, we consider the communication complexity of the following partial function, whose inputs are sets $x, y \in \EMD_s(\{0,1\}^d)$ which satisfies:
\begin{align*}
    F(x, y) = \left\{ \begin{array}{cc} 1 & \EMD(x,y) \leq r \\
                 0 & \EMD(x, y) > cr \end{array} \right. .
\end{align*}
In the EMD sketching communication problem, we assume that a player Alice receives as input $x \in \EMD_s(\{0,1\}^d)$ and Bob receives an input $y \in \EMD_s(\{0,1\}^d)$, and they must design a public-coin communication protocol $\Pi$ whose outputs align with $F$ (whenever $x, y$ satisfy the two promises) with probability at least $2/3$, and which minimizes the communication.

Furthermore, we define the distributional version of the EMD sketching problem to be the same as above, but when there is ``far'' distribution $\mu$, known to both Alice and Bob, such that the inputs $(x,y)$ satisfy that  either \textbf{(1)} $x,y$ are arbitrary such that $\EMD(x,y) \leq r$ and the protocol should output $1$, or \textbf{(2)} the inputs $x,y\sim \mu$ are drawn independently from $\mu$ and whenever $\EMD(x,y) \geq cr$ the algorithm should output $0$ . Whenever $\EMD(x,y) \leq r$ or $\EMD(x,y) > cr$, then the communication protocol must be correct with probability $2/3$ over it's own randomness, and over the randomness of $x,y \sim \mu$ (if this inputs come from case \textbf{(2)}), and the output is allowed to be arbitrary if $r < \EMD(x,y) \leq cr$.
\end{definition}

In~Theorem~4.1 of~\cite{AIK08}, the authors show a communication complexity lower bound for the above problem, showing that, for every dimension $d \geq 1$ and any approximation ratio $1 \leq c \leq d$, if $\Pi$ is a randomized communication protocol for $F$ on $\EMD_s(\{0,1\}^d)$ for $s = 2^{\Theta(d)}$, then the communication complexity at least $\Omega(d/c)$, which also implies the lower bound of $\Omega(\log s / c)$. In particular, any $O(1)$-bit communication protocol $\Pi$ which computes $F$ must do so with approximation $c = \Omega(\log s)$. Inspecting the proof of~\cite{AIK08} (and in particular, the distribution over inputs used to derive the lower bound), one sees that they prove the following (stronger formulation) of Theorem~4.1, which applies when the points $x,y$ are drawn \textit{independently} from a known distribution $\mu$ in the far case.

\begin{theorem}[Theorem~4.1 and Lemma~4.8 of~\cite{AIK08}]\label{thm:aik-lb}
    For any $d \in \N$ and $1 \leq c \leq d$, there exists a distribution $\mu$ supported on $\EMD_s(\{0,1\}^d)$ with $s = 2^{\Theta(d)}$ with the following properties:
    \begin{itemize}
        \item If $\bx, \by \sim \mu$ are drawn independently, then $\EMD(\bx, \by) \geq sd/100$ with probability at least $1 - 2^{-\Omega(d)}$.
        \item There is another distribution $\rho$ supported on pairs $\EMD_s(\{0,1\}^d) \times \EMD_s(\{0,1\}^d)$ for which $(\bx,\by) \sim \rho$ satisfies $\EMD(\bx, \by) \leq sd/(100 c)$ with probability at least $1 - 2^{-\Omega(d/c)}$. 
    \end{itemize}
    For any function $f \colon \EMD_s(\{0,1\}^d) \to \{0,1\}$, 
    \begin{align*}
        \Prx_{\bx,\by \sim \mu}\left[ f(\bx) = f(\by) \right] + \Prx_{(\bx,\by) \sim \rho}\left[ f(\bx) \neq f(\by)\right] \geq 1 - 2^{-\Omega(d/c)}.
    \end{align*}
\end{theorem}

From the above theorem, we show that any data-dependent LSH for $\EMD$ which is $(r, cr, p_1,p_2)$-sensitive with a constant setting of $0 < p_2 < p_1 < 1$ must incur the factor of $\log s$ in the approximation. This is because such a LSH can easily been seen to solve the  distributional variant of sketching EMD, by constructing the LSH dependending on the known ``far'' distribution $\mu$. Specifically, using this fact yields the following.

\begin{theorem}\label{thm:dd-lb}
    Consider any fixed constants $0 < p_2 < p_1 < 1$, and suppose there exists some $c > 1$ such that, for all $s,d \in \N$ and $r > 0$, there is a data-dependent LSH which is $(r, cr, p_1, p_2)$-sensitive for $\EMD_s(\{0,1\}^d)$. Then, $c = \Omega(\log s)$.
\end{theorem}

\begin{proof}
Consider a data-dependent hash family $\calH$ for $\EMD$ which is $(r, cr, p_1, p_2)$-sensitive for  for $\mu$, where $r = sd/(100c)$. Then, consider the distribution over Boolean functions $\boldf \colon \EMD_s(\{0,1\}^d) \to \{0,1\}$ given by \textbf{(i)} first hashing $\EMD_s(\{0,1\}^d)$ according $\bh \sim \calH$, and then \textbf{(ii)} choosing, for each bucket independently, whether to have $\boldf$ assign every point in that bucket to $1$ with probability $\alpha = 1/2$ (and otherwise $0$). Then, by Definition~\ref{def:data-dep}, we have
\begin{align*}
&\Ex_{\boldf}\left[ \Prx_{\bx, \by \sim \mu}\left[ \boldf(x) = \boldf(y)\right] + \Prx_{(\bx,\by) \sim \rho}\left[ \boldf(x) \neq \boldf(y)\right] \right] \\
&\qquad \leq \alpha^2 + (1 - \alpha)^2 + 2\alpha(1-\alpha)\left( 2^{-\Omega(d)} + \Ex_{\bx \sim \mu}\left[\Prx_{\substack{\bh \sim \calH \\ \by \sim \mu}}\left[\begin{array}{c} \EMD(\bx,\by) \geq sd/100 \\ \bh(\bx) = \bh(\by) \end{array} \right] \right] \right) \\
&\qquad\quad + 2\alpha(1-\alpha) \left( \Prx_{\substack{\bh \sim \calH \\ (\bx,\by) \sim \rho}}\left[ \bh(x) \neq \bh(y) \mid \EMD(x,y) \leq sd/(100c) \right] + 2^{-\Omega(d/c)} \right) \\
&\qquad\leq 1 + 2\alpha(1-\alpha) (p_2 - p_1) + 2^{-\Omega(d)} + 2\alpha(1-\alpha) \cdot 2^{-\Omega(d/c)}.
\end{align*}
So, there exists a Boolean function $f$ which is below the above expectation. By Theorem~\ref{thm:aik-lb}, this quantity must be at least $1 - 2^{-\Omega(d/c)}$, and hence
\[ \Omega(1) \leq p_1 - p_2 \leq 2^{-\Omega(d/c)} \leq 2^{-\Omega(\log s/c)} \]
and therefore, $c = \Omega(\log s)$.
\end{proof}

%% file: data-dep-hash-to-nns-proof.tex
%!TEX root = main.tex

\section{Data-Dependent LSH to ANN: Proof of Theorem~\ref{thm:hashing-to-nn}}\label{sec:dd-hash-to-ann}

The proof of Theorem~\ref{thm:hashing-to-nn} proceeds by executing multiple ``core'' data structures which output a dataset point and succeed at finding an approximate near neighbor with a small (but non-trivial) probability, just like in~\cite{IM98, HIM12}. We first describe the ``core'' data structure, $\CorePreprocess$ and $\CoreQuery$ in Figure~\ref{fig:core-preprocess} and Figure~\ref{fig:core-query}, which we show succeed with probability at least $p_1 n^{-\rho}$. By repeating $O(n^{\rho} / p_1)$ times, we amplify the success probability to $9/10$.

\begin{figure}[H]
\label{alg:preprocess}
	\begin{framed}
		\noindent Subroutine $\CorePreprocess(P, k)$
		
		\begin{flushleft}
			\noindent {\bf Input:} A dataset $P \subset X$, and a positive integer $k \in \N$.
			
			\noindent {\bf Output:} The pointer to a data-structure node $v$. % or ${\reject}$. 
			
			\begin{itemize}
				\item Initialize a data structure node $v$. Sample $\bp \sim P$ and store it in $v.\textit{point}$.
				\item If $k = 0$, store the dataset $P$ in $v.\textit{data}$ and return $v$. 
				\item If $k > 0$, perform the following:
				\begin{itemize}
				\item Execute the initialization algorithm to maintain a $(r, cr, p_1,p_2)$-sensitive hash family for the uniform distribution over $P$. % $\GenHash(P)$, and let $\bh$ denote the description to a hash function $\bh \colon X \to U$ from $\calD$ which is $(r, cr,p_1,p_2)$-sensitive for the uniform distribution over $P$. 
                Store a pointer to this data structure in $v.h$, which holds a draw to $\bh$.
				\item For every $p \in P$, query the data structure in $v.h$ with $p$ to compute $\bh(p)$. For every $u \in U$ for which there exists $p \in P$ where $\bh(p) = u$, let $P_u$ denote the set of points $p \in P$ where $\bh(p) = u$. 
				\item For each non-empty $P_u$, execute $\CorePreprocess(P_u, k-1)$ and store the data structure node as a child $v.u$ of $v$. 
				\item Return $v$. 
				\end{itemize}
			\end{itemize}
		\end{flushleft}
	\end{framed}
	\caption{The $\CorePreprocess$ Algorithm.}\label{fig:core-preprocess}
\end{figure}

\begin{figure}[H]
\label{alg:query}
	\begin{framed}
		\noindent Subroutine $\CoreQuery(q, v)$
		
		\begin{flushleft}
			\noindent {\bf Input:} A point $q \in X$ and a data structure node $v$ from $\CorePreprocess(P,k)$, for some $k$. \\
			\noindent {\bf Output:} A point $p \in P$, or ``fail.''% or ${\reject}$. 
			
			\begin{itemize}
				\item Let $p$ be the point stored in $v.\textit{point}$. Compute $d_X(p, q)$ and return $p$ if the distance is at most $cr$. 
				\item If $v.\textit{data}$ contains a set of points $P$ (i.e., it is a leaf node), scan for the first $\hat{p} \in P$ where $d_X(\hat{p}, q) \leq cr$ and return $\hat{p}$. If no such points are found, output ``fail.''
				\item Otherwise, $v.\textit{data}$ is empty and $v.h$ contains a data structure computing a hash function $\bh$. Query the data structure to compute $\bh(q)$ and let $u$ denote its output. If $v.u$ is empty, output ``fail,'' and otherwise, output $\CoreQuery(q, v.u)$.
			\end{itemize}
		\end{flushleft}
	\end{framed}
	\caption{The $\CoreQuery$ Algorithm.}\label{fig:core-query}
\end{figure}

The following claim, which upper bounds the preprocessing time of $\CorePreprocess$, is straight-forward. We simply bound, for each point $p \in P$, the number of times it evaluates a hash function maintained by a data structure, and the number of times that an initialization procedure of a hash function is called with a dataset containing $p \in P$. Both quantities are easily seen to be at most $k$ on each dataset, and this gives the desired bound.

\begin{claim}[Preprocessing Time of $\CorePreprocess$]\label{cl:preprocess-time}
For any dataset $P \subset X$ of $n$ points the algorithm $\CorePreprocess(P, k)$ runs in time $O\left(nk \cdot (I_{\sfh}(n) + Q_{\sfh}(n))\right)$.
\end{claim}

\begin{claim}[Success Probability $\CorePreprocess$ and $\CoreQuery$]\label{cl:success-prob}
For any dataset $P \subset X$ of $n$ points and any query $q \in X$. If there exists $p \in P$ with $d_X(p, q) \leq r$, then for any $k \in \N$,
\begin{align*}
\Prx\left[ \CoreQuery(q, \bv) \text{ doesn't fail when } \bv \leftarrow \CorePreprocess(P, k) \right] \geq p_1^k.
\end{align*}
\end{claim}

\begin{proof}
The proof is a straight-forward induction on $k$ using the definition of data-dependent hashing with $(r, cr, p_1, p_2)$-sensitive hash functions. Suppose for an inductive hypothesis that for some integer $k_0 \geq 0$, whenever there exists $p \in P$ with $d_X(p, q) \leq r$, the probability that an execution of $\bv_0 \leftarrow \CorePreprocess(P, k_0)$ and $\CoreQuery(q, \bv_0)$ outputs an approximate near neighbor is at least $p_1^{k_0}$. Note that the base case of $k_0 = 0$ is trivial, since $\bv_0 \leftarrow \CorePreprocess(P, 0)$ stores all of $P$ in $\bv_0.\textit{data}$ and this is scanned by $\CoreQuery(q, \bv_0)$. If we execute $\bv \leftarrow \CorePreprocess(p, k_0 + 1)$ then we can lower bound the probability that $\CoreQuery(q, \bv)$ outputs an approximate near neighbor by considering the following event. 

Suppose that, when we execute $\bv \leftarrow \CorePreprocess(P, k_0 + 1)$, the following occurs. 
\begin{enumerate}
\item First, we generate a hash function $\bh \colon X \to U$ which is stored in $\bv.h$, and we happen to satisfy $\bh(p) = \bh(q)$.  So, letting $u = \bh(p)$, the call to $\CorePreprocess(P, k_0 + 1)$ recursively executes $\bv_0 \leftarrow \CorePreprocess(P_u, k_0)$, where $p \in P_u$ and $\bv_0$ is stored in $\bv.u$. 
\item We note furthermore that $\CoreQuery(q, \bv)$ will evaluate the hash function $\bh(q)$ and will have $\bh(q) = u$, so it will return $\CoreQuery(q, \bv_0)$ where $\bv_0 = \bv.u$. If, the call to $\CoreQuery(q, \bv_0)$, where $\bv_0$ is generated from $\CorePreprocess(P_u, k_0)$ succeeds, then $\CoreQuery(q, \bv)$ succeeds.
\end{enumerate}
Since the hash function $\bh$ is sampled from a $(r, cr, p_1, p_2)$-sensitive hash family, $\bh(p) = \bh(q)$ with probability at least $p_1$. By the inductive hypothesis, the call $\CoreQuery(q, \bv_0)$ succeeds with probability $p_1^{k_0}$, and hence we succeed with probability at least $p_1^{k_0 + 1}$, completing the inductive claim.
\end{proof}

\begin{claim}[Query Time of $\CoreQuery$]\label{cl:query-time}
For any dataset $P \subset X$ and any query $q \in X$ let $P_f(q) \subset P$ be
\[ P_f(q) = \left\{ p \in P : d_X(p, q) > cr \right\}. \]
The expected running time of $\CoreQuery(q, \bv)$ where $\bv\leftarrow \CorePreprocess(P, k)$ is at most
\[ O\left(k \cdot (Q_{\sfh}(n) + 1) + |P_f(q)| \cdot p_2^k\right).\]
\end{claim}

\begin{proof}
Similarly to Claim~\ref{cl:success-prob}, we claim this by induction on $k$. 
The base case of $k = 0$ is trivial, as all points in $P$ are stored in $\bv.\textit{data}$ when $\bv \leftarrow \CorePreprocess(P,0)$. Therefore, the time to scan $\bv.\textit{data}$ before finding an approximate near neighbor is at most $|P_f(q)|$. So, suppose for an inductive hypothesis that the expected time complexity of $\CoreQuery(q, \bv_0)$ where $\bv_0 \leftarrow \CorePreprocess(P', k_0)$ is 
\[ O(k_0 \cdot (Q_{\sfh}(n) +1) + |P' \cap P_f(q)|  \cdot p_2^{k_0}). \]
We now upper bound the expected time of $\CoreQuery(q, \bv)$ where $\bv \leftarrow \CorePreprocess(P, k_0 + 1)$.
\begin{itemize}
\item First, we note that the call to $\CorePreprocess(P, \bv)$ had sampled $\bp \sim P$ and stored it in $v.\textit{point}$. If the sample satisfied $\bp \in P \setminus P_f(q)$, then $d_X(\bp, q) \leq cr$ and we can return $\bp$.
\item Otherwise, we let $\bh \colon X \to U$ denote the hash function stored in $\bv.h$, which is drawn from a $(r, cr, p_1,p_2)$-sensitive family $\calD$ for the uniform distribution over $P$. The time contains an additive term of at most $O(Q_{\sfh}(n))$ for computing $\bh(q)$.
\item Then, we execute $\CoreQuery(q, \bv_0)$ where $\bv_0 \leftarrow \CorePreprocess(P_{\bh(q)}, k_0)$. By the inductive hypothesis, the expected running time of $\CoreQuery(q, \bv_0)$ is at most
\[ O\left(k_0 \cdot (Q_{\sfh}(n)+1) +|P_{\bh(q)}\cap P_f(q)|\cdot p_2^{k_0} \right).\]
\end{itemize}
Therefore, the total expected time complexity becomes at most
\begin{align*}
&\Prx_{\bp \sim P}\left[ \bp \in P_{f}(q) \right] \cdot O\left( (k_0 + 1) \cdot Q_{\sfh}(n) + \Ex_{\bh \sim \calD}\left[ |P_{\bh(q)} \cap P_f(q)| \right] p_2^{k_0} \right) + O(1)  \\
&\qquad\qquad \leq \Prx_{\bp \sim P}\left[ \bp \in P_f(q) \right] \cdot O\left( (k_0 + 1) \cdot Q_{\sfh}(n) + |P| \cdot p_2^{k_0+1} \right) + O(1),
\end{align*}
where we used
\begin{align*} 
\Ex_{\bh \sim \calD}\left[ |P_{\bh(q)} \cap P_f(q)| \right] &= |P| \cdot \Prx_{\substack{\bh \sim \calD \\ \bp \sim P}}\left[ \begin{array}{c} d_X(q, \bp) > cr \\ \bh(q) = \bh(\bp) \end{array} \right]  \leq |P| \cdot p_2.
\end{align*}
This concludes the inductive hypothesis, since the probability that $\bp \in P_f(q)$ is exactly $|P_f(q)| / |P|$.
\end{proof}

\begin{proof}[Proof of Theorem~\ref{thm:hashing-to-nn}]
We let $k = \lceil \log_{1/p_2} n \rceil$ and instantiate $\ell = O(n^{\rho}/p_1)$ independent executions of $\CorePreprocess(P, k)$. By Claim~\ref{cl:success-prob}, the probability that any single data structure succeeds is at least $p_1^k$, so that the probability that all the data structures fail is at most
\[ \left( 1 - p_1^k\right)^{\ell} \leq \exp\left( - O(n^{\rho} / p_1) \cdot p_1^k \right) = \exp\left( - O\left( \frac{1}{p_1^{1 + \log_{1/p_2} n}}\right) \cdot p_1^{\lceil \log_{1/p_2} n \rceil}\right) \leq 0.1. \]
The preprocessing time follows from the setting of $k, \ell$ and Claim~\ref{cl:preprocess-time}. For the query time, Claim~\ref{cl:query-time} implies that the expected running time is at most
\[ \ell \cdot O\left(k \cdot (Q_{\sfh}(n)+1) + |P_f(q)| \cdot p_2^k \right) \leq \ell \cdot  O\left( \log_{1/p_2} n \cdot Q_{\sfh}(n) + 1 \right), \]
which concludes the theorem.
\end{proof}

%% file: extension-hypercube.tex
%!TEX root = main.tex

\section{Extension of Dynamic Data-Dependent Trees from the Hamming Cube to $\ell_1$} \label{sec:dynamic-l1}

First, for any $p \in (1, 2]$, there exists an embedding of $\ell_p^d$ to $\ell_1^{d'}$ which is implemented by a linear map and perturbs distances by $(1+\eps)$, where $d' = O(d\log(1/\eps)/\eps^2)$~\cite{johnson1982embedding}. Using this embedding increases the running time by an additive factor of $O(dd') = \poly(d)$, but all points are in $\ell_1^d$ and the aspect ratio $\Phi$ changes by at most a $(1+\eps)$-factor. By 
 re-scaling and discretizing by the aspect ratio $\Phi$, we may further consider inputs which lie in $([\Delta]^d, \ell_1)$ (where $\Delta$ is $O(\Phi)$). Observe that there is a simple isometric embedding $u:([\Delta]^d,\ell_1) \to \{0,1\}^{\Delta d}$ given by the unary encoding of each coordinate:
\[ u(x)_{d \cdot (i-1)+ j } = \mathbf{1}\left(x_i \geq j \right) \]
for any $i \in [d],j \in [\Delta]$. Thus, Theorem \ref{thm:dynamic-main} would follows, except, the unary embedding requires $O(\Delta d)$ running time, which is potentially exponential in the bit-representation of $x \in [\Delta]^d$. In this section, we show that, despite explicitly computing the unary embedding is too costly, the composition of the unary embedding $u$ from $[\Delta]^d \to \{0,1\}^{d \Delta}$ and the tree embedding $\{0,1\}^{d \Delta} \to \bT$ of Theorem \ref{thm:dynamic-main} can be maintained without fully-forming the intermediate unary embedding $u(x)$. In particular, we show that the result of applying the dynamic tree embedding to the unary encoding can be realized in only $\tilde{O}(d)$ time. We will first need the following. 

%\cite{farach2015exact}
\begin{lemma}[In Theorem 5 of \cite{bringmann2014internal} for the case of $p=1/2$, and Theorem 2 of \cite{farach2015exact} for reduction to general $q$]\label{lem:fastbinom}
    Fix any $q \in [0,1]$, $n \geq 1$, and constant $c>0$. There is an algorithm that samples $\bX \sim \emph{\texttt{Binomial}}(n,q)$ in expected $O(1)$ time in the WordRAM model with $O(\log n)$-bit words, and in time $\polylog(n)$ with probability $1-n^{-c}$. 
\end{lemma}

%We remark that, if attempting to call the below data structure on more than $\Delta/100$ queries, one can do this by explicitly sampling and representing the hash functions $\calH_t$, which will take a total initialization overhead of $O(\Delta d)$, which can be amortized over all the queries. 
We are now ready to state our reduction. 
\begin{lemma}\label{lem:dynamic-l1-to-ham}
There is a data-structure in the WordRAM model with $O(\log \Delta)$-bit words, that initializes in expected time $O(d \log (d\Delta))$, and supports the following: 
    \begin{itemize}
        \item \emph{\textbf{Maintenance}}: For the hash family $\calH_t$ in Equation \ref{eqn:proj-hash-def} and any $d,\Delta \geq 1$, the data structure maintains  draws of $\bphi_\ell \sim \mathcal{H}_{2^\ell, d \Delta}$ for $\ell=1,2,\dots,L$, where $L = O(\log_2 (d\Delta))$. 
        \item \emph{\textbf{Query}}$(x)$: given a point $x \in [\Delta]^d$, the data structure computes the value of $\bphi_\ell(u(x))$ for all $\ell$ in expected time $O(d \log (d\Delta))$. 
    \end{itemize}
\end{lemma}

\begin{proof}
    The data structure employs the principle of deferred decisions to avoid generating all random bits required to specify the hash functions $\bphi_\ell$. Instead, we condition on portions of this randomness as they become required to compute the values $\bphi_\ell(u(x))$ in a manner consistent with %all values $\phi_\ell(u(y))$ for previous 
    prior queries to points $y \in [\Delta]^d$.

    Fix any $\ell \in [L]$ and any $i \geq 0$.    Note that if we order the coordinate samples $i_1\dots,i_{2^\ell} \sim [d\Delta]$ used in the construction of the hash function $\phi_\ell$ so that $i_1 < i_2 < \dots < i_{2^\ell}$, then for a point $x \in [\Delta]^d$, to implicitly compute and represent the value of $\phi_\ell(u(x))$ it suffices to determine, for each $\tau \in [d]$, the number of indices $j \in [2^\ell]$ such that $ i_j \in [d(\tau-1) , d(\tau-1)  + x_\tau ]$ and the number of such indices such that $ i_j \in (d(\tau-1)  + x_\tau , d\tau )$. Let $B_\tau = [d(\tau -1), d\tau)$ denote the block of coordinates  of the hypercube corresponding to the $\tau$-coordinate in $[\Delta]^d$.    
    
    More generally, for a set of points $\Omega \subset [\Delta]^d$ and each $\tau \in [d]$, let $\omega_{1,\tau},\dots,\omega_{R,\tau} \in B_\tau$, where $\omega_{1,\tau}< \omega_{2,\tau} < \dots < \omega_{r,\tau}$, be the set of indicies appearing in the set $\{d(\tau-1) + x_\tau\}_{x \in \Omega}$. So long as we know the number of samples from $i_1\dots,i_{2^\ell}$ that appear in each interval $[\omega_{i,\tau},  \omega_{i+1,\tau})$, this is sufficient to compute the values of  $\phi_\ell(u(x))$ for all $x \in \Omega$. Thus, the goal of the data structure will be to maintain the number of samples $i_1\dots,i_{2^\ell}$  which appear between any two consecutive values $\{d(\tau-1) + x_\tau\}_{x \in \Omega}$, for each $\tau \in [d]$.

    In pre-processing, we can draw from the distribution on $\Z^d$ which specifies how many samples $i_j$  land in each block $B_\tau$. This can be done in $\tilde{O}(d)$ time by sampling from the Binomial distribution $s_1 \sim \texttt{Binomial}(2^\ell, \frac{|B_1|}{d\Delta})$ which specifies the number of samples $s_1$ in $B_1$, conditioning on it, and then sampling  $s_2 \sim \texttt{Binomial}(2^\ell - s_1, \frac{|B_2|}{d\Delta})$  to specifies the number of samples $s_2$ in $B_2$, and so on. By Lemma \ref{lem:fastbinom}, this can be done in expected constant time. 
    
    We now show how to compute a new value of $\phi_\ell(u(x))$ given that we have already compute the values of $\phi_\ell(u(x_i))$ for $x_1,\dots,x_i \in [\Delta]^d$. By adding $x$ to $\Omega$, this adds at most $d$ new values to the set $\{d(\tau-1) + x_\tau\}_{\substack{x \in \Omega \\ \tau \in [d]}}$. For each such value, this adds a new index within the interval $[\omega_{i,\tau},  \omega_{i+1,\tau})$ between two previously consecutive values in $\{d(\tau-1) + x_\tau\}_{\substack{x \in \Omega \\ \tau \in [d]}}$. This splits the interval the interval $[\omega_{i,\tau},  \omega_{i+1,\tau})$  into two parts, call them $I_1,I_2$.  Since by induction we will have already computed the number of indices $i_j$ that land in this interval, we simply sample from the correct Binomial distribution that determines how many of those indices will land in $I_1$ and how many land in $I_2$, which can be done in constant time by  Lemma \ref{lem:fastbinom}. Repeating this for all $d$ coordinates and $O(\log d\Delta)$ values of $\ell$ completes the proof. 
\end{proof}

\ignore{
\begin{lemma}
   $S \subset \EMD_s([\Delta]^d,\ell_1)$ of $|S| = n$ points. There is a data-dependent embedding  $\varphi: S \to (\R^m, \ell_1)$, for $m = \poly(n,s,\log d)$, such that $\varphi(x)$ is $\tilde{O}(s d^2)$-sparse for all $x \in S$, and such that for all $i \in [n]$ and all $x,y \in S$, with probability $1-1/\poly(n)$ we have
    \begin{equation}\label{eqn:embeddingbound}
    \EMD(x,y)   \leq  \|\varphi(x) - \varphi(y)\|_1  \leq  \tilde{O}(\log \Delta n) \cdot \EMD(x,y)    
    \end{equation}
    Moreover, there is a dynamic data structure that can compute $\varphi(x)$ for all $x \in S$ in time  $\tilde{O}(n s d)$, and can add a new point $x$ to the embedding in time $\tilde{O}(sd)$.
\end{lemma}
\begin{proof}
    \Raj{TODO: this proof}. In preprocessing, look at each coordinate $i \in [d]$, and for each hash function look at all the coordinates $i_1,\dots,i_t$ that occur at the $i$-th coordinate of all points $x \in S$. Sample from the distribution that determines how many samples were made in each of the intervals $[0,i_j]$ for $j \in [t]$. This allows us to compute the hash value for each point. With the new point, say its $i$-th coordinate is $b$ with $i_k \leq b \leq i_{k+1}$. Then we know how many sampled coordinates should be in this interval, so we can sample from the distribution to determine how many are in $[i_k,b]$ and how many are in $(b,i_{k+1})$, and thus compute the value of the hash of the new point. 
\end{proof}
}